\documentclass[sigconf,screen,authorversion]{acmart}

\usepackage{amsmath,amsfonts}

\usepackage{amssymb} 
\usepackage{ifxetex}
\usepackage[T1]{fontenc}
\usepackage[utf8]{inputenc}
\usepackage[cal=boondox]{mathalfa}
\usepackage{stmaryrd}
\usepackage{url}
\usepackage{tikz-cd}
\usepackage{amsthm}
\usepackage{mathpartir}

\newenvironment{proofof}[1]
{\begin{proof}[Proof of {#1}]}
{\end{proof}}
\newenvironment{proofsketch}
{\begin{proof}[Proof (sketch)]}
{\end{proof}}
\newenvironment{construction}
{\begin{proof}[Construction]}
{\end{proof}}

\renewcommand{\bar}[1]{\overline{#1}}

\newcommand\sym[1]{\mathsf{#1}}

\newcommand{\cat}[1]{\mathbb{#1}}

\newcommand{\id}{\mathsf{id}}

\newcommand{\slice}[2]{#1/{#2}}

\newcommand{\inv}[1]{#1^{-1}}

\renewcommand{\L}{\mathsf{L}}
\newcommand{\LCof}{\mathsf{L}_{\sF}}
\newcommand{\R}{\mathsf{R}}
\newcommand{\RFam}{\mathsf{R_{Fam}}}
\newcommand{\REl}{\mathsf{R_{El}}}

\newcommand{\tick}{\mathsf{tick}}

\newcommand{\opcat}[1]{{{#1}^{\mathrm{op}}}}

\newcommand{\timeobj}[2]{(#1;#2)}

\newcommand{\PSh}[1]{\mathsf{PSh}(#1)}

\newcommand{\pshfunc}[2]{#2 \cdot #1}

\newcommand{\FSA}{\mathcal{E}}
\newcommand{\FSB}{\mathcal{E'}}
\newcommand{\FSC}{\mathcal{E''}}

\renewcommand{\vartheta}{\delta}

\newcommand{\Set}{\mathsf{Set}}

\newcommand{\Fam}{\mathsf{Fam}}

\newcommand{\trm}{\mathsf{El}}

\newcommand{\p}[0]{\mathsf{p}}
\newcommand{\q}[0]{\mathsf{q}}
\newcommand{\compr}[2]{#1.#2}

\newcommand{\clkel}{\chi}
\newcommand{\clkset}{\chi}

\newcommand{\pair}[2]{\left(#1,#2\right)}
\newcommand{\Univ}[0]{\mathsf{U}}
\newcommand{\El}[1]{\mathsf{El}\,#1}

\newcommand{\Bool}[0]{\mathsf{Bool}}
\newcommand{\Nat}[0]{\nats}

\newcommand{\refl}[1]{\mathsf{refl}\,#1}

\newcommand{\dprod}[3]{\ensuremath{{\textstyle\prod\left(#1 : #2\right) . #3}}}

\newcommand{\subst}[2]{[#2/#1]}

\newcommand{\cuba}[1]{{\mathsf{TL}(#1)}}

\newcommand{\force}{\mathsf{force}}

\newcommand{\inl}{\mathsf{inl}}
\newcommand{\inr}{\mathsf{inr}}

\newcommand{\forallcode}{\ensuremath{\overline{\forall}}}

\newcommand{\latbindcode}[3]{\overline{\triangleright}\, (#1 \!:\!#2) . #3}

\newcommand{\timescode}{\overline{\times}}
\newcommand{\gStrcode}{\overline\gStr}
\newcommand{\natscode}{\overline\nats}

\newcommand{\fc}[1]{\mathsf{fc}(#1)}

\newcommand{\later}{\triangleright}
\DeclareMathOperator{\searlier}{\blacktriangleleft}
\DeclareMathOperator{\slater}{\blacktriangleright}

\newcommand{\clk}{\mathrm{Clk}}

\newcommand{\dfix}{\mathsf{dfix}}
\newcommand{\fix}{\mathsf{fix}}

\newcommand{\den}[1]{{\llbracket #1 \rrbracket}}

\newcommand\tabs[2]{\lambda (#1\! :\! #2).}
\newcommand\tickc{\diamond}
\newcommand\tappg[3]{#1\,_{#2}[#3] }
\newcommand\btick[2]{({#1},{#2})}
\newcommand\tapp[2][\tickA]{#2\,[#1] }

\newcommand{\tickA}{\alpha}
\newcommand{\tickB}{\beta}
\newcommand\latbind[2]{{\triangleright}\, (#1 \!: \!#2) .}
\newcommand\toksubst[3][\kappa]{\left[#2/#3\right]}

\newcommand{\tickcsub}[4]{[(#3 :\!  #4)/(#1 \! :\! #2)]}

\newcommand{\clocktype}{\mathsf{clock}}
\newcommand{\clockcnst}{\kappa_0}

\newcommand{\capp}[2][\kappa]{#2\,[#1]}
\newcommand{\cappp}[1]{#1[\bar\kappa]}

\newcommand\gLift[1][\kappa]{\sym{L}^{#1}}
\newcommand\Lift{\sym{L}}

\newcommand\Str{\sym{Str}}
\newcommand\gStr[1][\kappa]{\sym{Str}^{#1}}

\newcommand{\Pfin}{{\mathsf{P}_{\mathsf{f}}}}

\newcommand{\GLTS}[1][]{{\mathsf{LTS}^{#1}}}
\newcommand{\unfoldG}[1][]{{\mathsf{ufld}^{#1}}}

\newcommand{\Simforall}[1]{\mathsf{Sim}_{\forall}(#1)}
\newcommand{\Sim}{\mathsf{Sim}^{\kappa}}
\newcommand{\Simfix}{\mathsf{Sim}^{\kappa}_{\forall}}
\newcommand{\Bisim}{\mathsf{Bisim}^{\kappa}}
\newcommand{\Bisimforall}[1]{\mathsf{Bisim}_{\forall}(#1)}

\newcommand{\basicsub}[2]{[#1\mapsto #2]}
\newcommand{\subex}[3]{#1\basicsub{#2}{#3}}

\newcommand{\cirr}[1][\kappa]{\sym{cirr}^{#1}}
\newcommand{\tirrAx}[1][\kappa]{\sym{tirr}^{#1}}
\newcommand{\pfix}[1][\kappa]{\sym{pfix}^{#1}}

\newcommand{\tirr}{\sym{tirr}}

\newcommand{\nats}{\mathbb{N}}

\newcommand\hastype[4][]{
#2 \vdash_{#1} #3: #4
}
\newcommand\hasnotype[4][]{
#2 \vdash_{#1} #3
}
\newcommand\wfcxt[2][]{#2 \vdash_{#1}}
\newcommand\istype[3][]{
\ensuremath{#2 \vdash_{#1} #3 \, \operatorname{type}}
}

\newcommand\genericjudg[3][]{
#2 \vdash_{#1} #3
}

\newcommand{\istel}[2]{
\ensuremath{#1 \vdash #2 \, \operatorname{tel}}
}
\newcommand{\isconstrs}[2]{
\ensuremath{#1 \vdash #2 \, \operatorname{constrs}}
}
\newcommand{\isconstr}[3]{
\ensuremath{#1;#2 \vdash #3 \, \operatorname{constr}}
}
\newcommand{\iselimlist}[5]{
\ensuremath{{#1} \vdash {#2} : {#3} \rightharpoonup^{#4} {#5}}
}

\newcommand{\isclockelimlist}[5]{
\ensuremath{{#1} \vdash {#2} : {#3} \rightharpoonup^{#4} {#5}}
}

\newcommand{\isbox}[1]{
\ensuremath{\vdash #1 ~ _\square}
}

\newcommand{\istick}[4]{{#1} \vdash {#2} : {#3} \leadsto {#4}}
\newcommand{\isbtick}[4]{{#1} \vdash \btick {#2} {#3} \leadsto {#4}}

\newcommand{\issub}[3]{{#1} \vdash {#2} : {#3}}

\newcommand{\subcxt}{\sqsubseteq}

\newcommand{\Bndr}{\textsf{Bndr}}

\newcommand{\peq}{=}
\newcommand{\jeq}{\equiv}

\newcommand{\beq}{\equiv_b}

\newcommand{\defeq}{\mathbin{\overset{\textsf{def}}{=}}}

\newcommand{\clott}{CloTT}
\newcommand{\ctt}{CTT}
\newcommand{\cctt}{CCTT}

\newcommand{\I}{\mathbb{I}}
\newcommand{\F}{\mathbb{F}}
\renewcommand{\phi}{\varphi}
\newcommand{\Path}[3]{\mathsf{Path}_{#1}(#2, #3)}
\newcommand{\equi}{\simeq}
\newcommand{\comp}{\mathsf{comp}}

\newcommand{\HIT}[2]{\textsf{#1}\,{#2}}
\newcommand{\HITcode}[2]{\overline{\textsf{#1}}\,{#2}}
\newcommand{\con}{\mathsf{con}}
\newcommand{\hcomp}{\mathsf{hcomp}}
\newcommand{\hfillc}{\mathsf{hfill}}
\newcommand{\lemmaD}[1]{\mathsf{lemD}_{#1}}
\newcommand{\lemmaDg}[2]{\mathsf{lemDg}_{#1}(#2)}
\newcommand{\trans}{\mathsf{trans}}
\newcommand{\ctrans}{\mathsf{ctrans}}

\newcommand{\elim}[4]{(#1)-\mathsf{elim}_{#2}(#3, #4)}

\newcommand{\evalbsynclock}[4]{\llparenthesis{#4}\rrparenthesis^{#1}_{#2,#3}}
\newcommand{\evalbsynclockg}[5]{\llparenthesis{#5}\rrparenthesis^{#1}_{#2,#3,#4}}
\newcommand{\out}{\mathsf{out}}

\newcommand{\evalbsem}[1]{\lVert #1 \rVert}

\newcommand{\trunc}[1]{\lVert {#1} \rVert}
\newcommand{\htrunc}[2]{\lVert {#1} \rVert_{#2}}
\newcommand{\sphere}[1]{\mathbb S^{#1}}

\newcommand{\ev}[1]{\mathsf{ev}_{#1}}

\newcommand{\const}{\mathsf{const}}

\newcommand{\sI}{\I}
\newcommand{\sF}{\F}

\newcommand{\fibCwF}{\mathsf{Fib}}

\newcommand{\CubeCat}{\mathcal C}
\newcommand{\dm}[1]{\mathsf{dM}(#1)}

\newcommand{\TimeCat}{\mathcal T}

\newcommand{\psearlier}{\p_{\searlier}}

\newtheorem{remark}{Remark}
\newtheorem{operation}{Operation}

\begin{document}

\title[Greatest HITs]{Greatest HITs: Higher inductive types in coinductive definitions via induction under clocks}

\author{Magnus Baunsgaard Kristensen}
 \affiliation{
   \department{Department of Computer Science}              
   \institution{IT University of Copenhagen}            
   \country{Denmark}                    
 }
%
\author{Rasmus Ejlers M{\o}gelberg}
 \affiliation{
   \department{Department of Computer Science}              
   \institution{IT University of Copenhagen}            
   \country{Denmark}                    
 }
\author{Andrea Vezzosi}
 \affiliation{
   \department{Department of Computer Science}              
   \institution{IT University of Copenhagen}            
   \country{Denmark}                    
 }

\copyrightyear{2022} 
\acmYear{2022} 
\setcopyright{acmlicensed}
\acmConference[LICS '22]{37th Annual ACM/IEEE Symposium on Logic in Computer Science (LICS)}{August 2--5, 2022}{Haifa, Israel}
\acmBooktitle{37th Annual ACM/IEEE Symposium on Logic in Computer Science (LICS) (LICS '22), August 2--5, 2022, Haifa, Israel}

%
\begin{abstract}
We present Clocked Cubical Type Theory, the first type theory combining
multi-clocked guarded recursion with the features of Cubical Type Theory.
Guarded recursion is an abstract form of step-indexing, which can be used
for construction of advanced programming language models. In its multi-clocked
version, it can also be used for coinductive programming and reasoning, encoding
productivity in types. Combining this with Higher Inductive Types
(HITs) the encoding extends to coinductive types that are traditionally hard to represent
in type theory, such as the type of finitely branching labelled transition
systems. 

Among our technical contributions is a new principle of \emph{induction under
clocks}, providing computational content to one of the main axioms required
for encoding coinductive types. This principle is verified using a denotational
semantics in a presheaf model.
\end{abstract}

\maketitle

\section{Introduction}
\label{sec:intro}

Homotopy type theory~\cite{hottbook} is the extension of
Martin-L{\"o}f type theory~\cite{MartinLof:84} with
the univalence axiom and higher inductive types. It can be seen as a
foundation for mathematics, allowing for synthetic approaches to
mathematics, as exemplified by synthetic homotopy theory, as well
as more direct formalisations of mathematics, closer to the traditional set-theory
based developments. Higher Inductive Types (HITs) play a key role in both approaches:
In synthetic homotopy theory these are used to describe spaces
such as the spheres and the torus, and in the more direct formalisations,
HITs express free structures and quotients, both of which have traditionally
been hard to represent and use in type theory.

The kind of mathematics that is particularly important to the LICS community
often involves reasoning about recursion, and in particular programming with
and reasoning about coinductive types. Doing this in type theory has historically
been difficult for a number of reasons. One is that the definitions of the functors
of interest in coalgebra often involve constructions such as finite or
countable powerset (for modelling non-determinism),
or quotients. All 
of these are known to be representable as HITs \cite{FGGW18,quotientingDelay,hottbook}.
Another reason is that programming and reasoning about coinductive types
involves productivity checking, which in most proof assistants is implemented
as non-modular syntactic checks that can lead to considerable overhead
for programmers \cite{NAD:beat}. A third reason is that bisimilarity, which is the
natural notion of identity for coinductive types, rarely (at least provably)
coincides with the built-in identity types of type theory, and so proofs
cannot be transported along bisimilarity proofs.

Multi-clocked guarded recursion~\cite{atkey13icfp} allows for the productivity
requirement for coinductive definitions to be encoded in types. The key ingredient
in this is a modal operator $\later$ (pronounced `later') indexed by clocks $\kappa$,
encoding a notion of time-steps in types. Guarded recursive types such as
$\gStr \equi \nats \times \later^\kappa \gStr$  are recursive types in which the
recursion variable occurs only under a $\later$.
In the case of $\gStr$ the type equivalence above
states that the head of a guarded stream is available now, but the tail only after
a time step on clock $\kappa$. Guarded streams can be defined using a fixed
point operator $\fix^\kappa$ of type $(\later^\kappa A \to A) \to A$. In the case of
$\gStr$ the delay in the domain of the input precisely captures the productivity
requirement for programming with streams. Guarded recursive types can themselves
be encoded as fixed points on a universe. 
Using 
quantification over clocks, the coinductive type of streams can be encoded as
$\Str \defeq \forall\kappa . \gStr$.

Guarded recursion \cite{Nakano:Modality} can also be used as an abstract
approach to step-indexing techniques~\cite{Appel:M01,AppelMRV07}.
In particular,
the type variables in guarded recursive types can also appear in negative positions,
and so can be used to provide models of untyped lambda calculus or solve the
advanced type equations needed for modelling higher-order store~\cite{ToT}
and advanced notions of separation logic~\cite{jung2018iris}.
%
A type theory combining guarded recursion with homotopy type theory will
therefore provide a powerful framework, not just for coinductive reasoning
but also for reasoning about recursion and advanced programming language
features beyond the reach of traditional domain theory.

\subsection{Clocked Cubical Type Theory}
\label{sec:contributions}

This paper presents Clocked Cubical Type Theory (\cctt), the first type theory
to combine all the above mentioned features: Multi-clocked guarded recursion,
HITs and univalence. Rather than basing this on the original
formulation of homotopy type theory we build on Cubical Type
Theory (CTT)~\cite{CTT}, a variant of HoTT based on the cubical model of
type theory, and implemented in the Cubical Agda
proof assistant~\cite{CubicalAgda}. One reason for this choice is that
the operational properties of CTT as well as the concrete implementation
give hope for an implementation of our type theory.
In fact we have already extended the implementation of Guarded Cubical
Agda~\cite{GuardedCubicalAgda} to support most of the new constructions of \cctt.
Another is that the path types of
CTT make for a simple
implementation of the extensionality principle for
the $\later^\kappa$ modality~\cite{GCTT}.

\cctt\ extends \ctt\ with the constructions
of Clocked Type Theory (\clott)~\cite{bahr2017clocks}, a type theory for multi-clocked
guarded recursion. \clott\ uses a Fitch-style approach to programming with the $\later$
modality. This means that elements of modal type are introduced by abstracting
over tick variables $\tickA : \kappa$ and eliminated by application to ticks. One benefit of this
is that let-expressions are avoided, and these do not interact well with dependent types in general.
For example, one can define a dependently typed generalisation of applicative action
\cite{mcbride2008applicative} operator of type
\begin{align} \label{eq:appl:action}
 \later^\kappa({\dprod xAB}) \to \dprod y{\later^\kappa A}{\latbind\tickA \kappa B\subst x{\tapp y}}
\end{align}
In this term the tick variable $\tickA$ should be thought of as evidence that time has passed, which
is used in subterms like $\tapp y$ to access the element of type $A$ delivered by $y : \later^\kappa A$ in
the next time step. CloTT has a single tick constant $\tickc$ that can be used for a controlled
elimination of $\later$. \citet{bahr2017clocks} define a strongly normalising operational semantics for
CloTT, in which $\tickc$ unfolds fixed point operations. Since tick variables $\tickA$ can be substituted
by $\tickc$, their identity is crucial for normalisation. For reasoning in the type theory, however, it is often
necessary that $\tapp y$ does not depend on $\tickA$ up to path equality. This is referred to as the
\emph{tick irrelevance axiom}.

When adding such axioms to a univalent type theory one should be take care to avoid
introducing extra structure that one will eventually need to reason about.
Ideally, an axiom like tick irrelevance
should state that the type of ticks is propositional, i.e., that any two elements can be identified, any two
paths between elements can be identified, and so on for higher dimensions.
In \cctt, there appears to be no such way of formulating the axiom,
because ticks do not form a type. This paper shows how to obtain tick irrelevance
by enriching the language of ticks
with paths $\tirr \,u\, v$ between any pair of ticks $u, v$. This has the additional benefit that it allows for
rules giving computational content to the axiom.

\subsection{Induction under clocks}

The encoding of coinductive streams from guarded streams mentioned above generalises to an
encoding of coinductive solutions to type equations of the form $X \equi F(X)$ for all functors $F$
commuting with clock quantification in the sense that the canonical map
\begin{equation} \label{eq:comm:forall}
 F(\forall\kappa . X) \to \forall\kappa. F(X)
\end{equation}
is an equivalence of types.
Note that this is a semantic condition, rather than the more standard grammars for
coinductive types. For this to be useful, one must have a large supply of such functors,
including functors formed using inductive and higher inductive types.
This property has previously been obtained via axioms~\cite{atkey13icfp}.
Here we formulate a principle of
\emph{induction under clocks}, which gives computational content to these
axioms also in the case of higher inductive types.

Consider for example the finite powerset $\Pfin$ which
can be defined as a HIT given by constructors for empty set, singletons, union,
equalities stating associativity, commutativity and idempotency of union as well as an
axiom forcing $\Pfin(X)$ to be a hset~\citep{FGGW18}. In this case induction under clocks
states that to inhabit a family
over an element of type $\forall \kappa. \Pfin{(X)}$
it suffices to inductively describe the cases for elements of the form
$\lambda \kappa. \con_i(p)$ for each constructor $\con_i$ for $\Pfin$, i.e., for
$\lambda\kappa.\emptyset, \lambda\kappa.\{a\}, \lambda\kappa. \capp x \cup \capp y$ 
and for equality
constructors such as $\lambda\kappa. \mathsf{idem}(\capp x,r)$, where
$\mathsf{idem}$ implements idempotency of union. In these cases, $x,y$ are
assumed to be of type $\forall\kappa . \Pfin(X)$. Using this, one can construct 
an inverse to (\ref{eq:comm:forall}) for $F =\Pfin$, and also prove that both
compositions are identities. As a consequence, the coinductive solution to
$\GLTS[] \equi \Pfin(A \times\GLTS)$ can be encoded as
$\GLTS \defeq \forall\kappa . \GLTS[\kappa]$
where $\GLTS[\kappa] \equi \Pfin(A \times\later^\kappa\GLTS[\kappa])$
is a guarded recursive type. $\GLTS[]$ is the type of $A$-labelled finitely
branching transition systems. We moreover prove that path equality
for this type coincides with bisimilarity.

The results for $\GLTS[]$ mentioned above hold for all types of labels $A$ that are \emph{clock-irrelevant},
meaning that the map $A \to \forall\kappa . A$ is an equivalence for $\kappa$ fresh.
In previous type theories for multi-clocked guarded recursion clock-irrelevance of all types
has been taken as an axiom. Most types are clock-irrelevant, but ensuring this for universes
requires special care both syntactically and semantically~\citep{GDTTmodel}. In this paper
we opt for a simpler solution, taking clock-irrelevance to be a side condition to the
correctness of the encoding of coinductive types. Using the principle of induction
under clocks we show that HITs constructed using clock-irrelevant types are themselves
clock-irrelevant, an important step towards reintroducing clock irrelevance as an axiom.

\subsection{Denotational semantics}

Our type theory is justified by a denotational semantics in the form
of a presheaf model. This model combines previous models of
\ctt~\cite{CTT,OPAx} and \clott~\cite{clottmodel}. A key challenge for
the construction of this model is the definition of a composition
structure on the operation modelling $\later$. We give a general theorem
for defining composition structures for dependent
right adjoints~\cite{drat} which can be used to model other type theories
combining Fitch-style modal operators~\cite{clouston2018fitch} with Cubical Type Theory.

\paragraph{Overview of paper}

The paper is organised as follows. We start by recalling
Cubical Type Theory in \autoref{sec:CTT}, and \autoref{sec:CCTT} extends this
to \cctt.
The encoding of coinductive types using guarded
recursion is recalled in \autoref{sec:coinductive:types}
which also details the example of labelled transition systems.
The syntax for HITs in \cctt\ is presented in \autoref{sec:hits}, which also states
the principle of induction under clocks.
The denotational semantics of Clocked
Cubical Type Theory is sketched in \autoref{sec:model}.
Finally, we discuss related work in \autoref{sec:related}, and
conclude and discuss future work in \autoref{sec:conclusion}.

\section{Cubical Type Theory}
\label{sec:CTT}

Cubical Type Theory \cite{CTT} is an extension of
Martin-L\"{o}f type theory which gives a computational interpretation
to extensionality principles like function extensionality and
univalence. It does so by internalizing the notion from Homotopy Type
Theory that equalities are paths, by explicitly representing them as
maps from an interval object.
The interval, $\I$, is not a type but contexts can contain $i : \I$ assumptions and there is a judgement, $\hastype \Gamma r \I$ which specifies that $r$ is built according to this grammar
\[
r,s ::= 0 \mid 1 \mid i \mid 1 - r \mid r \wedge s \mid r \vee s
\]
where $i$ is taken from the assumptions in $\Gamma$.  Equality of two
interval elements, $\hastype \Gamma {r \jeq s} \I$ corresponds to the
laws of De Morgan algebras.  A good intuition is to think of
$\I$ as the real interval $[0,1]$ with $r \wedge s$ and $r \vee s$
given by minimum and maximum operations.

Given $x,y : A$, we write $\Path A x y$ for the type of paths between $x$ and $y$, which correspond to maps from $\I$ into $A$ which are equal to $x$ or $y$ when applied to the endpoints $0,1$ of the interval $\I$.
\begin{mathpar}
  \inferrule*{\hastype {\Gamma, i : \I} t A}
  {\hastype \Gamma {\lambda i.\,t}
    {\Path A {t\subst i 0} {t\subst i 1}}}
  \and
  \inferrule*{\hastype {\Gamma} p {\Path A {a_0} {a_1}} \\ \hastype \Gamma r \I}
  {\hastype \Gamma {p\,r} A}
\end{mathpar}
We sometimes write $=_A$ or simply $=$  as infix notation for path equality, and use $\jeq$ for judgemental equality.
Path types support the $\beta$ and $\eta$ equalities familiar from function types, but we also have that a path applied to $0$ or $1$ reduces according to the endpoints specified in its type.
\begin{mathpar}
  (\lambda i.\, t)\,r \jeq t\subst i r
  \and
  p \jeq (\lambda i.\,p\,i)
  \and
  p\,0 \jeq a_0
  \and
  p\,1 \jeq a_1
\end{mathpar}
The path witnessing reflexivity is defined as $\mathsf{refl}\,x \defeq \lambda i.\, x : \Path A x x$, given any $x : A$.
It is also straightforward to provide a proof of function extensionality, as it is just a matter of reordering arguments:
\begin{mathpar}
\inferrule*{\hastype \Gamma p {\Pi (x : A).\, {\Path {B} {f\,x} {g\,x}}}}
{\hastype \Gamma {\lambda i.\,\lambda x.\,p\,x\,i} {\Path {\Pi (x : A). {B}} f g}}
\end{mathpar}
Other constructions like transitivity, or more generally transport along path, require a
further primitive operation called composition, which we now recall.

Define first the \emph{face lattice} $\F$, as the free
distributive lattice with generators $(i = 0)$ and $(i = 1)$, and
satisfying the equality $(i = 0) \wedge (i = 1) = 0_\F$. Explicitly, elements $\hastype \Gamma \phi \F$ are given by the following grammar, where $i$ ranges over the interval variables from $\Gamma$:
\[
\phi,\psi ::= 0_\F \mid 1_\F \mid (i = 0) \mid (i = 1) \mid \phi \wedge \psi \mid \phi \vee \psi
\]
Given $\hastype \Gamma \phi \F$ we can form the restricted context
$\Gamma,\phi$.  Types and terms in context $\Gamma,\phi$ are called
\emph{partial}, matching the intuition that they are only defined for
elements of $\Gamma$ that satisfy the constraints specified by $\phi$.
For example while a type $\istype {i : \I} A$ corresponds to a whole
line, a type $\istype {i : \I, (i = 0) \vee (i = 1)} B$ is merely two
disconnected points. Given a partial element $\hastype {\Gamma,\phi} u
A$ we write $\hastype \Gamma t {A[\phi \mapsto u]}$ to mean the
conjunction of $\hastype \Gamma t A$ and $\hastype {\Gamma,\phi} {t
  \jeq u} A$. Likewise we write $\istype \Gamma {A[\phi \mapsto B]}$ to mean
$\istype\Gamma A$ and $\istype{\Gamma, \phi}{A \jeq B}$. 

Composition is given by the following typing rule:
\[
\inferrule*{ \hastype \Gamma \phi \F \\ \istype {\Gamma, i : \I} A \\ \hastype {\Gamma,\phi,i : \I} u A \\ \hastype \Gamma {u_0} {A\subst i 0[\phi \mapsto u\subst i 0]} }
         {\hastype \Gamma {\comp^i_A\,[\phi \mapsto u]\,u_0} {A\subst i 1[\phi \mapsto u\subst i 1]} }
\]
the inputs $u$ and $u_0$ represent the sides and the base of an open
box, while composition gives us the lid, i.e. the side opposite to the
base. The behaviour of composition is specified by $u$ whenever $\phi$
is satisfied, and otherwise depends on the type $A$, consequently when
we introduce new type formers we will also give a corresponding
equation for $\comp$.
As an example we can prove transitivity for paths with the following term:
\begin{mathpar}
  \inferrule*{\hastype \Gamma p {\Path A x y} \\ \hastype \Gamma q {\Path A y z}}
  {\hastype \Gamma {\lambda i.\,\comp^j\,[ i=0 \mapsto x, i = 1 \mapsto q\,j]\,(p\,i)} {\Path A x z}}
\end{mathpar}
The standalone syntax for partial elements is as a list of faces and
terms $[\phi_1~u_1,\ldots,\phi_n~u_n]$, but we will write $[ \vee_i \phi_1 \mapsto [\phi_1~u_1,\ldots,\phi_n~u_n]]$
as $[\phi_1 \mapsto u_1,\ldots,\phi_n \mapsto u_n]$.
We refer to \citet{CTT} for more details on partial elements,
composition, and the \emph{glueing construction}.
The latter is what allows to derive the univalence principle as a
theorem, by turning a partial equivalence into a total one.
Writing $A \equi B$ for the type stating that $A$ and $B$ are
equivalent~\cite{hottbook}, univalence states that the canonical map of type $A
=_\Univ B \to A \equi B$ is itself an equivalence for any types $A, B : \Univ$.
We will use univalence in the rest of the paper, but we will not need
to refer to a specific implementation, so we omit the details in this
brief overview.

From Homotopy Type Theory~\cite{hottbook} we also recall the notion of homotopy
\emph{proposition} which are types for which any two elements are path
equal, and homotopy \emph{set} which are types whose path types are
propositions.
Given any type $A$ its propositional truncation $\trunc{A}_0$ is the
least proposition that admits a map from $A$.
It can be defined as an Higher Inductive Type (HIT), which are inductive
types defined by providing generators for both the type itself and its
path equality. Another example of a HIT is the finite powerset~\cite{FGGW18} of
a type $A$ mentioned in the introduction. 
%
We will discuss HITs in more detail in \autoref{sec:hits}.

\section{Clocked Cubical Type Theory}
\label{sec:CCTT}

\begin{figure*}[tbp]
\begin{center}
\textbf{Context and type formation rules}
\begin{mathpar}
\inferrule*
{\wfcxt{\Gamma} \and \kappa\notin \Gamma}
{\wfcxt{\Gamma, \kappa : \clocktype}}
\and
\inferrule*
{\hastype{\Gamma}\kappa\clocktype \and \tickA\notin \Gamma}
{\wfcxt{\Gamma, \tickA : \kappa}}
\and
  \inferrule*
  {\istype{\Gamma,\tickA:\kappa}{A}}
  {\istype{\Gamma}{\latbind{\tickA}{\kappa} A}}
  \and
  \inferrule*
  {\istype{\Gamma,\kappa : \clocktype}{A}}
  {\istype{\Gamma}{\forall \kappa . A}}
\end{mathpar}
\textbf{Typing rules}
\begin{mathpar}
  \inferrule*
  {\hastype{\Gamma,\kappa : \clocktype}{t}{A}}
  {\hastype{\Gamma}{\lambda\kappa. t}{\forall \kappa . A}}
  \and
  \inferrule*
  {\hastype{\Gamma}{t}{\forall \kappa . A}\\
    \hastype\Gamma{\kappa'}\clocktype}
  {\hastype{\Gamma}{t [\kappa']}{A \subst{\kappa}{\kappa'}}}
  \and
  \inferrule*
  {\hastype{\Gamma,\tickA:\kappa}{t}{A}}
  {\hastype{\Gamma}{\tabs{\tickA}{\kappa} t}{\latbind{\tickA}{\kappa} A}}
  \and
  \inferrule*{ \istick {\Gamma} {u} {\kappa} {\Gamma'} \\
\hastype{\Gamma'} {t} {\latbind\tickA\kappa A}}
{\hastype \Gamma {\tappg {t} {\Gamma'} {u}} {A\subst \tickA u}}
\and
\inferrule*{ \isbtick {\Gamma} {\kappa'} {u} {\Gamma'} \\
\hastype{\Gamma', \kappa : \clocktype} {t} {\later (\tickA : \kappa). A}}
{\hastype \Gamma {\tappg {(\kappa.t)} {\Gamma'} {\btick {\kappa'} u}} {A\tickcsub {\tickA}{\kappa} {u}{\kappa'} }}
 \and
  \inferrule*
  {\hastype{\Gamma}{t}{\later^\kappa A \to A}}
  {\hastype{\Gamma}{\dfix^\kappa\,t}{\later^\kappa A}}
  \and
  \inferrule*
{\hastype[]{\Gamma}{t}{\later^\kappa A \to A} }
{\hastype{\Gamma}{\pfix[\kappa] \, t}{\latbind{\tickA}{\kappa}{\Path{A}{\tapp{(\dfix^{\kappa} t)}}{t(\dfix^{\kappa} t)}}}}
  \and
  \inferrule*
  {\hastype{\Gamma, \tickA : \kappa}{A}{\Univ}}
  {\hastype{\Gamma}{\latbindcode\tickA\kappa A}{\Univ}}
  \and
  \inferrule*
  {\hastype{\Gamma, \kappa : \clocktype}{A}{\Univ}}
  {\hastype{\Gamma}{\forallcode\kappa . A}{\Univ}}
  \and
  \inferrule*
  {\hastype{\Gamma}{t}{A} \\ \genericjudg \Gamma {A \jeq B}}
  {\hastype{\Gamma}{t}B}
  \and
  \inferrule*
  {
  \kappa : \clocktype \in \Gamma}
  {\hastype{\Gamma}\kappa\clocktype}
  \and
  \inferrule*
  {
   x : A \in \Gamma}
  {\hastype{\Gamma}xA}
\end{mathpar}
\textbf{Simple ticks}
\begin{mathpar}
\inferrule*{\Gamma'' \subcxt \Gamma}
           {\istick{\Gamma, \tickA : \kappa, \Gamma'}
                   {\tickA} {\kappa} {\Gamma'',\cuba{\Gamma'}}}
\and
\inferrule*{
\istick {\Gamma} {u} {\kappa} {\Gamma'} \\
\istick {\Gamma} {v} {\kappa} {\Gamma'} \\
\hastype {\Gamma} {r} {\I}}
{\istick{\Gamma} {\tirr(u,v,r)} {\kappa} {\Gamma'}}
\end{mathpar}

\textbf{Forcing ticks}
\begin{mathpar}
\inferrule*{\Gamma' \subcxt \Gamma \\ \hastype {\Gamma'} {\kappa'} {\clocktype}}
{\isbtick {\Gamma} {\kappa'} {\diamond} {\Gamma'}}
\and

\inferrule*{\Gamma'' \subcxt \Gamma}
           {\isbtick{\Gamma, \tickA : \kappa', \Gamma'}
                   {\kappa'} {\tickA}  {\Gamma'',\cuba{\Gamma'}}}
\and
\inferrule*{
\isbtick {\Gamma} {\kappa'} {u} {\Gamma'} \\
\isbtick {\Gamma} {\kappa'} {v} {\Gamma'} \\
\hastype {\Gamma} {r} {\I}}
{\isbtick{\Gamma} {\kappa'} {\tirr(u,v,r)} {\Gamma'}}

\end{mathpar}

%

\textbf{Timeless assumptions}
\begin{align*}
 \cuba{\Gamma, x :A} & = \cuba{\Gamma} &
 \cuba{\Gamma, \tickA : \kappa} & = \cuba{\Gamma} &
 \cuba{\Gamma, i : \I} & = \cuba{\Gamma}, i : \I \\
 \cuba{\Gamma, \kappa : \clocktype} & = \cuba{\Gamma}, \kappa : \clocktype &
 \cuba{\Gamma, \phi} & = \cuba{\Gamma}, \phi &
 \cuba{\cdot} & = \cdot
\end{align*}
\caption{Selected typing rules for \cctt. The rules for ticks use the relation defined as
$\Gamma, \cuba{\Gamma'} \subcxt \Gamma, \Gamma'$}
\label{fig:clott:typing}
\end{center}
\end{figure*}

\begin{figure*}[tbp]
\begin{center}
\textbf{Judgemental equalities on terms}
  \begin{align*}
   \tappg t{\Gamma} u & \jeq \tappg t{\Gamma'} u &
   \tappg {(\kappa .t)}{\Gamma} {\btick \kappa u} & \jeq \tappg {(\kappa .t)}{\Gamma'} {\btick \kappa u} &
    (\lambda\kappa.t)[\kappa']&\jeq  t \subst\kappa{\kappa'} \\
    \lambda \kappa. (t [\kappa]) &\jeq t & 
   \tappg {(\tabs{\tickA}\kappa t)} {} {u} & \jeq t\toksubst{u}{\tickA} &
    \tabs{\tickA}{\kappa} (\tapp t) &\jeq t \\ 
   \tappg {(\kappa. \tabs{\tickA}\kappa t)} {} {\btick {\kappa'} u} & \jeq t \tickcsub {\tickA}{\kappa} {u}{\kappa'} &
    \tappg{(\kappa . t)}{}{\btick{\kappa'}u} & \jeq \tappg {(t \subst\kappa{\kappa'})}{}u &
   \tappg {(\kappa. \dfix^\kappa \,t)} {} {\btick {\kappa'} \tickc} & \jeq  t(\dfix^\kappa \, t)\subst{\kappa}{\kappa'} \\
    \El(\forallcode\kappa . A) & \jeq \forall \kappa . \El(A)   &
    \tappg {(\kappa. \pfix \,t)} {} {\btick {\kappa'} \tickc}\,r & \jeq
  t(\dfix^\kappa \, t)\subst{\kappa}{\kappa'}
    & \El(\latbindcode\tickA\kappa A) & \jeq \latbind\tickA\kappa{\El(A)}
  \end{align*}
\textbf{Judgemental equalities on ticks}
  \begin{align*}
   \tirr(u,v,0) & \jeq u & \tirr(\tickc,\tickc,r) & \jeq \tickc &
   \tirr(u,v,1) & \jeq v &
\end{align*}
\caption{Selected judgemental equality rules of \cctt. The three $\eta$ rules are subject
to the standard conditions of $\kappa$ and $\tickA$, respectively, not appearing in $t$.
As a consequence of the first two rules, the residual context is omitted for tick application in the rules below.
The 8th axiom listed assumes that $\tickc$ does not appear in $u$, so that $u$ can be considered a simple tick.}
\label{fig:clott:equalities}
\end{center}
\end{figure*}

\cctt\ extends \ctt\ with the constructions of
\clott~\cite{bahr2017clocks,clottmodel}, a type theory with
Nakano style guarded recursion, multiple clocks and ticks. This means that
the delay type operator is associated with clocks: $\later^\kappa A$ is the type of data
of type $A$ which is available in the next time step on clock $\kappa$. There are no
operations on clocks, they can only be assumed by placing assumptions such as
$\kappa : \clocktype$ in the context, and abstracted to form elements of type
$\forall\kappa . A$. The status of $\clocktype$ is similar to $\I$ in the sense
that it is not a type, and so, e.g., $\clocktype \to \clocktype$ is not a
wellformed type. It is sometimes convenient to assume a single
clock constant $\clockcnst$, and this can be achieved by a precompilation adding this
to the context. Similarly to function types, path equality in $\forall \kappa.\, A$ is equivalent 
to pointwise equality, and clock quantification preserves truncation levels.

Ticks are evidence that time has passed on a given clock. Tick assumptions in a context
separate a judgement such as $\hastype{\Gamma, \tickA : \kappa, \Gamma'}tA$
into a part ($\Gamma$) of assumptions arriving before the time tick $\tickA$ on the
clock $\kappa$ and the rest of the judgement that happens after. The introduction
rule for ${\latbind\tickA\kappa A}$ in \autoref{fig:clott:typing} can therefore be read
as stating that if $t$ has type $A$ one time step after the assumptions in $\Gamma$,
then $\tabs\tickA\kappa t$ is delayed data of type $A$. Because terms can appear
in types, $\tickA$ can appear in $A$, and must therefore
be mentioned in the type of $\tabs\tickA\kappa t$. In many cases, however,
it does not, and we will then write simply $\later^\kappa A$ for $\latbind\tickA\kappa A$.

Elimination for $\latbind\tickA\kappa A$ is by application to ticks. There are two forms of ticks:
simple ticks and forcing ticks. The judgement for simple ticks
$\istick{\Gamma}{u} {\kappa} {\Gamma'}$
states that $u$ is a tick on the clock $\kappa$ in context $\Gamma$ with \emph{residual
context} $\Gamma'$. The residual context consists of the assumptions that were available
before the tick $u$. In the case of a tick variable $\tickA$, these are the assumptions to
the left of $\tickA$ as well as the \emph{timeless} assumptions to the right of $\tickA$.
Timeless assumptions are interval and clock assumptions ($i : \I$
and $\kappa : \clocktype$) as well as all faces, and $\cuba{\Gamma}$
computes the timeless assumptions of $\Gamma$. The assumption $\Gamma'' \subcxt \Gamma$
defined by the rule 
\[
\Gamma, \cuba{\Gamma'} \subcxt \Gamma, \Gamma'
\]
allows for further trimming of the residual context.
The term $\tappg {t} {\Gamma'} {u}$
applies $t$ to the simple tick $u$. The assumption that $t$ is well-typed in the
residual context $\Gamma'$ should be read as stating that it has the type $\latbind \tickA\kappa A$
before the tick $u$, and can therefore be opened after the tick $u$ to produce something of
type $A\subst \tickA u$. We will often omit the residual context $\Gamma'$ from the term, writing
simply $\tapp[u]{t}$ for $\tappg {t} {\Gamma'} {u}$. Different choices of $\Gamma'$ give the same
term up to judgemental equality.

As basic examples of programming
with simple ticks we mention the dependently typed generalisation of applicative action
(\ref{eq:appl:action}), 
an equivalence of types 
\begin{equation} \label{eq:later:sigma}
 \later^\kappa (\Sigma (x : A) . B(x)) \equi \Sigma (x : \later^\kappa A) . \latbind \tickA\kappa(B(\tapp x)) 
\end{equation}
and the unit
\begin{equation} \label{eq:next}
 \lambda x . \tabs\tickA \kappa x : A \to \later^\kappa A
\end{equation}
Note that the latter
uses the rule for variables which allows these to be introduced from anywhere in the
context, also across ticks. Some Fitch style modal type theories do not allow
such introductions~\cite{clouston2018fitch,drat}.
The timelessness of the interval is needed to type
the extensionality principle 
\[
\Path{\later^\kappa A}xy
\equi
\latbind\tickA\kappa{\Path A{\tapp x}{\tapp y}}
\]
where the left to right direction $\lambda p. \tabs \tickA \kappa{\lambda i. \tapp{p\, i}}$
requires $\tapp{p\, i}$ to be welltyped in a context where $i$ occurs after $\tickA$.

\subsection{Forcing ticks}

Unrestricted elimination of $\later$ is unsound, since any type $A$ with a map
$\later^\kappa A \to A$ can be inhabited using fixed points. However, it is safe to
eliminate $\later^\kappa$ in contexts of the form $\Gamma, \kappa : \clocktype$,
i.e., where no variables mention $\kappa$ in their type. To close such an elimination
rule under substitutions, the clock variable $\kappa$ must be abstracted in the term.

Application to forcing ticks does exactly that:
If $\kappa'$ is a clock in $\Gamma$ then $\isbtick {\Gamma} {\kappa'} {\diamond} {\Gamma}$
is a forcing tick and so if $\hastype{\Gamma, \kappa : \clocktype} {t} {\later (\tickA : \kappa). A}$
then
\[\hastype \Gamma {\tappg {(\kappa.t)} {\Gamma} {\btick {\kappa'} \tickc}} {A\tickcsub {\tickA}{\kappa} {\tickc}{\kappa'}}
\]
This construction binds $\kappa$ in $t$, and 
$A\tickcsub {\tickA}{\kappa} {\tickc}{\kappa'}$ uses a special simultaneous substitution of the pair $(\tickc, \kappa')$
for the pair $(\tickA, \kappa)$. In particular $\tickcsub {\tickA}{\kappa} {\tickc}{\kappa'}$ commutes as expected
with type and term formers, replacing each of $\tickA$ and $\kappa$ with $\tickc$ and $\kappa'$ as intended,
but it will also have to turn a simple tick application like $\tappg t {} \alpha$ into a forcing tick application
$\tappg {(\kappa.t)} {} {\btick {\kappa'} \tickc}$ to preserve typing. This substitution operation was originally
described by \citet{clottmodel}.

Forcing ticks can be used to define the $\force$ operator of \citet{atkey13icfp}
\begin{equation} \label{eq:force}
\force \defeq \lambda x . \lambda \kappa . \tappg{(\kappa . \capp x)}{}{\btick\kappa \tickc} : \forall\kappa . \later^\kappa A \to \forall\kappa . A
\end{equation}
which is used in the encoding of coinductive types. Replacing the delayed substitution of $\kappa'$ 
for $\kappa$ in $\tappg {(\kappa.t)} {\Gamma} {\btick {\kappa'} \tickc}$ with an actual substitution, 
gives the elimination rule used by \citet{bahr2017clocks}.
Here we opt for an explicit substitution because it is easier to type check and give semantics to, which is
why \citet{clottmodel} made a similar choice.
\citet{bahr2017clocks} designed their application rule for $\tickc$ to obtain a normalisation result for their theory.
We expect that using the present rule, one can prove the same property for \cctt.

\subsection{Tick irrelevance}


Since fixed points unfold when applied to $\tickc$, the identity of ticks is crucial
for the operational properties of Clocked Type Theory, in particular for ensuring
strong normalisation. However, on the extensional level the identity of ticks is
irrelevant, as
stated by the \emph{tick irrelevance axiom}
\begin{equation} \label{eq:tirrAx}
 \tirrAx 
  : \Pi(x : \later^\kappa A) . \latbind{\tickA}{\kappa}  \latbind{\tickB}{\kappa} \Path{A}{\tapp{x}}{\tapp[\tickB] x}
\end{equation}
This can be inhabited in \cctt\ using the novel construction $\tirr$ building a path between any pair of ticks
by defining
\[
\tirrAx(t) \defeq \tabs \tickA \kappa {\tabs\tickB\kappa {\lambda i .\tappg t{}{\tirr(\tickA, \tickB, i)}}}
\]
Similar axioms can be constructed for forcing tick applications.
One important application of $\tirr$ is in showing that $\force$ (\ref{eq:force})
is a type equivalence with inverse
$\inv{\force} \defeq\lambda x . \lambda \kappa . \tabs \tickA\kappa{\capp x}$: If $x : \forall\kappa . \later^\kappa A$, then
\begin{align*}
 \inv\force(\force\, x) & \jeq \lambda \kappa . \tabs \tickA \kappa{\tappg{(\kappa .{\capp x})}{}{\btick \kappa \tickc}}\\
 & \peq \lambda \kappa . \tabs \tickA \kappa{\tappg{(\kappa .{\capp x})}{}{\btick \kappa \tickA}}\\
 & \jeq \lambda \kappa . \tabs \tickA \kappa{\tappg{{\capp x}}{}{\tickA}}
 \jeq x
\end{align*}
where the path equality is witnessed by
\[\lambda i . \lambda \kappa . \tabs \tickA \kappa{\tappg{(\kappa .{\capp x})}{}{\btick \kappa{\tirr(\tickc,\tickA, i)}}}.\]

The ability of $\tirr$ to form a path between any two ticks allows us to prove coherence laws between different uses of tick irrelevance. For example we can
construct fillers for boxes whose boundary is given by tick applications as in the lemma below.
\begin{lemma}\label{lem:tirr:coherence}
Let $\hastype{\Gamma'}{t}{\latbind \alpha \kappa A}$ and $\istick{\Gamma, i_0, \ldots, i_n : \I, \phi} {u} {\kappa}{\Gamma'}$ where $\phi = \bigvee_{k,b}\,(i_k = b)$ for $k = 0,\ldots,n$ and $b = 0,1$.
Then we have $\istype{\Gamma, i_0, \ldots, i_n : \I} {B[\phi \mapsto A\subst \alpha u]}$ and a term $\hastype{\Gamma, i_0, \ldots, i_n : \I}{\mathsf{filler}(t,u)}{B[ \phi \mapsto \tappg t {} u]}$.
\end{lemma}
For example, the lemma constructs a filler of the diagram below, thus constructing a path from
the composition of $\tappg {\tappg {\tirrAx(t)}{}{\alpha}} {} {\beta}$ and
$\tappg {\tappg {\tirrAx(t)}{}{\beta}} {} {\gamma}$ to $\tappg {\tappg {\tirrAx(t)}{}{\alpha}} {} {\gamma}$,
     \[
      \begin{tikzcd}
        A \ar[dash]{rr}{\tappg {\tappg {\tirrAx(t)}{}{\beta}} {} {\gamma}}
        \ar[dash]{d}[swap]{\tappg {\tappg {\tirrAx(t)}{}{\alpha}} {} {\beta}}
        && A \ar[dash]{d}{\tappg {\tappg {\tirrAx(t)}{}{\alpha}} {} {\gamma}} \\
        A \ar[dash]{rr}{\tappg{t}{}\tickA} && A
      \end{tikzcd}
    \]


The rules for judgemental equality (\autoref{fig:clott:equalities}) include standard
$\beta$ and $\eta$-rules for functions, clock abstraction and tick abstraction. Note that there are two $\beta$
equalities for tick application, one for each kind of ticks. There is also a rule stating that if a forcing tick
does not contain $\tickc$, application to it reduces to application to a simple tick. One consequence of this
is that an $\eta$-rule for forcing tick application can be derived as follows
\begin{align*}
  \tabs{\tickA}{\kappa'} (\tappg {(\kappa . t)}{}{\btick \tickA{\kappa'}}) &\jeq
  \tabs{\tickA}{\kappa'} (\tappg {t\subst\kappa{\kappa'}}{}{\tickA}) \jeq t\subst\kappa{\kappa'}
\end{align*}
Although we do not prove canonicity for our type theory, we set up equalities that we believe are
necessary to compute canonical forms in contexts of free clock and interval variables, but not free tick
variables. For example, $\tirr(\tickc,\tickc,r) \jeq \tickc$ ensures that in a context with no tick variables
all ticks are equal to $\tickc$.

\subsection{Fixed points}

The operator $\dfix^\kappa : (\later^\kappa A \to A) \to \later^\kappa A$ computes fixed
points of productive functions. Using this, one can construct Nakano's fixed point operator
as
\begin{equation} \label{eq:fix}
\fix^\kappa \defeq \lambda f . f(\dfix^\kappa f) : (\later^\kappa A \to A) \to A
\end{equation}

To ensure termination, fixed points do not unfold judgementally, except
when applied to $\tickc$. As in Guarded Cubical Type Theory~\cite{GCTT} we add a path $\pfix$
between a fixed point and its unfolding.
Note that the right hand endpoint of $\pfix t$ is $\fix^\kappa t$ and so, 
\begin{equation} \label{eq:fix:unfold}
\fix^\kappa t \jeq t(\tabs\tickA\kappa{\tapp{\dfix^\kappa t}})
\peq t(\tabs\tickA\kappa{\fix^\kappa t})
\end{equation}
With this we can prove that fixed points are unique, in fact the following stronger statement is true.
\begin{lemma} \label{lem:fix:contractible}
The types $\Sigma (x : A). \Path Ax{f\,(\tabs\tickA\kappa{x})}$ 
and 
\[\Sigma (x : \later^\kappa A). \latbind{\tickA}{\kappa}{\Path{A}{\tapp{x}}{f(x)}}\] 
are contractible for every $f : \later^\kappa A \to A$.
\end{lemma}
By lemma \ref{lem:fix:contractible} the pair $(\dfix^\kappa t, \pfix t)$ is uniquely determined up to path, as it witnesses that $\dfix^\kappa t$ is the fixed point of $\lambda x. \tabs \_ \kappa t\,x$.
Moreover both $\dfix^\kappa$ and $\pfix$ unfold when applied to $\tickc$, which means they will not be stuck terms in a context with no tick variables.

The principles above can be used to encode guarded recursive types as fixed points on the universe.
For example, assuming codes for products and natural numbers, we define
\[
 \gStrcode(\Nat) \defeq \fix^\kappa(\lambda X. \natscode \timescode\latbindcode\tickA\kappa{\tapp X})
\]
and then by (\ref{eq:fix:unfold})
 $\gStrcode(\Nat) \peq \natscode \timescode \latbindcode\tickA \kappa{\gStrcode(\Nat)}$
so, defining $\gStr(\Nat) \defeq \El(\gStrcode(\Nat))$ gives the type equivalence
\begin{equation} \label{eq:gstr:unfold}
  \gStr(\Nat) \equi \Nat \times \later^\kappa \gStr(\Nat)
\end{equation}

%

%

Note that we do not assume the \emph{clock irrelevance axiom}
\begin{equation} \label{eq:cirr}
\inferrule*
{\hastype{\Gamma}{t}{\forall\kappa . A} \\ \kappa \notin \fc{A}}
{\hastype\Gamma{\cirr t}{\forall{\kappa'} . \forall{\kappa''} . \Path{A}{t [\kappa']}{t [\kappa'']}}}
\end{equation}
often assumed
in type theories with multi-clock guarded recursion \cite{GDTT,bahr2017clocks}.
Instead we will use the following definition.
\begin{definition} \label{def:cirr}
A type $A$ is \emph{clock-irrelevant} if the canonical map $A \to \forall \kappa . A$
(for $\kappa$ fresh) is an equivalence.
\end{definition}
If $A$ is clock-irrelevant it satisfies axiom (\ref{eq:cirr}), and assuming a clock constant
the two statements are equivalent.
Clock irrelevance is needed for the correctness of encodings of coinductive types.
For example, using (\ref{eq:gstr:unfold}) one can prove
\[
\forall\kappa . \gStr(\Nat) \equi (\forall\kappa. \Nat) \times \forall\kappa .(\later^\kappa \gStr(\Nat))
\equi \Nat \times \forall\kappa . \gStr(\Nat)
\]
assuming $\Nat$ clock-irrelevant, and using that $\force$ is an equivalence.
In this paper, rather than assuming all types clock irrelevant, we will prove clock-irrelevance
for a large class of types and use this to construct final coalgebras for a correspondingly large
class of functors.

\section{Encoding coinductive types}
\label{sec:coinductive:types}

In the previous section we saw that the type $\forall\kappa . \gStr(\Nat)$ satisfies the
type equivalence expected of a type of streams. In this section we prove that (assuming
$\Nat$ clock-irrelevant) this type also satisfies the universal property characterising
streams, as a special case of a more general property. In Section~\ref{sec:hits}
we will show that $\Nat$ is indeed clock-irrelevant (as a special case of \autoref{prop:cirr:HITs})
as are a large collection of inductive and higher inductive types.

Here we essentially adapt the final coalgebra construction from \citet{Mogelberg14}. 
Given a type $I$ consider the category\footnote{Here we mean the naive internal notion where the type of arrows is not necessarily truncated, but we also do not require higher coherences. Same for functor below.}
$I \to \Univ$ of $I$-indexed types whose type of objects is
$I \to \Univ$  and whose morphisms from $X$ to $Y$ is the type
\[
 X \to Y \defeq \Pi (i:I). \, \El(X\,i) \to \El(Y\,i)
\]
An \emph{$I$-indexed endofunctor}
is an endofunctor on $I \to \Univ$.
%

Let $F$ be an $I$-indexed endofunctor, an $F$-\emph{coalgebra} is a
pair $(X,f)$ such that $X : I \to \Univ$ and $f : X \to F\,X$.
We say an $F$-coalgebra $(Y,g)$ is \emph{final} if for all coalgebras $(X,f)$ the following type is contractible
  \[
  \Sigma (h : X \to Y).\, g \circ h \peq F(h) \circ f
  \]
Using contractibility we address the fact that there can be more
than one proof of path equality, as is also done by \citet{Ahrens15}.

If $X : \forall \kappa.~ (I \to \Univ)$ write $\forall\kappa . X$ for
$\lambda (i :I) . \forallcode\kappa .{(\capp X\, i)}$.

\begin{definition} \label{def:comm:forall}
An
$I$-indexed endofunctor $F$ \emph{commutes with clock quantification}
if for all families $X$ the canonical
map \[\mathsf{can}_F : F\,(\forall \kappa. \capp X) \to \forall \kappa.\, F\,(\capp X)\]
is a pointwise equivalence.
\end{definition}

For example, if $A : \Univ$ the constant functor to $A$
commutes with clock quantification if and only if $A$ is clock-irrelevant.

\begin{lemma} \label{lem:comm:kappa}
 The collection of endofunctors commuting with clock quantification is
 closed under composition, pointwise product,
 pointwise $\Pi$, pointwise $\Sigma$ over clock irrelevant types,
 and pointwise universal quantification over clocks.
 If $F$ commutes with clock quantification then the guarded
 recursive type $X$ satisfying $X \equi F(\later^\kappa X)$
 is clock irrelevant. Any path type of a clock irrelevant
 type is clock irrelevant.
\end{lemma}

The following theorem states that all such endofunctors
have final coalgebras. The idea of this originates with
\citet{atkey13icfp}, and our proof is an adaptation of the
proof by \citet{Mogelberg14} of a similar statement in extensional type theory.
It refers to the endofunctor $\later^\kappa$, defined as the pointwise
 extension of the functor $\Univ \to \Univ$ mapping
 $A$ to $\latbindcode\tickA\kappa A$ for a fresh $\tickA$.

\begin{theorem} \label{thm:final:coalg}
  Let $F$ be an $I$-indexed endofunctor which commutes with clock quantification,
  and let $\nu^\kappa(F)$ be the guarded recursive type satisfying
  $\nu^\kappa(F) \equi F(\later^\kappa(\nu^\kappa(F))$. Then
  $\nu(F) \defeq \forall \kappa. \nu^\kappa(F)$ has a final $F$-coalgebra structure.
\end{theorem}

\begin{example}
In Section~\ref{sec:hits} we will prove that
$\nats$ is clock-irrelevant, and hence the functor $F(X) = \nats \times X$
commutes with clock quantification. Theorem~\ref{thm:final:coalg} then states correctness
of the encoding of streams as $\forall\kappa . \gStr(\Nat)$.
\end{example}

\begin{example} \label{ex:glts}
 Consider the functor $F(X) = \Pfin(A \times X)$, where $\Pfin$ is the finite powerset
 functor defined as a HIT~\cite{FGGW18},
 and $A$ is assumed to be a clock-irrelevant set. In Section~\ref{sec:hits}
 we will prove that $\Pfin$ commutes with clock quantification,
 which implies that $F$ does the same. Theorem~\ref{thm:final:coalg}
 then gives an encoding of the final coalgebra for $F$. This type plays an important role
 in automata theory as it describes the finitely branching $A$-labelled transition systems.
 Let $\GLTS[\kappa]$ denote fixed point for $F\circ \later^\kappa$, let
 $\GLTS \defeq \forall\kappa . \GLTS[\kappa]$ be the coinductive type and let
 $\unfoldG : \GLTS \to \Pfin(A \times \GLTS)$ be the final coalgebra structure.

 Suppose now $x,y: \GLTS$ and $R : \GLTS \times \GLTS \to \Univ$. Define 
 \[
   \Bisimforall{R}(x,y) = \Simforall{R}(x,y) \times  \Simforall{R}(y,x)
 \]
 where $\Simforall{R}(x,y)$ is
\begin{align*}
 \Pi x', a . (a, x')\in \unfoldG(x) \to
  \exists y' . ((a,y') \in \unfoldG(y)) \times R(x', y')
\end{align*}
where $\in$ is defined by induction~\cite{FGGW18}.
By Lemma~\ref{lem:comm:kappa}, $\GLTS$ is clock irrelevant, and an easy induction
on $X : \Pfin(A \times \GLTS)$ shows that the predicate $(a, x') \in X$ is clock irrelevant.
In Section~\ref{sec:hits}
we will prove that propositional truncation commutes with clock quantification,
and since $\exists$  in homotopy type theory is defined as the composition
of truncation and $\Sigma$~\cite{hottbook}, putting all this together shows that $\Bisimforall -$ defines a
$\GLTS \times \GLTS$-indexed endofunctor commuting with clock quantification,
and so Theorem~\ref{thm:final:coalg}
gives an encoding of bisimilarity as a coinductive type.
\end{example}

\citet{mogelbergPOPL2019} prove that path equality
coincides with bisimilarity for guarded recursive types. Using their results
we can prove the following.

\begin{theorem}\label{thm:bisim:is:identity}
 Two elements of the type $\GLTS$ from Example~\ref{ex:glts}
 are bisimilar if and only if they are path equal. Since both
 these are propositions, they are equivalent as types.
\end{theorem}


The results of Example~\ref{ex:glts} and Theorem~\ref{thm:bisim:is:identity} 
are liftings of results by \citet{mogelbergPOPL2019} 
for guarded recursive types to coinductive types. 
One benefit of passing to coinductive types 
is the existence of non-causal functions. For example, it is not
possible to write an endofunction on $\gStr(\Nat)$ filtering out every other 
element of the input stream. Instead this can be written as an endofunction
on $\Str(\Nat)$~\cite{atkey13icfp}.
Similar examples can be imagined for
\autoref{ex:glts}, for example, a function that collapses an LTS over labels to
one over pairs of labels. Another benefit is that weak bisimilarity can be defined 
as an equivalence relation on coinductive types. 
\citet{Paviotti:FPC:journal} construct a weak bisimilarity relation
on the guarded recursive type $\gLift X \equi X + \later^\kappa\gLift X$, 
and lift this to a relation on a model of FPC. 
The weak bisimilarity relation on $\gLift X$ is not transitive, in fact, its transitive closure
relates all pairs of values, but the transitive closure of the induced relation on 
$\Lift X \defeq \forall\kappa . \gLift X$ coincides with the 
one studied by \citet{capretta2005general} for the coinductive delay monad.

\section{Higher Inductive Types}
\label{sec:hits}

We now extend \cctt\ with higher inductive
types,  adapting a schema for these from
\citet{CavalloHarper} to
our version of CTT. For simplicity we exclude indexed families, but the
schema is still general enough to cover examples like spheres,
pushouts, W-suspensions \cite{WSusp}, (higher) truncations, as well as
finite and countable powersets. We first present the schema and
the introduction rules for HITs, then (\autoref{sec:hits:clock}) we
present our principle of induction under clocks, which generalises
the elimination rule for HITs.

%

%

\subsection{Introduction and formation}
\label{sec:hits:syn}

\begin{figure}
  \begin{mathpar}
  \inferrule*
      { }
      {\isbox {\cdot}} \and
  \inferrule*
      {\isbox \Psi}
      {\isbox {\Psi , i : \I}}\and
  \inferrule*
      {\wfcxt {\Gamma_0}}
      {\istel {\Gamma_0} {\cdot}} \and
  \inferrule*
      {\istel {\Gamma_0} \Gamma \\ \istype {\Gamma_0,\Gamma} A}
      {\istel {\Gamma_0} {\Gamma , x : A}}
  \end{mathpar}
  \caption{Telescope judgements.}
  \label{fig:tel}
\end{figure}

\begin{figure}
  \begin{center}
  \textbf{Constructor declarations} (presuppose \istel \cdot \Delta)
  \begin{mathpar}
  \inferrule*{ \mathcal{K} = (\ell_1, C_1) \ldots (\ell_n, C_n) \\
    (\forall i)~ \isconstr \Delta {(\ell_1, C_1) \ldots (\ell_{i-1}, C_{i-1})} {C_i}}
  {\isconstrs \Delta {\mathcal{K}}}
  \and
    \inferrule*
          {
              \istel \Delta \Gamma \quad (\forall k)~\istel {\Delta,\Gamma} {\Xi_k}
              \quad \isbox \Psi
              \quad \hastype {\Psi} {\phi} \F \\
              \hastype [\mathcal K,\HIT{H}\delta] {\delta : \Delta,\Gamma,
              \bar{\Xi \to \HIT H \delta},\Psi,\varphi} {e} {\HIT{H}{\delta}} \\
            e = [\varphi_1~M_1 \ldots \varphi_m~M_m] \quad
            (\forall k)~M_k \in \Bndr(\mathcal{K};\bar{\Xi\to\HIT{H}\delta})\\
          }
          {\isconstr {\Delta}{\mathcal{K}} {(\Gamma;\bar{\Xi};\Psi;\varphi;e)}}
          \end{mathpar}
  \textbf{Grammar for boundary terms}
  \begin{align*}
    {N, M} \in \Bndr(\mathcal{K};\Theta)  ::= & \quad
     x~\bar{u} \\ 
    & \mid  \con_{\ell}(\bar{t},\overline{\lambda \xi.\,M},\bar{r}) \\ 
    & \mid \hcomp_{\HIT{H}\delta}^j~[\phi \mapsto M]~{M}_0
  \end{align*}
  \textbf{Formation}
  \[
  \inferrule*{\hastype \Gamma \delta \Delta}
             {\istype \Gamma {\HIT{H}{\delta}} }
  \]
  \textbf{Introductions}
\begin{mathpar}
    \inferrule*
          {(\ell,(\Gamma;\bar{\Xi};\Psi;\varphi;e)) \in \mathcal{K} \\
         \hastype {\Gamma_0} \delta \Delta \\
         \hastype {\Gamma_0} {\bar{t}} {\Gamma[\delta]} \\
         \hastype {\Gamma_0} {\bar{a}} {\bar{\Xi[\delta,\bar{t}]\to \HIT{H}\delta}} \\
         \hastype {\Gamma_0} {\bar{r}} \Psi
          }
          {
      \hastype {\Gamma_0} {\con_\ell(\bar{t},\bar{a},\bar{r})} {\HIT{H}\delta[ \varphi[\bar{r}] \mapsto e[\bar{t},\bar{a},\bar{r} ] ] }
    }
\and
\inferrule*{ \hastype \Gamma \delta \Delta \\ \hastype \Gamma \phi \F \\ \hastype \Gamma {u_0} {\HIT H \delta[\phi \mapsto u\subst i 0]} }
         {\hastype \Gamma {\hcomp^i_{\HIT H \delta}\,[\phi \mapsto u]\,u_0} {\HIT H \delta[\phi \mapsto u\subst i 1]} }
\and
         \inferrule*{\hastype \Gamma \phi \F \\
           \hastype {\Gamma, i : \I} \delta \Delta \\ \hastype {\Gamma, i : \I, \phi} {\delta \jeq \delta[0]} {\Delta}
           \\ \hastype \Gamma {u_0} {\HIT{H}{\delta[0]}}}
           {\hastype {\Gamma} {\trans^i_{\HIT{H}{\delta}}\,\phi\,u_0} {\HIT{H}{\delta[1]} [ \phi \mapsto u_0 ]} }
 \end{mathpar}

\caption{Schema for Higher Inductive Types, $\istype {\delta : \Delta}
  {\HIT H \delta}$.  In the typing of the boundary $e$ the subscript
  $\mathcal K, \HIT{H}{\delta}$ indicates that this judgement can refer
  to the labels from $\mathcal K$ and $\HIT{H}{\delta}$ itself. The
  grammar for boundary terms assumes $x \in \Theta$, that $\overline u, \overline t$ not
  mention variables in $\Theta$ and that $\ell \in \mathcal{K}$.
  For further judgemental equalities for $\trans$ see Appendix~\ref{appendix:trans}.}
\label{fig:hits:constrs}
\end{center}
\end{figure}

The judgements of \autoref{fig:hits:constrs} capture declarations of the form
\[
\begin{array}{l}
  \textsf{data}\,\HIT H {(\delta : \Delta)}\,\textsf{where}\\
  \quad \vdots\\
  \quad \ell_i : (\gamma : \Gamma) \to (\bar x : \overline{\Xi \to \HIT H \delta})
                 \to (\bar{i} : \Psi) \to \HIT H \delta [ \phi \mapsto e ]\\
  \quad \vdots
\end{array}
\]
which list constructors and their types for a new datatype
\textsf{H}, taking parameters in the telescope $\Delta$.
On top of the declared constructors every HIT has an
introduction form for \emph{homogeneous composition},
written $\hcomp^i_{\HIT{H}{\delta}}\,[\phi \mapsto u]\,u_0$, where
$\delta$ is not allowed to depend on $i$, as well as a transport operation.
Following \citet{CTTHITS}, the composition
structure for \textsf{H} is given by combining these two operations.
We omit the details for lack of space,
but just recall that from the homogeneous composition
operator one can derive a
\emph{homogeneous filling} operator $\hfillc^i_A\,[\phi \mapsto
  u]\,u_0 : \Path A {u_0} {\hcomp^i_A\,[\phi \mapsto u]\,u_0}$ which
provides a filler for the box specified by $u$ and $u_0$ and closed by
homogeneous composition \cite{HCompNote}.

A constructor for a HIT \textsf{H} is specified by a tuple \\ $\isconstr \Delta
{\mathcal K} {(\Gamma,\bar \Xi,\Psi,\phi,e)}$ where $\mathcal K$ lists
the previously declared constructors. The telescope $\Gamma$ lists
the non-recursive arguments, and each $\Xi$ in $\bar{\Xi}$ is the
arity for a recursive argument. Path constructors further take a
non-empty telescope of interval
variables $\Psi$ and have a boundary $e$ specified as a partial
element over the formula $\phi$.
The judgements for telescopes are given in \autoref{fig:tel}. Note that
these imply that all of
$\Delta$, $\Gamma$, and $\Xi$ only contain variables of proper types, i.e.,
no interval, face, clock or tick assumptions.
Only the boundary is allowed to refer to $\textsf{H}$
and the previous constructors from $\mathcal K$. We signify this by
adding those as subscripts to the typing judgement for $e$.
Cavallo and Harper require the components $M$ of the boundary to
conform to a dedicated typing judgement, which restricts their shape
and makes it possible to correctly specify the inputs to the dependent eliminator for \textsf{H}.
For conciseness we replicate those restrictions with a grammar,
$\Bndr(\mathcal K,\Theta)$, which specifies that a boundary term must be
built either from an applied recursive argument $x\,\bar{u}$, a
previous constructor, or a use of $\hcomp$.
We also assume a code $\HITcode{H}{\delta} : \Univ_i$
whenever the types in all of the $\Gamma$ and $\Xi$ telescopes in
$\mathcal K$ have a code in $\Univ_i$ as well.

\begin{remark}
 Formally, boundary terms are a separate syntactic entity which we
 silently include in ordinary terms. In particular, they have a separate
 equality judgement which is used when typing systems and compositions
 of boundary terms. This equality is the congruence relation induced by
 reducing compositions and constructors  to their boundary terms. It is
 likely that this equality can be proved to coincide with the one induced by
 equality on terms. 
 This approach is similar to the one of~\citet{CavalloHarper}.
\end{remark}

We now give some concrete examples. For readability we write
$(\ell,(\Gamma;\overline{(x : \Xi \to \HIT H \delta)};\Psi;[\phi_1 \mapsto
  M_1, \ldots, \phi_n \mapsto M_n]))$ in place of
$(\ell,(\Gamma,\bar\Xi,\Psi,\vee_i \phi_i,[\phi_1 \mapsto M_1, \ldots,
  \phi_n \mapsto M_n]))$. We also write $\ell$ rather than $\con_\ell$.

\begin{example} \label{ex:HITs}
\textbf{Pushout}. The context $\Delta$ contains the data of a pushout, i.e.,
$\Delta \defeq (A\, B\, C : \Univ)(f : \El (C \to A))(g : \El (C \to B))$. The pushout has
two point constructors, and a path.
  \begin{itemize}
  \item $(\mathsf{inl}, ((a : \El A);\cdot;\cdot;[]))$
  \item $(\mathsf{inr}, ((b : \El B),\cdot;\cdot;[]))$
  \item $\left(\mathsf{push}, \left( (c : \El C); \cdot; (i : \I); \left[
    \begin{array}{l}
      i = 0 \mapsto \mathsf{inl}(f\,c,\cdot,\cdot)\\
      i = 1 \mapsto \mathsf{inr}(g\,c,\cdot,\cdot)
    \end{array}
  \right]\right)\right)$
  \end{itemize}

  \textbf{Sphere}. We can encode the circle and higher spheres, $\sphere {n}$ for $n \geq 1$,
  as a point and an $n$-dimensional surface
  \begin{itemize}
  \item $(\mathsf{base}, (\cdot;\cdot;\cdot;[]))$
  \item $\left(\mathsf{surface}, \left(\cdot;\cdot;\bar{i} : \I^{n};
    \left[ \begin{array}{l} i_1 = 0 \vee i_1 = 1 \mapsto \mathsf{base}(\cdot;\cdot;\cdot)  \\
        \ldots\\
         i_n = 0 \vee i_n = 1 \mapsto \mathsf{base}(\cdot;\cdot;\cdot)  \end{array}
        \right]\right)\right)$
  \end{itemize}
  Note that this defines $\sphere {n}$ for each external $n$. A type
  of $n$-spheres internally indexed over $n : \nats$ can be defined 
  using suspensions~\cite{hottbook}.

  \textbf{Propositional and Higher Truncation}.
  Let $\Delta \defeq (A : \Univ)$, and write propositional truncation as $\trunc A_0$.
  \begin{itemize}
  \item $(\mathsf{in}, (\El A;\cdot;\cdot;[]))$
  \item $\left(\mathsf{squash}, \left(\cdot; (a_0, a_1 : \trunc A_0); (i : \I);
     \left[ \begin{array}{l} i = 0 \mapsto a_0 \\ i = 1 \mapsto a_1 \end{array} \right]\right)\right)$
  \end{itemize}
  For higher truncations $\htrunc A n$, where $n \geq 0$, we use the hub and spoke construction \cite[Sect. 7.3]{hottbook},
  as the schema does not allow quantifying over paths of the HIT itself.\footnote{We believe the semantics would support doing so, but it would complicate the schema.}
  Instead of $\mathsf{squash}$ then we have the following two constructors:
  \begin{itemize}
  \item $(\mathsf{hub}, (\cdot; (f : \sphere {n+1} \to \htrunc A n); \cdot; []))$
  \item $(\mathsf{spoke},((x : \sphere{n+1}); (f : \sphere {n+1} \to \htrunc A n);(i : \I);e))$
  \end{itemize}
  where $e \defeq [ i = 0 \mapsto f\,x, i = 1 \mapsto \mathsf{hub}(\cdot, f,\cdot)]$.\\[1ex]


\textbf{Finite Powerset}. The finite powerset $\Pfin(A)$ is defined in context 
$\Delta \defeq (A : \Univ)$ using the following constructors
%
  \begin{itemize}
  \item $(\emptyset, (\cdot; \cdot; \cdot; []))$
  \item $(\{-\}, (a : \El A; \cdot; \cdot; []))$
  \item $(\cup ,(\cdot; X : \Pfin (A), Y : \Pfin (A), ;))$
  \item $\left(\mathsf{nl}, \left(\cdot, (X : \Pfin(A)), (i : \I),
    \left[ \begin{array}{l}
         i = 0 \mapsto \emptyset \cup X \\
         i = 1 \mapsto X
        \end{array}
   \right]\right)\right)$
  \item $\!\left(\mathsf{as}, \left(\cdot, (X, Y, Z : \Pfin(A)), (i : \I),
    \left[ \begin{array}{l}
       \! \! \! i = 0 \!\mapsto \!(X \!\cup\! Y) \!\cup \! Z \! \! \!\\
        \! \! \! i = 1\! \mapsto \!X \!\cup\! (Y\! \cup\! Z) \! \!\!
        \end{array}
   \right]\right)\right)$
  \item $\left(\mathsf{comm}, \left(\cdot, (X, Y : \Pfin(A)), (i : \I),
    \left[ \begin{array}{l}
         i = 0 \mapsto X \cup Y \!\! \\
         i = 1 \mapsto Y \cup X\!\!
        \end{array}
   \right]\right)\right)$
  \item $\left(\mathsf{idem}, \left(\cdot, (X : \Pfin(A)), (i : \I),
    \left[ \begin{array}{l}
         i = 0 \mapsto X \cup X \\
         i = 1 \mapsto X
        \end{array}
   \right]\right)\right)$
  \item $(\mathsf{hub}, (\cdot; (f : \sphere {1} \to \Pfin (A)); \cdot; []))$
  \item $(\mathsf{spoke},((s : \sphere{1}); (f : \sphere {1} \to \Pfin (A));(i : \I);e))$
  \end{itemize}
where $e \defeq [ i = 0 \mapsto f\,x, i = 1 \mapsto \mathsf{hub}(\cdot, f,\cdot)]$.\\[1ex]
 \end{example}

\subsection{Induction under clocks}
\label{sec:hits:clock}

We now present the principle of induction under clocks. This is an elimination principle
defining elements of dependent families of the form
$\istype {\Gamma, h : \forall\bar\kappa . \HIT H {(\delta[\bar\kappa])}} D$
for a vector of clock variables $\bar \kappa$ by defining its action on elements
of the form $h = \lambda \bar\kappa . \con_\ell(\bar{t},\bar{a},\bar{r})$ in such
a way that boundary conditions are respected. In the case of
an empty list of clock variables, this principle specialises to the elimination
principle for HITs of \citet{CavalloHarper}, which we refer to as the \emph{usual
elimination principle for $H$}.


\begin{figure*}[tbp]
\begin{center}
\textbf{Elimination lists}
\begin{mathpar}
  \inferrule*
      {\hastype {\Gamma_0} \delta \forall \bar\kappa . \Delta \\
       \istype {\Gamma_0, h : \forall \bar\kappa . \HIT{H}{(\cappp \delta})} D}
      {\isclockelimlist {\Gamma_0} {\cdot} {\mathcal{\cdot}} \delta D}
\and
  \inferrule*
      {
         \isclockelimlist {\Gamma_0} {\mathcal{E}} {\mathcal{K}} \delta D \\
         			\Gamma_0, \gamma : \forall \bar\kappa . \Gamma_i[\cappp \delta],
         			\bar{x} : \bar{(\forall \bar\kappa . \Xi_i[\cappp \delta,\cappp \gamma] \to
         							 \HIT H {(\cappp \delta)})},
			        \bar{y}:\bar{\Pi(\xi : \forall \bar\kappa. \Xi_i[\cappp\delta,\cappp \gamma]).
			        D[\lambda \bar\kappa . \cappp x \cappp \xi]},
			        \\ \bar{i} : \Psi \vdash u (\gamma, x, y, \bar{i})
         			 : D[\lambda \bar\kappa . \con_\ell(\cappp \gamma,
         										  \bar{\cappp x},
         										  \bar{i})]
         			  [\phi \mapsto
             		   \evalbsynclock{\delta}{\mathcal{E}}{\bar{x} \mapsto \bar{y}}{e}]}
        {\isclockelimlist {\Gamma_0} {\mathcal{E},u} {\mathcal{K},(\ell,(\Gamma,\bar{\Xi},\Psi,\phi,e))} \delta {D}}
\end{mathpar}

\textbf{Boundary interpretation}
  \begin{align*}
\evalbsynclock{\delta}{\mathcal E}{\bar{x} \mapsto \bar{y}}{e} &= \evalbsynclockg{\delta}{\mathcal E}{\bar{x} \mapsto \bar{y}}{\cdot}{e}\\
  \evalbsynclockg{\delta}{\mathcal{E}}{\bar{x} \mapsto \bar{y}}{\hat\gamma}{[\varphi_0~M_0, \ldots, \varphi_n~M_n]}
  &= [\varphi_0~\evalbsynclockg{\delta}{\mathcal{E}}{\bar{x} \mapsto \bar{y}}{\hat\gamma}{M_0}, \ldots, \varphi_n~\evalbsynclockg{\delta}{\mathcal{E}}{\bar{x} \mapsto \bar{y}}{\hat\gamma}{M_n}]\\
    \evalbsynclockg{\delta}{\mathcal{E}}{\bar{x} \mapsto \bar{y}}{\hat\gamma}{x_j~\bar{u}} &=
    y_j~\lambda \bar\kappa . \bar{u}[\cappp \delta,\cappp \gamma, \bar i, \cappp {\hat \gamma}]\\
    \evalbsynclockg{\delta}{\mathcal{E}}{\bar{x} \mapsto \bar{y}}{\hat\gamma}{\con_{\ell}(\bar{t_\ell},\overline{\lambda \xi.\,M},
\bar{r_\ell})} &=
    \mathcal{E}_\ell[\lambda \bar\kappa.\bar{t_\ell}[\cappp \delta,\cappp \gamma,\bar i,\cappp {\hat \gamma}],
    \bar{S},\bar{R},\bar{r_\ell}]\\
    & \hspace{.4cm} (S_i = \lambda \bar\kappa.\lambda \bar{\xi_i}.\, M[\cappp \delta,\cappp \gamma,\cappp {\bar x},\bar i, \cappp {\hat \gamma}, \bar{\xi_i}], \hspace{.4cm}  
    R_i = \lambda \xi_i.
    \evalbsynclockg{\delta}{\mathcal{E}}{\bar{x} \mapsto \bar{y}}{(\hat{\gamma},\xi_i)}{M_i}
     )\\
    \evalbsynclockg{\delta}{\mathcal{E}}{\bar{x} \mapsto \bar{y}}{\hat\gamma}{\hcomp^j\,[\psi \mapsto M]\,{M}_0} &=
    \comp^j_{D[v\,j]}\,
         [\psi \mapsto \evalbsynclockg{\delta}{\mathcal{E}}{\bar{x} \mapsto \bar{y}}{(\hat\gamma,j)}{M}]
         \,\evalbsynclockg{\delta}{\mathcal{E}}{\bar{x} \mapsto \bar{y}}{\hat\gamma}{M_0}\\
        & \hspace{.4cm} (v = \mathsf{hfill}^j_{\forall \bar{\kappa}.\,\HIT H {(\cappp \delta)}}
    [\psi \mapsto \lambda \bar{\kappa}.\, M[\cappp \delta,\cappp \gamma,\cappp {\bar x},\bar i,\cappp {\hat \gamma},j]]\, (\lambda \bar{\kappa}.\, M_0[\cappp \delta,\cappp \gamma,\cappp {\bar x},\bar i,\cappp {\hat \gamma}]))
  \end{align*}
%
%
%
  \textbf{Induction under clocks}
\begin{mathpar}
  \inferrule*{
      \hastype \Gamma \delta {\forall \bar\kappa. \Delta} \quad \istype
      {\Gamma,x : \forall \kappa. \HIT H {(\cappp\delta)}} D \\
      \isclockelimlist \Gamma {\mathcal{E}} {\mathcal{K}} \delta D \\
      \hastype \Gamma u {\forall \bar\kappa . \HIT H {(\cappp \delta)}}
        }{\hastype{\Gamma}{\elim{H}{D}{\mathcal{E}}{u}} {D[u]}}
\end{mathpar}
  \textbf{Judgemental equalities}
  \begin{align*}
    \elim{H}{D}{\mathcal{E}}{\lambda \bar\kappa . \con_\ell(\bar{t},
    \bar{a},\bar{r})} & \jeq
    \mathcal{E}_\ell[\bar{\lambda\bar\kappa. t},\bar{\lambda\bar\kappa. a},
    \overline{\lambda \xi.\, \elim{H}{D}{\mathcal{E}}{\lambda\bar\kappa. a(\cappp \xi)}},\bar{r}]\\
    \elim{H}{D}{\mathcal{E}}{\lambda \bar\kappa . \hcomp
    [\phi \mapsto u]\, u_0} & \jeq
    \comp^i_{D[v\, i]} [\phi \mapsto \elim{H}{D}{\mathcal{E}}{\lambda\bar\kappa. u}] \,
    \elim{H}{D}{\mathcal{E}}{\lambda\bar\kappa. u_0}\\
    & \hspace{.4cm} (v = \mathsf{hfill}^i_{\forall \bar\kappa . \HIT H {(\cappp \delta)}}
    [\phi \mapsto \lambda \bar{\kappa}.\,u]\, (\lambda \bar{\kappa}.\,u_0))
  \end{align*}

\caption{The principle of induction under clocks. In the definition of the boundary interpretation
the notation $\mathcal{E}_\ell$ refers to the component of $\mathcal{E}$ for the label $\ell$.
Also the environment $\hat\gamma$ contains the extra assumptions introduced in the $\con$ and $\hcomp$ cases:
in the case of interval variables $\cappp i$ stands for just $i$.
}
\label{fig:hits:clockelim}
\end{center}
\end{figure*}

Figure~\ref{fig:hits:clockelim} presents the rule along with judgemental equalities.
The figure first defines the judgement $\iselimlist \Gamma {\mathcal{E}}
{\mathcal{K}} \delta D$ for an \emph{elimination list} $\mathcal E$,
which contains the premises necessary to handle the constructors
in $\mathcal K$.
The case for $\mathcal K, (\ell,(\Gamma,\bar\Xi,\Psi,\phi,e))$
requires an elimination list of the form $\mathcal E, u$ where
$\mathcal E$ is an elimination list for $\mathcal K$ and $u$ is an
element of the type family $D$ at an index built with $\con_\ell$.
In particular $u$ is typed in a context extended with non-recursive
arguments $\gamma$,
recursive arguments $\bar{x}$, interval
variables $\bar{i}$ for the constructor
$\con_\ell(\gamma,\bar{x},\bar{i})$, and also induction hypotheses
$\bar{y}$
about the variables $\bar{x}$.
The element $u$ also needs to suitably match the boundary $e$ whenever
$\phi$ is satisfied.  To express this we transform $e : \HIT H \delta$
into $\evalbsynclock{\delta}{\mathcal E}{\bar{y}}{e}$, the
corresponding element of the family $D$, built using the elimination
list $\mathcal E$ to handle the previous constructors, and the
induction hypotheses $\bar{y}$ to handle the recursive arguments
$\bar{x}$. This transformation satisfies a typing principle that we omit 
for space reasons.

In the following examples we will see how instantiating induction under clocks
with 1 or 0 clocks respectively induces equivalences which prove
that many HITs commute with clock abstraction. In particular, these
examples verify the results used in \autoref{sec:coinductive:types}.

\begin{example} \label{ex:clockedHITs}
\textbf{Spheres}.
In this case induction under clocks unfolds to the principle
\begin{equation*}
\inferrule*{ \istype {\Gamma, x : \forall \kappa . \sphere n} D \quad
			 \hastype{\Gamma}{u_b}{D[\lambda\kappa. \mathsf{base}]}\quad
 			 \hastype{\Gamma}{t}{\forall \kappa .\sphere n}\\
 			 \hastype{\Gamma, \bar i :\I^n}{u_s}{D[\lambda\kappa. \mathsf{surface}(i)]
 			 [\bigvee_{0\leq k < n} (i_k = 0 \vee i_k = 1) \mapsto u_b]}}
			{\hastype{\Gamma}{\elim{\sphere n}{D}{u_b, u_s}{t}}{D[t]}}
\end{equation*}
Using this, we define $\alpha : \forall\kappa. \sphere n \to \sphere n$ as follows
\[\alpha(t) \defeq \elim{\sphere n}{\sphere n}{\mathsf{base}, \lambda \bar i . \,
\mathsf{surface}(\bar i)}{t}\]
We must verify the boundary condition
which states that 
if a component of
$\bar{i}$ is an endpoint then $u_s (\bar i)\jeq u_b$. This follows from the
boundary condition for $\mathsf{surface}$.

Since $s \jeq \capp[\kappa_0]{(\const \, s)}$ for all $s : \sphere n$, if
$\alpha$ as defined above is a right inverse to $\const$ we will have produced
an equivalence $\sphere n \to \forall \kappa. \, \sphere n$ and shown the sphere
to be clock irrelevant.
We can achieve this with another
application of induction under a clock, inhabiting the type $s = \const
(\alpha(s))$: 
\begin{align*}
\hastype{\Gamma} {&\refl}
{\lambda \kappa. \, \mathsf{base} = \const(\mathsf{base})}\\
\hastype{\Gamma, \bar i : \I^n} {&\refl}
{\lambda \kappa. \, \mathsf{surface}(\bar i) = \const(\mathsf{surface} (\bar i))}
\end{align*}
The boundary condition in this case states that
%
the second case reduces to the first when $\bar i$ contains an endpoint. This
follows from congruence of $\refl$.

\textbf{Propositional and Higher Truncation}.
To show that propositional truncation commutes with clock quantification we first
define a map $\alpha : \forall \kappa. \, \htrunc A {} \to \htrunc{\forall \kappa. \,  A} {}$
by induction under clocks as follows:
\begin{align*}
\hastype{\Gamma, a : \forall \kappa. \, A} { \mathsf{in}(a)}
{\htrunc {\forall \kappa. \, A} {} &}\\
\hastype{\Gamma, x_0, x_1 : \forall \kappa . \, \htrunc{A}{}, y_0 ,y_1 : \htrunc
{\forall \kappa. \, A} {}, i : \I} {& \\ \mathsf{squash}(y_0,y_1,i) }
{\htrunc {\forall \kappa. \, A} {}  [(i=0) \mapsto y_0, & (i=1) \mapsto y_1]}
\end{align*}

The above data defines a map $\alpha$ 
satisfying
\begin{align*}
 \alpha(\lambda \kappa . \mathsf{in}(\capp a)) & = \mathsf{in}(a) \\
 \alpha(\lambda \kappa. \, \mathsf{squash}(\capp {x_0},\capp {x_1},r) &
 = \mathsf{squash}(\alpha (x_0),\alpha (x_1),r)
\end{align*}
Let $\beta : \htrunc{\forall \kappa. \,  A} {} \to \forall \kappa. \, \htrunc A {}$ be
the canonical map. To show that
$\beta \circ \alpha = \id{}$, it suffices by induction under clocks to show
\begin{align*}
 \beta(\alpha(\lambda \kappa . \mathsf{in}(\capp a))) & =  \lambda \kappa . \mathsf{in}(\capp a) \\
 \beta (\alpha(\lambda \kappa. \, \mathsf{squash}(\capp {x_0\!},\capp {x_1\!},r)) &
 = \lambda\kappa .\mathsf{squash}(\capp {x_0\!},\capp {x_1\!},r)
\end{align*}
In the latter case the left hand side reduces to
\[
\lambda\kappa . \mathsf{squash}(\capp{\beta(\alpha (x_0))}, \capp{\beta(\alpha (x_1))},r)
\]
and so the case follows by induction. Note that in both these uses of induction under clocks,
the boundary condition is satisfied. For example, in the latter case, when $r=0$ the term
reduces to the proof of $\beta(\alpha (x_0)) = x_0$ given by the induction hypothesis.
Showing that $\alpha \circ \beta = \id{}$  follows from an 
application of the usual elimination principle. 


For the higher truncations we apply induction under clocks again, starting by
observing that we have terms as follows:
\begin{align*}
\hastype {a : \forall \kappa. A} { & \mathsf{in}(a)} {\htrunc{\forall \kappa. \, A}{n}}\\
\hastype {x, y}
         { & \mathsf{hub}(y\circ \const)} {\htrunc{\forall \kappa. A}{n}}\\
\hastype {s : \forall \kappa. \, \sphere{n+1},
          x, y, i : \I} {&\mathsf{spoke}  (\capp [\kappa_0] s,  y \circ \const, i)}
         {\htrunc{\forall \kappa. A}{n}}
\end{align*}
where $x : \forall \kappa. \, (\sphere{n+1} \to \htrunc{A}{n})$ and
$y : (\forall \kappa.  \sphere{n+1})  \to \htrunc{\forall \kappa. A}{n}$.
For us to apply the principle, the third term would need to reduce the second when
$i=1$, which it does, and $y(s)$ when $i=0$. When $i=0$ we have instead
$\mathsf{spoke}  (\capp [\kappa_0] s, y \circ \const, i) = y (\lambda \kappa.
\, \capp [\kappa_0] s)$. From the sphere example we know that $\sphere{n+1}$ is
clock irrelevant, and hence we have a path $p : \lambda \kappa. \, \capp [\kappa_0] s
= s$. This means that we can obtain the term needed for the induction as
\begin{align*}
\hcomp^j \! \left[ \!\!\!
\begin{array}{l}
 (i=0)\mapsto y\,(p\, j) \\ (i=1)  \mapsto \mathsf{hub}(y\!\circ \!\const)
\end{array}\!\!\!
\right]
 \mathsf{spoke}  (\capp [\kappa_0] s, y\! \circ\! \const, i).
\end{align*}
We omit the details of showing that this map is inverse to the canonical one.

\textbf{Pushout}.
Using constructions similar to the examples above, one can prove
that pushouts commute with $\forall \kappa$ using induction under clocks.
More precisely, let
$A, B, C : \forall \kappa. \Univ$, $f : \forall\kappa . \capp C \to \capp A$,
and $g : \forall\kappa . \capp C \to \capp B$.
%
Then the canonical map
\[
(\forall \kappa. \, A)
\sqcup_{(\forall \kappa. \, C)}(\forall \kappa. \, B)
\to \forall\kappa. ((\capp A)\sqcup_{(\capp C)} (\capp B))
\]
is an equivalence.
%
%

\textbf{Finite Powerset}.
A full account of the constructors for the finite powerset is not enlightening,
so we instead exemplify the induction principle in the cases for singleton, union,
and idempotence constructors:
\begin{align*}
&\hastype{\Gamma, x: \forall \kappa . \capp A} {u_{\{-\} }} {D[\lambda \kappa. \{\capp x\}]}\\
&\hastype{\Gamma, x,x' : \forall \kappa. \Pfin(\capp A), y : D(x), y' : D(x')}
{\\  & \qquad \qquad \qquad \qquad u_{\cup}(x,x',y,y')}{D(\lambda \kappa. \capp x \cup \capp {x'})}\\
&\hastype {\Gamma, x : \forall \kappa. \Pfin(\capp A), y : D(x), i:\I}
{\\ & u_{\mathsf{idem}}(x,y,i)\!} {{D(\lambda\kappa . \mathsf{idem}(\capp x,i))
\left[\arraycolsep=1.4pt 
	\begin{array}{l c l}
		\!(i=0) \!&\mapsto& u_{\cup}(x,x,y,y) \! \\
		\!(i=1)\! &\mapsto& y
    \end{array}
\right]}}
\end{align*}

In addition it will
contain the hub and spoke constructors derived from $\sphere 1$, towards obtaining a set type.
Finally, this principle shows that there is an equivalence $\Pfin(\forall \kappa. A)\simeq
\forall \kappa. \Pfin(A)$ as required for the encoding of labelled transition systems.
The concrete calculations are analogous to those for the pushout or trunctation
depending on the constructor.

\end{example}


As a step towards reintroducing clock-irrelevance as an axiom, one would need to
show that the collection of clock irrelevant types is closed under HIT formation:

\begin{theorem}\label{prop:cirr:HITs}
Let $\istype{\delta : \Delta}{\HIT H \delta}$ be a higher inductive type with constructors
\[\con_\ell : (\gamma \!:\! \Gamma_\ell[\delta])
           (x \!:\! \Xi_\ell [\delta, \gamma] \to \HIT H \delta)
           (r \!:\! \Psi_\ell) \to \HIT H \delta [\phi_\ell \mapsto e_\ell]\]
Then $\HIT H {\delta}$ is clock irrelevant if $\Gamma_i[\delta]$ is clock irrelevant
for all $\ell$.
\end{theorem}

As a special case $\nats$ is clock-irrelevant as needed for the encoding of streams of $\nats$.

\section{Denotational semantics}
\label{sec:model}

Previous work has defined a model of \ctt\ in presheaves over a cube category
\cite{CTT, OPAx} and a model of \clott\ in covariant presheaves over a category
of time objects $\TimeCat$~\cite{clottmodel}. We combine these
by considering the category $\PSh{\CubeCat \times \TimeCat}$ of
covariant presheaves over $\CubeCat \times \TimeCat$,
where $\CubeCat$ is the opposite category of the usual choice of cube
category. Let $\Univ$ be a Hofmann-Streicher universe~\cite{Hofmann-Streicher:lifting} in
$\PSh{\CubeCat \times \TimeCat}$.



Following the approach of~\citet{LOPS}, we construct the model using the internal type theory
of $\PSh{\CubeCat \times \TimeCat}$.
For this, we must provide an interval object $\sI$ and a cofibration object $\sF$ satisfying certain axioms. 
We refer to this data as a \emph{cubical model}, and recall that 
presheaves over $\CubeCat \times \mathcal{D}$
admits a cubical model~\cite{sheafModels20} for any small $\mathcal{D}$, where $\sI$ and 
$\sF$ are taken to be the inclusion of those
from $\PSh{\CubeCat}$.
%
%

Recall~\cite{OPAx} that a \emph{CCHM fibration} $(A,\alpha)$ over a context $\Gamma : \Univ$ is a family
$A :\Gamma \to \Univ$ equipped with a fibration structure $\alpha$.
The denotational semantics is given by the category with families (CwF)~\cite{dybjer1996} $\fibCwF$ constructed as
follows:
%
Contexts are global elements of $\Univ$, families are global
CCHM fibrations, and elements of $(A,\alpha)$ are elements of $A$. 

To model HITs
we follow~\citet{CTTHITS}. 
The principle of induction under clocks is justified semantically using that
$\forall \kappa.(-)$
can be modelled as an $\opcat\omega$-limit~\cite{GDTTmodel} and that the structure maps of this limit
cannot change the outermost constructor in HITs, allowing for an inductive proof.


To model a Fitch-style modality like $\later$ we must define a dependent adjunction \cite{drat}.
%
A dependent adjunction consists of
an endofunctor $L$ and an action on types $R$ such that  elements of
$\istype{L(\Gamma)}{A}$ correspond bijectively to elements of $\istype{\Gamma}{R(A)}$.
Both $R$ and the bijective correspondence are required to be natural in $\Gamma$. 
\citet{clottmodel} define a dependent right adjoint on the slice of $\PSh{\TimeCat}$
over an object of clocks in that model, capturing the dependence of $\later$ on a clock. 
This construction extends to $\PSh{\CubeCat\times\TimeCat}$. To extend 
to $\fibCwF$ we show how to lift CCHM fibration structures from $A$ to $R(A)$ by a general construction. 

\begin{definition}
Let $\cat C$ be a cubical model, and let $L$ be an endofunctor on $\cat C$ which preserves
finite limits. We say that $L$ \emph{preserves the interval} if there is an isomorphism
$L(\sI) \cong \sI$ which preserves endpoints. We say that $L$ \emph{preserves cofibrations}
if it maps cofibrations to cofibrations and preserves pullbacks where the vertical
maps are cofibrations.
\end{definition}

Writing $\compr\Gamma A$ for the comprehension object for the family $A$ over $\Gamma$,
and $[\phi]$ for the family classified by a cofibration $\phi$, 
the condition of preserving cofibrations can be presented in the internal language as an operation
mapping cofibrations $\hastype{\Gamma}{\phi}{\sF}$ on $\Gamma$ to
cofibrations $\hastype{\L\Gamma}{\LCof(\phi)}{\sF}$  such that
$\compr{\L\Gamma}{[\LCof(\phi)]} \cong \L(\compr{\Gamma}{[\phi]})$ as
subobjects of $\L\Gamma$ and satisfying $\LCof(\phi[\sigma]) = \LCof(\phi)[\L\sigma]$.

\begin{theorem} \label{thm:dep:right:adj}
 Let $\cat C$ be a cubical model, let $\L : \cat C \to \cat C$
 be a functor preserving finite products, the interval and cofibrations, and let $\R$ be
 a dependent right adjoint to $\L$. If a family $A$ over $\L \Gamma$ carries a global
 composition structure, so does $\R A$ over $\Gamma$. Moreover, this assignment
 is natural in $\Gamma$.
\end{theorem}

Note that this is a statement about \emph{global} composition structures in the model.
The theorem can not be proved in the internal logic of the topos, but can be proved
in an extension of this using crisp type theory, similarly to
the construction of universes for cubical type theory~\cite{LOPS}.
The reason is that the proof uses the bijective correspondence of
dependent right adjoints which only applies to global terms.

\begin{proofsketch}
 A CCHM composition structure on $R(A)$ corresponds to an
 assignment of any $\sigma : \Delta \times \sI \to \Gamma$ to a rule   
 \[
   \inferrule*
  {
   \hastype{\Delta}{\phi}{\sF}
   \and
   \hastype{\compr{\compr{\Delta}{\sI}}{[\phi[\p]]}}{u}{(\R A)[\sigma\p]}
   \\
   \hastype{\Delta}{u_e}{(\R A)[\sigma \!\circ\! \pair{\id}e]}
   \and \!\!\!\!\!
   \hasnotype{\compr{\Delta}{[\phi]}}{u_e[\p] = u[\compr{\pair{\id_{\Delta}}{e}}{[\phi]}]}{}
  }
  {\hastype{\Delta}{c^{\R A}_{\sigma}\, \phi\, u\, u_e}{(\R A)[\sigma \circ \pair{\id}{1-e}]}}
 \]
 for $e=0,1$, which is natural in $\Delta$  and satisfies
 \[\hasnotype{\compr{\Delta}{[\phi]}}{(c_{\sigma}\, \phi\, u\, u_e)[\p] = u[\compr{\pair{\id}{1-e}}{[\phi]}]}{}
 \]
 Using $(\R A)[\sigma\p] = \R(A[\L(\sigma\p)])$ the assumptions can be transported
 along the bijective correspondence of dependent right adjoints to give terms
\begin{align*}
    \hastype{\L(\compr{\compr{\Delta}{\sI}}{[\phi[\p]]})}{\overline u}{A[\L(\sigma\p)]}
   & &
   \hastype{\L\Delta}{\overline{u_e}}{A[\L(\sigma \circ \pair{\id}e)]}
\end{align*}
Since $\L$ preserves cofibrations 
the domain of $\overline u$ is isomorphic to 
$\compr{\compr{\L(\Delta)}{\sI}}{[\LCof(\phi[\p])]}$, so
this data can be transformed into input data for the composition operation on $A$. 
The result of this composition can then be transported back along the bijective 
correspondence of the dependent right adjoint to give the output for the composition 
operation on $RA$. 
\end{proofsketch}

\section{Related work}
\label{sec:related}

Guarded Cubical Type Theory~\cite{GCTT} combines Cubical Type Theory
with single-clocked guarded recursion.
While this case is useful for many purposes, it cannot be used to encode
coinductive types.
\citet{mogelbergPOPL2019}
extend Guarded Cubical Type Theory with ticks as in CloTT.
They give a model of this calculus including HITs, and
show that bisimilarity coincides with path equality for a large class
of guarded recursive types. This paper can be seen as an extension of
that with multiple clocks, allowing for these results to be lifted from guarded
recursive types to coinductive types. Note that modelling the multiclocked case
is much more complex than the single clock case. In particular, equipping
$\later$ with a composition structure is much more challenging, because the simple
description of the left adjoint in the single clock case allowed for a simple construction.
The extended language of ticks giving computation rules for clock irrelevance
presented here is also new.


The encoding of coinductive types using guarded recursion was first
described by \citet{atkey13icfp} in the simply typed
setting. Since then a number of dependent type theories have been
developed for programming and reasoning with
these~\cite{Clouston:Programming,Mogelberg14}, of which CloTT is the most
advanced. \citet{bahr2017clocks} prove syntactic properties of
CloTT including strong normalisation and canonicity, but only for a
pure calculus without identity or path types. The model of CloTT
constructed by \citet{clottmodel} considers extensional
identity types, but no cubical features.
Coinductive types can also be encoded using a combination of guarded
recursion and a $\Box$-modality~\cite{Clouston:Programming}. This approach
has not been studied in combination with CTT yet, and
also appears to be less flexible, e.g., it does not seem possible
to define nested coinductive types.

Sized types~\cite{HughesPS96} is another approach to encoding
productivity in types, by annotating (co)inductive types
with sizes indicating a bound on the size of the allowed elements.
Specifically, sized types reduce both termination and productivity to (well-founded) induction on sizes.
They have been extensively studied from the syntactic
perspective~\cite{Abel:Wellfounded,Abel:NBE:sized:types,Sacchini13} but
are not well understood from the perspective of denotational semantics.
Our view
is that sized types are closer to working in the models
of type theory, like the one provided here, and guarded recursion is
a more abstract, principled perspective. This is supported by the model of
guarded recursion in sized types constructed by \citet{Veltri:19}.
To our knowledge sized types have never been used to
solve equations with negative occurrences, which is an important application
of guarded recursion.

The coincidence of bisimilarity and
path equality for streams as coinductive records has been proved
in Cubical Agda~\cite{CubicalAgda}, and it is likely that the proof can be extended to general M-types.
%
\citet{VeltriFSCD2021} proves that bisimilarity implies path equality for the
final coalgebra for the $\Pfin$ using sized types in Cubical Agda.
This coincidence should therefore be seen as
a feature of Cubical Type Theory, rather than guarded recursion.
On the other hand, when proving such results guarded recursion is a powerful
framework for ensuring productivity of definitions, as illustrated by the examples
in this paper.

The final coalgebra for the finite powerset functor can be constructed in set theory
as a limit of an $\omega + \omega$-indexed sequence~\cite{Worrell05}.
This construction has been formalised in Cubical Agda by \citet{VeltriFSCD2021}
using the lesser limited principle of omniscience, a weak choice
principle.
%
From this perspective it is
interesting that our model uses $\omega$-indexed step-indexing only,
and therefore $\GLTS[]$ as constructed in \autoref{sec:coinductive:types} is realised
in the model as an $\omega$-limit. This construction works because
of the formulation of $\Pfin$ as a HIT and because the $\omega$-chain is constructed
\emph{externally}, using judgemental equality. In particular, the counter example
constructed in Proposition~5 of \citep{VeltriFSCD2021} uses an $\omega$-chain
of elements whose projections are only equal up to path equality.

Multimodal dependent type theory \citep{gratzer2020multimodal} is a general framework
for dependent modal type theories parametrised over \emph{mode theories}.
By instantiating this
appropriately, one can recover e.g., the basic modal framework needed for internalising
parametricity in type theory \citep{cavallo2020internal,bernardy2015presheaf}, or
for the combination of $\later$ and $\Box$ used for guarded recursion by
\citet{Clouston:Programming}. Generalising this to multiple clocks seems to require a notion
of dependent mode theory, as for example, the modal operator $\later$ depends on the object of 
clocks.

\section{Conclusion and future work}
\label{sec:conclusion}

We have presented the type theory \cctt, and shown that the principle of induction under clocks
can be used to construct a
rich supply of functors for which coinductive types can be encoded using guarded recursion.
This allows for simple programming with a wide range of coinductive types, including ones constructed
using higher inductive types.
We have seen by example how to prove coincidence
of path equality with bisimilarity for these types.
We believe this type theory is useful not just for coinductive reasoning, but also for reasoning about
advanced programming language features using a form of synthetic guarded domain
theory~\cite{Paviotti:FPC:journal,paviottiPCF}. In fact, an earlier version of \cctt\
has already been used for a semantic proof of applicative simulation being 
a congruence for a lambda calculus with finite non-determinism~\cite{mogelberg2021two}.

We are currently developing a prototype implementation of \cctt\ as an extension of Cubical Agda.
The current implementation\footnote{https://github.com/agda/guarded/tree/forcing-ticks} 
includes clock, ticks, $\later$ (and its composition structure), and fixed points, but not
the principle of induction under clocks. We believe this can be simulated using Agda's rewrite 
features~\cite{cockx:rewrite}, 
but should in the long run be build into Agda's pattern matching. The proof of Theorem~\ref{thm:final:coalg}
has been verified in this. 
%

We would also like to prove canonicity for \cctt, building on similar
results for Cubical Type Theory \citep{HCanCTT,huber2019canonicity,sterling2021normalization}
and Clocked Type Theory \citep{bahr2017clocks}.
For Cubical Type Theory
canonicity is proved for terms in a context of only interval variables.
 \citet{bahr2017clocks} prove that in Clocked Type Theory, terms of type $\Nat$ or
$\Bool$ (the only inductive types considered in that paper) 
in contexts with free clock variables, but no other assumptions, reduce to introduction forms.
Generalising this result to general HITs and contexts also containing interval variables should allow 
%
for terms of type $\forall\bar\kappa . H(\capp[\bar \kappa]\delta)$ to reduce
to a form in which the $\beta$-rule for induction under clocks can be applied. On the other hand,
we believe that it should not be necessary to include free tick variables, and so the only tick that
needs to be considered in reductions is $\tickc$. Our equational rules have been designed with
this in mind.


Unlike this paper, other type theories for multi-clocked guarded
recursion~\cite{GDTT, clottmodel} take clock irrelevance
(\ref{eq:cirr}) as an axiom. This requires that universes be indexed by clock contexts, and the
$\later$ modality be restricted to $\later^\kappa : \Univ_{\Delta} \to \Univ_{\Delta}$ for
$\kappa \in \Delta$, because an unrestricted $\later^{\kappa}$ breaks clock
irrelevance~\cite{GDTTmodel}. \citet{GDTTmodel} show how to
construct such universes 
of 
types that are clock-irrelevant in the sense of the map $A \to \forall\kappa. A$ being an isomorphism,
rather than an equivalence. Future work includes constructing larger
universes of types clock-irrelevant in the more liberal sense used in this paper.


\begin{acks}
%
%
This work was supported by a research grant (13156) from VILLUM FONDEN.
\end{acks}

\bibliographystyle{ACM-Reference-Format}
\bibliography{paper}

\appendix
\section{Appendix}

\subsection{Omitted proofs Section~\ref{sec:CCTT}}
\begin{proofof}{Lemma~\ref{lem:fix:contractible}}
Given $f : \later^\kappa A \to A$, let $p$ be the corresponding proof of (\ref{eq:fix:unfold}), then as mentioned the center of contraction for
\[
\Sigma (x : A). \Path Ax{f\,(\tabs\tickA\kappa{x})}
\]
will be the pair $(\fix^\kappa f, p)$.
Then for any other such pair $(h,p_h)$ we have to show $(\fix^\kappa f, p) = (h,p_h)$, which is equivalent to
\[
\Sigma (q : \fix^\kappa f = h).~q_* p = p_h
\]
We define $q$ by guarded recursion
\[
\fix^\kappa\,\lambda (r : \later^\kappa (\fix^\kappa f = h)). (p^{-1},p_h^{-1})_*(\mathsf{ap}_f(\lambda i. \tabs \tickA\kappa {\tapp r\,i}))
\]
We then proceed to prove $q_* p = p_h$ by first observing that it is equivalent to
\[
(p,p_h)_*q = \mathsf{ap}_f(\lambda i. \tabs \tickA\kappa {\tapp q\,i})
\]
so that by expanding $q$ by (\ref{eq:fix:unfold}) on the left hand side and canceling the transports we obtain the right hand side and the proof is concluded.

To prove that $\Sigma (x : \later^\kappa A). \latbind{\tickA}{\kappa}{\Path{A}{\tapp{x}}{f(x)}}$ is
contractible, note that
if $B$ is contractible, so is $\later^\kappa B$. Applying this to the first part of the lemma, we get that 
the following equivalent types are all contractible:
\begin{align*}
 \later^\kappa & (\Sigma (x : A). \Path Ax{f\,(\tabs\tickB\kappa{x})}) \\
 & \equi \Sigma (x : \later^\kappa A). \latbind \tickA\kappa{\Path A{\tapp x}{f\,(\tabs\tickB\kappa{(\tapp x)})}} \\
 & \equi \Sigma (x : \later^\kappa A). \latbind \tickA\kappa{\Path A{\tapp x}{f\,(\tabs\tickB\kappa{(\tapp[\tickB] x)})}} \\
 & \equi \Sigma (x : \later^\kappa A). \latbind \tickA\kappa{\Path A{\tapp x}{f\,(x)}} 
\end{align*}
where the first equivalence is by (\ref{eq:later:sigma}) and the second by tick irrelevance. 
\end{proofof}

\begin{proofof}{Lemma~\ref{lem:tirr:coherence}}
We have
\[\istick{\Gamma, i_0, \ldots, i_n : \I} {u_0 \defeq u\subst{i_0}{0}} {\kappa}{\Gamma'}\]
and
\[\istick{\Gamma, i_0, \ldots, i_n : \I, \phi, j : \I}{p(j) \defeq \tirr(u_0,u,j)}{\kappa}{\Gamma'}\]
Define $\istype{\Gamma, i_0, \ldots, i_n : \I, j : \I} C$ by
\[\mathsf{hfill}^j\,[\phi \mapsto A\subst \alpha {p(j)}]\,A\subst \alpha {u_0}\]
and then take $B \defeq C\,1$.
Then we define $\mathsf{filler}(t,u)$ as \[\hcomp^j_{C}\,[\phi \mapsto \tappg t {} {p(j)}]\,(\tappg t {} {u_0})\]
\end{proofof}

\subsection{Substitution for Tick Application}
\label{app:subst}
In figure~\ref{fig:substitutions} we present the formation rules for
substitutions, based on the ones from \citet{clottmodel}, and extended
to account for the new tick judgemnts and the contexts from cubical
type theory. In what follows we explain how to apply a substitution to a tick application term.

\begin{figure}
\textbf{Substitutions}
\begin{mathpar}
  \inferrule*{ \wfcxt{\Gamma} }{\issub{\Gamma} {[]} {\cdot}}
  \and
  \inferrule*{\issub {\Gamma} {\sigma} {\Gamma'} \\ \hastype \Gamma t {A\sigma}}
             {\issub \Gamma {(\sigma, t)} {\Gamma', x : A}}
             \and
  \inferrule*{\issub \Gamma \sigma {\Gamma'} \\ \hastype \Gamma {\kappa'} \clocktype}
             {\issub \Gamma {(\sigma, \kappa')} {\Gamma', \kappa : \clocktype}}
             \and
  \inferrule*{\issub \Gamma \sigma {\Gamma'} \\ \hastype \Gamma {r} \I}
             {\issub \Gamma {(\sigma, r)} {\Gamma', i : \I}}
             \and
  \inferrule*{\issub \Gamma \sigma {\Gamma'} \\ \hastype {\Gamma'} {\phi} \F \\ \hastype {\Gamma} {\phi\sigma = 1_\F} {\F}}
             {\issub \Gamma {\sigma} {\Gamma', \phi}}
             \and
  \inferrule*{\issub {\Gamma_0} \sigma {\Gamma'} \\ \istick{\Gamma} u {\kappa\sigma} {\Gamma_0}}
             {\issub \Gamma {(\sigma , u)} {\Gamma', \tickA : \kappa}}
             \and
  \inferrule*{\issub {\Gamma_0} \sigma {\Gamma'} \\ \isbtick{\Gamma} {\kappa'} {v} {\Gamma_0}}
             {\issub \Gamma {\sigma , \btick {\kappa'} {v}} {\Gamma', \kappa : \clocktype, \tickA : \kappa}}
\end{mathpar}
  \caption{Formation rules for substitutions.}
  \label{fig:substitutions}
\end{figure}

\newcommand{\residual}[2]{\mathsf{residual}({#1},{#2})}
\newcommand{\bresidual}[2]{\mathsf{bresidual}({#1},{#2})}

\begin{operation}
Given $\issub {\Delta} {\sigma} {\Gamma}$ and 
$\istick {\Gamma} {u} {\kappa} {\Gamma'}$ we 
define an operation $\residual \sigma u$ returning tuples of one of two forms
\begin{itemize}
\item $(\Delta',\sigma')$ such that $\istick{\Delta} {u\sigma} {\kappa\sigma} {\Delta'}$ and $\issub {\Delta'} {\sigma'} {\Gamma'}$, and $\issub {\Delta} {\sigma |_{\Gamma'} \jeq \sigma'} {\Gamma'}$.
\item $(\Delta',\kappa'',\sigma')$ such that $\isbtick{\Delta} {\kappa\sigma} {u\sigma} {\Delta'}$ and \[{\issub {\Delta', \kappa'' : \clocktype} {\sigma'} {\Gamma'}},\] such that $\kappa\sigma' = \kappa''$ and $\issub {\Delta} {\sigma |_{\Gamma'} \jeq \sigma'\subst {\kappa''} {\kappa\sigma}} {\Gamma'}$.
\end{itemize}
\end{operation}
\begin{construction}
The tick $u$ contains tick variables $\tickA_0 : \kappa \ldots \tickA_n : \kappa$, here given in the order they appear in $\Gamma$, so in particular we have $\Gamma = \Gamma_1, \tickA_0 : \kappa, \Gamma_2$ and $\Gamma' \subcxt \Gamma_1, \cuba{\Gamma_2}$.
Then let us look at the restriction of $\sigma$ to $\Gamma_1, \tickA_0 : \kappa$.
We have two cases:
\begin{itemize}
\item[(i)] $(\sigma_1,v)$ with $\issub {\Delta'} {\sigma_1} {\Gamma_1}$ and $\istick {\Delta_0} v {\kappa\sigma_1} {\Delta'}$ with $\Delta_0 \subcxt \Delta$.
\item[(ii)] $(\sigma_0,(\kappa',v))$ with $\Gamma_1 = \Gamma_0, \kappa : \clocktype$, and $\issub {\Delta'} {\sigma_0} {\Gamma_0}$, and  $\isbtick {\Delta_0} {\kappa'} {v} {\Delta'}$, with $\Delta_0 \subcxt \Delta$.
\end{itemize}
In either case the tick $v$ lives in the context $\Delta_0 \subcxt \Delta$ because other components of the substitution $\sigma$, e.g. for $\tickA_1 \ldots \tickA_n$, might have shrunk the context so.

In case (i) we extend $\sigma_1$ to $\issub {\Delta'} {\sigma_1'} {\Gamma_1, \cuba{\Gamma_2}}$ because whenever we have one of $\hastype {\Delta} {\kappa_i} {\clocktype}$, or $\hastype {\Delta} {r} {\I}$, or $\hastype {\Delta} {\phi \jeq 1_\F} {\F}$ we also have the same in $\Delta' \subcxt \Delta$. Finally we take $\sigma'$ to be $\sigma_1' |_{\Gamma'}$, which agrees with $\sigma |_{\Gamma'}$ by construction. The constructed tuple will be $(\Delta',\sigma')$.

In case (ii) we extend $\sigma_0$ first to $\issub {\Delta', \kappa''} {(\sigma_0, \kappa'')} {\Gamma_0, \kappa : \clocktype}$ then to a substitution $\issub {\Delta',\kappa''} {\sigma_0'} {\Gamma_1, \cuba{\Gamma_2}}$ as above, furthermore noting that $\cuba{\Gamma_2}$ does not depend on $\kappa$. Finally we take $\sigma'$ to be $\sigma_0' |_{\Gamma'}$. The substitution $\sigma'\subst {\kappa''} {\kappa'}$ then agrees with $\sigma |_{\Gamma'}$ by construction, and we also have $\kappa\sigma' \jeq \kappa''$. The constructed tuple will be $(\Delta',\kappa'',\sigma')$.

To show that we have the correct typing for $u\sigma$ we observe that $\Delta'$ is smaller as a subcontext of $\Delta$ than the ones the ticks $\tickA_1\sigma \ldots \tickA_n\sigma$ target, so they can all be weakened to fit the typing $\istick {\Delta} {\tickA_i\sigma} {\kappa\sigma_1} {\Delta'}$. In case (i) then we are done by extending this observation to $v$. In case (ii) we can further derive $\isbtick {\Delta} {\kappa\sigma_1} {\tickA_i\sigma} {\Delta'}$ which gives us what we want.
\end{construction}

For $\issub {\Delta} {\sigma} {\Gamma}$, and $\hastype{\Gamma'} t {\latbind \tickA \kappa \, A}$, and 
$\istick {\Gamma} {u} \kappa {\Gamma'}$, we have that substitution commutes with tick application in the 
sense that $(\tappg t {\Gamma'} {u})\sigma$ equals 
\[
  \left\{ \begin{array}{l l}
                              \tappg {t\sigma'} {\Delta'} {u\sigma} & \text{if } \residual{\sigma}{u} = (\Delta',\sigma')\\
                              \tappg {(\kappa''.\, t\sigma')} {\Delta'} {\btick{\kappa\sigma} {u\sigma}}  & \text{if } \residual{\sigma}{u} = (\Delta',\kappa'',\sigma')
                              \end{array}
                       \right.
\]
We want to show that typing is preserved. 
For the first case, we can assume $\hastype {\Delta'} {t\sigma'} {\latbind \tickA {\kappa\sigma'} A(\sigma',\tickA)}$, so by the tick application rule we have $\hastype {\Delta} {\tappg {t\sigma'} {\Delta'} {u\sigma}} {A(\sigma',\tickA)\subst \tickA {u\sigma}}$, where the latter type is equal to $A\subst \tickA u\sigma$ as expected.
For the second case, we can assume $\hastype {\Delta', \kappa''} {t\sigma'} {\latbind \tickA {\kappa\sigma'} A(\sigma',\tickA)}$, so by the forcing tick application rule we have 
\[\hastype {\Delta} {\tappg {(\kappa''.\, t\sigma')} {\Delta'} {\btick {\kappa\sigma} {u\sigma}}} 
       {A(\sigma',\tickA)\subst {\tickA : \kappa''} {u\sigma : \kappa\sigma}},\] where the latter type is equal to $A\subst \tickA u\sigma$ as expected.

\begin{operation}
Given $\issub {\Delta} {\sigma} {\Gamma}$ and 
$\isbtick {\Gamma} {\kappa} {u} {\Gamma'}$ we 
define an operation $\bresidual \sigma {\btick {\kappa} u}$ returning tuples of the form
\begin{itemize}
\item $(\Delta',\sigma')$ such that $\isbtick{\Delta} {\kappa\sigma} {u\sigma} {\Delta'}$ and $\hastype {\Delta'} {\sigma'} {\Gamma'}$, and $\hastype{\Delta} {\sigma |_{\Gamma'} \jeq \sigma'} {\Gamma'}$
\end{itemize}
\end{operation}
\begin{construction}
Here the tick $u$ might not contain any tick variables, in which case we take $\Delta'$ to be $\Delta$ and $\sigma'$ to be $\sigma |_{\Gamma'}$.
If $u$ does contain tick variables then we have cases (i) and (ii) as in the construction of $\residual{-}{-}$. We chose to use $\kappa$ in the assumption $\isbtick {\Gamma} {\kappa} {u} {\Gamma'}$ so that the names in the cases would line up with the previous construction.

In case (i) we construct the tuple $(\Delta',\sigma')$ as we did before, and the same reasoning extends to the well-typing of ${\isbtick{\Delta} {\kappa\sigma} {u\sigma} {\Delta'}}$.

In case (ii) let us recall that here $\Gamma$ is of the form $\Gamma_0, \kappa : \clocktype, \tickA_0 : \kappa, \Gamma_2$, with $\Gamma' \subcxt \Gamma_0, \kappa : \clocktype, \cuba{\Gamma_2}$. We also have $\issub {\Delta'} {(\sigma_0,\kappa')} {\Gamma_0, \kappa : \clocktype}$, which agrees with $\sigma$ restricted to the same context. As before we can extend $(\sigma_0,\kappa')$ to a substitution for $\Gamma_0, \kappa : \clocktype, \cuba{\Gamma_2}$ by using the relevant components of $\sigma$, and finally obtain the desired $\sigma'$ by restriction to $\Gamma'$.
\end{construction}

For $\issub {\Delta} {\sigma} {\Gamma}$, and $\hastype{\Gamma', \kappa : \clocktype} t {\latbind \tickA \kappa A}$, and
$\isbtick {\Gamma} {\kappa'} {u} {\Gamma'}$, we have that substitution commutes with forcing tick application in the following way:
\[
  (\tappg {(\kappa.t)} {\Gamma'} {\btick {\kappa'} u})\sigma = 
                              \tappg {(\kappa.\, t(\sigma',\kappa))} {\Delta'} {\btick{\kappa'\sigma} {u\sigma}} 
\]
where $\bresidual{\sigma}{\btick {\kappa'} u} = (\Delta',\sigma')$. 
We want to show that typing is preserved. We can assume 
\[\hastype {\Delta', \kappa : \clocktype} {t(\sigma',\kappa)} {\latbind \tickA \kappa \,A(\sigma',\kappa,\tickA)},\] 
then by the forcing tick application rule we have 
\[\hastype {\Delta} {\tappg {(\kappa.\, t(\sigma',\kappa))} {\Delta'} {\btick{\kappa'\sigma} {u\sigma}}} {A(\sigma',\kappa,\tickA)\subst {\tickA : \kappa} {u\sigma : \kappa'\sigma}},\] 
where the latter type is equal to $A\subst {\tickA : \kappa} {u : \kappa'}\sigma$ as expected.

\subsection{Omitted proofs Section~\ref{sec:coinductive:types}}

\begin{proofof}{Lemma~\ref{lem:comm:kappa}}
 The case of composition is clear, and products follow from the
 fact that $\forall\kappa . (A \times B) \equi (\forall\kappa .A) \times (\forall\kappa . B)$.
 Likewise, the case of $\Pi$-types follows from the fact that
 $\forall\kappa . \Pi(a:A) .B \equi \Pi(a:A) .\forall\kappa .B$ which can be
 proved by commuting two arguments. In the case of $\Sigma$,
 since $\forall\kappa. (-)$ behaves as a function space from a type
 of clocks, one can prove
 \begin{align*}
\forall\kappa . \Sigma (a :A) . B(a)
 & \equi \Sigma (a : \forall\kappa . A) . \forall\kappa . B(a[\kappa]) \\
 & \equi \Sigma (a : A) . \forall\kappa . B(a)
\end{align*}
 using the assumption that $A$ is clock invariant in the last equivalence.
 The case for universal quantification over clocks uses
 $\forall\kappa. \forall\kappa'. A \equi\forall\kappa'. \forall\kappa. A$.

 In the case of guarded recursive types, first note that if $F$  commutes
 with clock quantification, so does $\later^\kappa F$. This can be proved
 using $\forall\kappa'.\later^\kappa A \equi \later^\kappa \forall\kappa'.A$,
 the left to right map of which maps $a$ to
 \[
 \tabs\tickA\kappa{\lambda\kappa'. \tapp{a[\kappa']}}
 \]
 for $\tickA$ fresh. This map type checks because
 $\cuba{\kappa' : \clocktype} = \kappa' : \clocktype$.
 Using this, we can prove by guarded recursion that
 $X$ is clock irrelevant as follows
\begin{align*}
 \forall\kappa' . X & \equi \forall\kappa' . F(\later^\kappa X) \\
 & \equi F(\forall\kappa' . \later^\kappa X) \\
 & \equi F(\later^\kappa \forall\kappa' . X) \\
 & \equi F(\later^\kappa X)
\end{align*}
using the guarded recursion assumption in the last line.

In the case of path types, if $x,y : A$ then
\begin{align*}
 \forall\kappa. (\Path A x y)
 & \equi \Path {\forall\kappa. A}{\lambda \kappa . x}{\lambda \kappa. y} \\
 & \equi \Path {A}{x}{y}
\end{align*}
The first of these equivalences simply swaps the clock and interval
argument, the second uses the assumption that $A$ is clock invariant, which
means precisely that $\lambda a. \lambda\kappa .a$ is an equivalence,
and so preserves path types.
\end{proofof}

We now give a proof of Theorem~\ref{thm:final:coalg}.
It uses the following lemma establishing the existence of a final
$F\circ\later^\kappa$-coalgebra
for any endofunctor $F$.

\begin{lemma}\label{lemma:guarded:final}
  Let $F$ be an $I$-indexed endofunctor, then for all $\kappa$, the type
  $\nu^\kappa(F) \defeq \fix^\kappa\,(\lambda X. F(\later^\kappa X)) : I \to \Univ$ has a final $F \circ \later^\kappa$-coalgebra structure, i.e., there is a map $\out^\kappa : \nu^\kappa(F) \to F\,(\later^\kappa\,\nu^\kappa(F))$, such that for all maps $f : X \to F (\later^\kappa\,X)$ the following type is contractible
  \[
  \Sigma (h : X \to \nu^\kappa(F)).\, \out^\kappa \circ h \peq F(\later^\kappa(h)) \circ f
  \]
\end{lemma}
\begin{proof}
The map $\out^\kappa$ is given by the equivalence between a fixpoint in the universe and its unfolding, so we also have an inverse ${(\out^\kappa)}^{-1}$. The functor $\later^\kappa$ is locally contractible~\cite{ToT} in the sense that the
action on morphisms factors as a composition of two maps
\[
 (X \to Y) \to \later^{\kappa}(X \to Y) \to (\later^\kappa X \to \later^\kappa Y)
\]
and so also the mapping $\lambda h . {(\out^\kappa)}^{-1} \circ F(\later^\kappa(h)) \circ f$ factors as a
composition
\[
(X \to \nu^\kappa(F)) \to \later^\kappa(X \to \nu^\kappa(F)) \to (X \to \nu^\kappa(F))
\]
Then by uniqueness of fixpoints (Lemma~\ref{lem:fix:contractible}) we get the contractibility of
\[
\Sigma (h : X \to \nu^\kappa(F)).\, h \peq {(\out^\kappa)}^{-1} \circ F(\later^\kappa(h)) \circ f
\]
which in turn is equivalent to our goal by postcomposition with $\out^\kappa$.
\end{proof}

\begin{proofof}{Theorem~\ref{thm:final:coalg}}
  We define the coalgebra $\out : \nu(F) \to F\,\nu(F)$ as $F(\force) \circ \mathsf{can}_F^{-1} \circ \forall \kappa(\out^\kappa)$.
  Given any coalgebra $f : X \to F\,X$, we can extend it to $\tilde{f}^\kappa : X \to F\,(\later^\kappa X)$ for any $\kappa$.
  Then by lemma~\ref{lemma:guarded:final} we have that for any $\kappa$ the type
  \[
  \Sigma (h : X \to \nu^\kappa(F)).\, \out^\kappa \circ h \peq F(\later^\kappa(h)) \circ \tilde{f}^\kappa
  \]
  is contractible. By clock quantification preserving contractibility and commuting with $\Sigma$ types we have that
  \[
  \Sigma (h : \forall \kappa.\,X \to \nu^\kappa(F)).\, \forall\kappa.\,
  (\out^\kappa) \circ \capp h \peq F(\later^\kappa(\capp h)) \circ \tilde{f}^\kappa
  \]
  is also contractible. Then by $\forall \kappa.\,X \to \nu^\kappa(F) \simeq X \to \nu(F)$ and that
  $(\lambda \kappa.\, F(\later^\kappa(\capp h)) \circ \tilde{f}^\kappa) \peq \mathsf{can}_F \circ F(\force^{-1}) \circ F(h)$
  we obtain the desired result. More details on the calculations can be found in \cite{Mogelberg14}.
  \end{proofof}

Example~\ref{ex:glts} uses the following lemma.

\begin{lemma}
 Let $A$ be a clock irrelevant set and let $X : \Pfin(A)$, $a : A$. Then $a \in X$ is clock
 irrelevant.
\end{lemma}

\begin{proof}
 The proof is by induction on $X$, which is valid since statements of the form $\mathsf{IsEquiv}(f)$
 are propositions. If $X = \{b\}$ then $a \in X$ is by definition $\Path A ab$, which is
 clock irrelevant by Lemma~\ref{lem:comm:kappa}. If $X = Y \cup Z$ then
 $a \in X$ is $(a \in Y)\times(a \in Z)$ which is clock-irrelevant by induction and
 Lemma~\ref{lem:comm:kappa}.
\end{proof}

We now give a proof of Theorem~\ref{thm:bisim:is:identity}. We will write
\[
  \unfoldG[\kappa] : \GLTS[\kappa] \to \Pfin(A \times \later^\kappa \GLTS[\kappa])
\]
for the equivalence associated with the guarded recursive type $\GLTS[\kappa]$.
First note the following.

\begin{lemma} \label{lem:bisim:is:identity:comm:diag}
 The following diagram commutes up to path equality.
     \[
      \begin{tikzcd}
        \GLTS \ar{r}{\unfoldG} \ar{d}[swap]{\ev{\kappa}} & \Pfin(A \times \GLTS) \ar{d}{\Pfin(A \times f)} \\
        \GLTS[\kappa] \ar{r}{\unfoldG[\kappa]} & \Pfin(A \times \later^\kappa\GLTS[\kappa])
      \end{tikzcd}
    \]
  where $\ev\kappa \defeq \lambda x . x[\kappa]$ and $f = \lambda x . \tabs\tickA\kappa{x[\kappa]}$.
\end{lemma}

\begin{proof}
 The map $\unfoldG$ is defined to be the composition of the following maps
\begin{align*}
 \forall\kappa . \unfoldG[\kappa] & : \forall\kappa . \GLTS[\kappa]
    \to \forall\kappa . \Pfin(A \times \later^\kappa\GLTS[\kappa]) \\
  \phi & : \forall\kappa . \Pfin(A \times \later^\kappa\GLTS[\kappa])  \to
    \Pfin(A \times \forall\kappa . \later^\kappa\GLTS[\kappa]) \\
  \Pfin(A \times \force) & : \Pfin(A \times \forall\kappa . \later^\kappa\GLTS[\kappa]) \to
  \Pfin(A \times \forall\kappa . \GLTS[\kappa])
\end{align*}
 where $\forall\kappa . \unfoldG[\kappa] (x) = \lambda \kappa. \unfoldG[\kappa](x[\kappa])$, $\phi$
 is the witness that $\Pfin(A \times (-))$ commutes with clock quantification, and
 $\force : \forall\kappa . \later^\kappa\GLTS[\kappa] \to \forall\kappa . \GLTS[\kappa]$
 is the inverse to $\lambda x. \tabs\tickA\kappa x$ up to path equality. By the latter, we get
 the following sequence of path equalities
\begin{align*}
 \Pfin(A \times f) \circ \unfoldG & = \Pfin(A \times \ev\kappa) \circ \phi \circ \forall\kappa . \unfoldG[\kappa] \\
  & = \ev\kappa \circ \forall\kappa . \unfoldG[\kappa] \\
  & = \unfoldG[\kappa] \circ \ev\kappa
\end{align*}
 as desired.
\end{proof}

\begin{proofof}{Theorem~\ref{thm:bisim:is:identity}}
M{\o}gelberg and Veltri~\cite{mogelbergPOPL2019} prove that path equality
coincides with bisimilarity for guarded recursive types. Using their results
we can prove that, given $x,y : \GLTS$
\begin{align}
 \Path \GLTS x y
 & \equi \forall\kappa . \Path {\GLTS[\kappa]}{x [\kappa]}{y[\kappa]} \nonumber \\
 & \equi \forall\kappa .\Bisim(x [\kappa], y[\kappa]) \label{eq:bisim}
\end{align}
where the first equivalence uses functional extensionality for universal quantification over clocks
and the second is \cite[Corollary~5.4]{mogelbergPOPL2019}. Here
$\Bisim(x,y) = \Sim(x,y)\times \Sim(y,x)$ where
\begin{align*}
 \Sim(x,y) \equi \, & \Pi(x' : \later^\kappa\GLTS[\kappa], a : A) . (a, x')\in \unfoldG[\kappa](x) \to \\
 & \phantom{\Pi} \exists y' : \later^\kappa\GLTS[\kappa] . (a,y') \in \unfoldG[\kappa](y) \times \\
 & \phantom{\Pi}\latbind\tickA\kappa \Sim(\tapp{x'}, \tapp{y'})
\end{align*}

We must compare this to bisimilarity of $x$ and $y$ which is defined as
 \[
   \forall\kappa . (\Simfix(x,y) \times  \Simfix(y,x))
 \]
 where
\begin{align*}
 \Simfix(x,y)  \equi\, & \Pi(x' : \GLTS, a : A) . (a, x')\in \unfoldG(x) \to \\
 & \phantom{\Pi} \exists y' : \GLTS . (a,y') \in \unfoldG(y) \times \\
 & \phantom{\Pi}  \latbind\tickA\kappa{\Simfix(x', y')}
\end{align*}
By an easy guarded recursive argument one can show that
$\Simfix$ is a reflexive relation, and from this it follows that
path equality implies bisimilarity. To prove the other implication
it suffices to show that
\[
  \Pi(x,y : \GLTS) . \Simfix(x,y) \to \Sim(x[\kappa],y[\kappa])
\]
and this statement is proved by guarded recursion.
So suppose $x,y : \GLTS$ and $\Simfix(x,y)$. Suppose
further that $x' : \later^\kappa\GLTS[\kappa], a : A$ and
$(a, x')\in \unfoldG[\kappa](x[\kappa])$. By
Lemma~\ref{lem:bisim:is:identity:comm:diag}
this means that $(a,x') \in \Pfin(A \times f)(\unfoldG(x))$
where $f = \lambda x . \tabs\tickA\kappa{x[\kappa]}$.
By~\cite[Lemma~4.1]{mogelbergPOPL2019} there then
(merely, i.e. in the sense of $\exists$) exists an $x'' : \GLTS$
such that $x' = f(x'')$, i.e.,
\begin{equation} \label{eq:bisim:proof:1}
x' = \tabs\tickA\kappa{(x''[\kappa])}
\end{equation}
and $(a,x'') \in \unfoldG(x)$. By the assumption that $\Simfix(x,y)$
there then merely exists a $y'' : \GLTS$ such that
$(a,y'') \in \unfoldG(y)$ and
\begin{equation} \label{eq:bisim:proof:2}
\latbind\tickA\kappa{\Simfix(x'', y'')}.
\end{equation}
Setting $y' = f(y'')$ then, again
by~\cite[Lemma~4.1]{mogelbergPOPL2019}
$(a,y') \in  \Pfin(A \times f)(\unfoldG(y))$ and so by
Lemma~\ref{lem:bisim:is:identity:comm:diag},
$(a,y') \in \unfoldG[\kappa](y[\kappa])$. It remains
to show that
\[
\latbind\tickA\kappa \Sim(\tapp{x'}, \tapp{y'})
\]
which reduces to
\[
\latbind\tickA\kappa \Sim(x''[\kappa], y''[\kappa])
\]
using (\ref{eq:bisim:proof:1}) and definition of $y'$.
This follows by guarded recursion
from (\ref{eq:bisim:proof:2}).
\end{proofof}

\subsection{Omitted proofs \autoref{sec:hits}} \label{app:hits}

\begin{figure*}
\textbf{Rules for equality of boundary terms}
\begin{mathpar}
    \inferrule*{\hasnotype\Gamma{\bar{u} \jeq \bar{v}}{}}{\hasnotype\Gamma{x~\bar{u} \beq x~\bar{v}}{}} \and
    \inferrule*{\hasnotype\Gamma{\bar{t} \jeq \bar{t'}}{} \\ \hasnotype\Gamma{\bar{r} \jeq \bar{r'}}{} \and
    \forall i. (\hasnotype{\Gamma, \bar{ \xi_i}}{M_i \beq M'_i}{})}
    {\hasnotype{\Gamma}{\con_{\ell}(\bar{t},\overline{\lambda \xi.\,M},\bar{r})
    \beq \con_{\ell}(\bar{t'},\overline{\lambda \xi.\,M'},\bar{r'} )}{}} \and
    \inferrule*{\hasnotype{\Gamma, j, \phi}{M \beq M'}{} \\ \hasnotype\Gamma{M_0 \beq M_0'}{}
    }{
    \hasnotype{\Gamma}{\hcomp_{\HIT{H}\delta}^j~[\phi \mapsto M]~{M}_0 \beq
    \hcomp_{\HIT{H}\delta}^j~[\phi \mapsto M']~{M}_0'}{}} \and
    \inferrule*{ e_\ell = [\varphi_1~N_1 \ldots \varphi_m~N_m] \\
      \hasnotype\Gamma{\phi_i \jeq \top}{}
    }{\hasnotype{\Gamma}{\con_{\ell}(\bar{t},\overline{\lambda \xi.\,M},\bar{r})
    \beq N_i(\bar{t},\overline{\lambda \xi.\,M},\bar{r})}{}
    } \and
    \inferrule*{\hasnotype\Gamma{\phi \jeq \top}{}
    }{\hasnotype\Gamma{\hcomp_{\HIT{H}\delta}^j~[\phi \mapsto M]~{M}_0 \beq  M[1/j]}{}
    }
\end{mathpar}
\textbf{Substitution operation}
\begin{align*}
 (x_j~\bar u)(\bar{t},\overline{\lambda \xi.\,M},\bar{r}) & =
 M_j[\bar u [\bar t, \bar r] / \bar{\xi_j}] \\
 \con_{\ell}(\bar{t'},\overline{\lambda \xi'.\,M'},\bar{r'} )(\bar{t},\overline{\lambda \xi.\,M},\bar{r}) & =
 \con_{\ell}(\bar{t'}[\bar t, \bar r] ,\overline{\lambda \xi'.\,M'(\bar{t},\overline{\lambda \xi.\,M},\bar{r}, \bar{\xi'})},\bar{r'}[\bar r] ) \\
 (\hcomp_{\HIT{H}\delta}^j~[\phi \mapsto N]~{N}_0)(\bar{t},\overline{\lambda \xi.\,M},\bar{r})  & =
  \hcomp_{\HIT{H}\delta}^j~[\phi \mapsto (N(\bar{t},\overline{\lambda \xi.\,M},\bar{r},j))]~({N}_0 (\bar{t},\overline{\lambda \xi.\,M},\bar{r}))
\end{align*}
  \caption{The equational theory of boundary terms is the least equivalence relation generated by the rules above.
  The rule for reducing a constructor to its boundary uses the substitution
  also defined above.}
  \label{fig:eq:boundary}
\end{figure*}

The equational theory of boundary terms is given by $\beq$ defined in Figure~\ref{fig:eq:boundary}. This theory is used in the
typing of boundary terms, when typing systems and homogenous compositions. More precisely, the requirement is that for
a system of boundary conditions $[\varphi_1~M_1 \ldots \varphi_m~M_m]$ to be wellformed, it must be the case that
$\phi_i \wedge \phi_j$ implies $M_i \beq M_j$, for any $i,j$. Similarly, for
$\hcomp^{j}_{\HIT H {\delta}} [\psi \mapsto M'] \, M'_0$ to be well typed, we must have $M'[0/j] \beq M'_0$.

The next lemma states that boundaries are well-typed.

\begin{lemma}\label{lem:evalbsyn:typing}
Let $e = [\phi_0 \, M_0, \dots, \phi_{n_\ell} \, M_{n_\ell}]$ be the boundary condition for
$\con_\ell$. Assume $\isclockelimlist {\Gamma} {\mathcal E} {\mathcal K_{<\ell}} {\delta} {D}$; then
the following typing holds:
\begin{align*}
\hastype {\Gamma, \gamma : \forall \bar\kappa . \Gamma_\ell[\cappp \delta],
			\bar{x} : \bar{\forall \bar\kappa . \Xi_\ell[\cappp \delta, \cappp \gamma]\to \HIT H {(\cappp \delta)}},&
			\\ \bar{y} : \bar{\Pi(\xi: \forall \bar\kappa. \Xi[\cappp \delta, \cappp \gamma])
			. D[\lambda \kappa . \cappp x (\cappp \xi)]}, \bar{i} : \Psi_\ell, \phi_\ell &}
			{ \\ \evalbsynclock{\delta}{\mathcal{E}}{\bar{x} \mapsto \bar{y}}{e} }
			{D[\lambda \bar\kappa. \con_\ell(\cappp \gamma, \bar{\capp[\bar\kappa] x}, & \bar{i})]}
\end{align*}
\end{lemma}

\begin{proof}
Throughout we assume that $\hastype{\Gamma}{\delta}{\forall \kappa . \Delta}$ unless otherwise specified.
Through an induction on the structure of the boundary terms, we prove the
following more general typing:
\begin{align*}
\hastype {\Gamma', \hat\gamma : \forall \bar\kappa.\,\hat\Gamma[\cappp \delta, \capp \gamma]&}
			{ \evalbsynclockg{\delta}{\mathcal{E}}{\bar{x} \mapsto \bar{y}}{\hat \gamma}{M}\\ &}
			{D[\lambda \bar \kappa. M[\cappp \gamma, \bar{\cappp x}, \bar{i},\cappp {\hat \gamma}]]}
\end{align*}
where $\Gamma'$ is defined as follows
\begin{align*}
\Gamma' \defeq & \Gamma, \gamma : \forall \bar\kappa . \Gamma_i[\cappp \delta],
			\bar{x} : \bar{\forall \bar\kappa . \Xi_i[\cappp \delta, \cappp \gamma]\to \HIT H {(\cappp \delta)}},
			\\ & \bar{y} : \bar{(\xi: \forall \bar\kappa. \Xi[\cappp \delta, \cappp \gamma])
			\to D[\lambda \bar \kappa. \cappp x (\cappp \xi)]}, \bar{i} : \Psi_i, \phi_i
\end{align*}
and $M$ is assumed to be a boundary term of type $\HIT{H}{\delta}$ in context
\begin{align*}
\delta : \Delta, \gamma : \Gamma_i[\delta],
			\bar{x} : \Xi_i[\delta, \gamma]\to \HIT H {\delta},
			\bar{i} : \Psi_i, \phi_i, \hat\gamma : \hat\Gamma[\delta,\gamma]
\end{align*}
The desired typing is then the case where $\hat \Gamma$ above is empty, since we
have that $\phi_k$ implies that $\lambda \bar\kappa. \con_i(\ldots)$ reduces to $\lambda \bar\kappa.\,M_k$.
The $\hat \Gamma$ crops up because the interpretation of constructors adds a $\xi: \forall
\bar\kappa. \Xi[\cappp \delta, \cappp \gamma]$ to the context while the $\hcomp$ case adds an
interval variable and a face restriction
to the context for the system of the composition. In the very first step,
unfolding the list of partial elements adds a face restriction as well.
Concretely we proceed as follows in each case.
Also, we write $\sigma$ for the substitution $[\cappp \delta/\delta,\cappp \gamma/\gamma,\bar{\cappp x}/\bar x, \cappp {\hat\gamma}/\hat\gamma]$ and $\tau$ for the same but without the $\bar{x}$ component.
\begin{itemize}
\item
Assume $M = x_j \, \bar{u}$. From the typing assumptions on $M$ we have that
\begin{align*}
\hastype{\delta : \Delta, \gamma : \Gamma_i, \bar{x} : \bar{\Xi_i[\delta, \gamma]
\to \HIT H {\delta}}, \bar{i} : \Psi,\phi_i, \hat\Gamma}{ u }{\Xi_{i,j}}
\end{align*}
This means that $\xi = \lambda \bar\kappa . \bar{u}\tau$ has type
$\forall \bar\kappa . \Xi_{i,j} [\cappp  \delta, \cappp \gamma]$, which is exactly the input to $y_j$, so that $y_j\, \xi$ has type $D[\lambda \bar\kappa.\,\cappp x
(\bar{u}\tau)]$, as desired.

\item
Assume $M = \con_j(\bar{t}, \bar{\lambda \xi . M'}, \bar{r})$. By inductive hypothesis we have
that
\begin{align*}
\hastype {\Gamma', \hat\gamma : \forall \bar\kappa.\, \hat\Gamma[\cappp \delta,\cappp \gamma],
			\xi : \forall\bar\kappa. \Xi_{j,k}[\cappp \delta, \bar{t}\tau]\\}
			{\evalbsynclockg{\delta}{\mathcal{E}}{\bar{x} \mapsto \bar{y}}{(\hat\gamma,\xi)}{M'_k}}
			{D[\lambda \bar\kappa. M'_k\sigma]&}
\end{align*}
This gives us the input necessary to apply $\mathcal{E}_j$. The $\bar{t}$ family has a similar typing
structure to $\bar{u}$ in the previous example so we get $\bar{t}' = \lambda\bar\kappa
. \bar{t}\tau$ of type $\forall \bar\kappa. \Gamma_j[\cappp \delta]$.
We obtain maps
\begin{align*}
S_k &: \forall\bar\kappa.\,\Xi_{j,k}[\cappp \delta, \cappp {\bar t'}] \to \HIT H {(\cappp \delta)}\\
S_k &= \lambda \bar\kappa.\lambda \xi. M'_k\sigma
\end{align*}
directly from the typing assumptions and by inductive hypothesis we obtain maps
\begin{align*}
R_k & : (\xi : \forall\bar\kappa. \Xi_{j,k}[\cappp \delta, \cappp {\bar t'}])
        \to D[\lambda \bar \kappa.\cappp{S_k}(\cappp\xi)]\\
R_k & = \lambda \xi.
      \evalbsynclockg{\delta}{\mathcal{E}}{\bar{x} \mapsto \bar{y}}{(\hat\gamma,\xi)}{M_k'}
\end{align*}
This means we have the required data to apply $\mathcal{E}_j$, and by the
definition of clock abstracted elimination lists we have that
$\mathcal{E}_j[\bar{t}', \bar{S}, \bar{R}, \bar{r}]$ inhabits $D$ over $\lambda\bar\kappa. M\sigma$.

\item
Assume $M = \hcomp^{j}_{\HIT H {\delta}} [\psi \mapsto M'] \, M'_0$. By inductive
hypothesis and the typing assumptions for the $\hcomp$ to be well formed we
have that $\evalbsynclockg{\delta}{\mathcal{E}}{\bar{x}
\mapsto \bar{y}}{\hat\gamma}{M'_0}$ inhabits $D$ over $\lambda \bar \kappa.\, M'_0\sigma$,
and $\evalbsynclockg{\delta}{\mathcal{E}}{\bar{x} \mapsto \bar{y}}{\hat\gamma,j}{M'}$
inhabits $D$ over $\lambda \bar \kappa.\, M'\sigma$ in the context extended by $j : \I$ and restricted by $\psi$. The $v$ term provides a path between $\lambda \bar \kappa.\, M'_0\sigma$ and its composition with $\lambda \bar \kappa.\, M'\sigma$, which means that the composition in $D[v j]$ provides a term over $\lambda \bar\kappa.\, M\sigma$ as desired.
Finally, for the composition to be well typed, we must verify that
\[
\evalbsynclockg{\delta}{\mathcal{E}}{\bar{x} \mapsto \bar{y}}{\hat\gamma,j}{M'}[0/j]
\jeq \evalbsynclockg{\delta}{\mathcal{E}}{\bar{x} \mapsto \bar{y}}{\hat\gamma}{M'_0}
\]
An easy induction shows that the left hand side equals
\[
\evalbsynclockg{\delta}{\mathcal{E}}{\bar{x} \mapsto \bar{y}}{\hat\gamma}{M'[0/j]}\] 
and since
$M_0' \beq M'[0/j]$, this follows from Lemma~\ref{lem:bound:inter:eq} below.
\item
Assume $M = [\phi_0 \, M_0, \dots, \phi_{n_i} \, M_{n_i}]$. In this case we have the following
typing by inductive hypothesis:
\begin{align*}
\hastype {\Gamma' , \phi_k}
			{ \evalbsynclockg{\delta}{\mathcal{E}}{\bar{x} \mapsto \bar{y}}{\hat\gamma}{M_k}}
			{ D[\lambda \bar\kappa. M_k\sigma]}
\end{align*}
Finally, in order to conclude that
\[
[\phi_1 \, \evalbsynclockg{\delta}{\mathcal{E}}{\bar{x} \mapsto \bar{y}}{\hat\gamma}{M_1}, \dots ,
\phi_{n_i} \, \evalbsynclockg{\delta}{\mathcal{E}}{\bar{x} \mapsto \bar{y}}{\hat\gamma}{M_{n_i}}
]
\]
defines a system, and so a partial element of
$D[\lambda \bar\kappa. M\sigma]$, we must show that on faces of the form $\phi_{j} \wedge \phi_k$
the judgemental equality
\[
\evalbsynclockg{\delta}{\mathcal{E}}{\bar{x} \mapsto \bar{y}}{\hat\gamma}{M_j} \jeq
\evalbsynclockg{\delta}{\mathcal{E}}{\bar{x} \mapsto \bar{y}}{\hat\gamma}{M_k}
\]
holds. Since on this face $M_j \beq M_k$, this follows from Lemma~\ref{lem:bound:inter:eq} below. \qedhere
\end{itemize}
\end{proof}

\begin{lemma} \label{lem:bound:inter:eq}
 If $\hasnotype{\Gamma}{M \beq N}{}$ then $\hasnotype{\Gamma}{\evalbsynclockg{\delta}{\mathcal{E}}{\bar{x} \mapsto \bar{y}}{\hat\gamma}M \jeq \evalbsynclockg{\delta}{\mathcal{E}}{\bar{x} \mapsto \bar{y}}{\hat\gamma}N}{}$
\end{lemma}

\begin{proof}
The proof is by induction on the proof of $M \beq N$. The interesting case is the reduction of a constructor
to its boundary, which requires showing that $\phi_i \jeq \top$  implies
\begin{equation} \label{eq:interp:bound:con:reduction}
\evalbsynclockg{\delta}{\mathcal{E}}{\bar{x} \mapsto \bar{y}}{\hat\gamma}{\con_{\ell}(\bar{t},\overline{\lambda \xi.\,M},\bar{r})}
\jeq \evalbsynclockg{\delta}{\mathcal{E}}{\bar{x} \mapsto \bar{y}}{\hat\gamma}{N_i(\bar{t},\overline{\lambda \xi.\,M},\bar{r})}
\end{equation}
The left hand side of this equation is $\mathcal{E}_\ell[\bar{t}', \bar{S}, \bar{R}, \bar{r}]$, which under the
assumption that $\phi_i \jeq \top$ equals
\[
\evalbsynclockg{\delta}{\mathcal{E}}{\bar{x} \mapsto \bar{y}}{\hat\gamma}{N_i}[\bar{t}', \bar{S}, \bar{R}, \bar{r}/
\bar \gamma, \bar x, \bar y, \bar i]
\]
It thus suffices to prove equality of the above with the right hand side of
(\ref{eq:interp:bound:con:reduction}). We omit the straight forward induction proof of this substitution property.
\end{proof}

\emph{Details of higher truncation.}
Recall that we defined a map $\alpha : \forall \kappa. \, \htrunc{A}{n} \to
\htrunc{\forall \kappa. \, A}{n}$. For $f : \forall \kappa . \, \sphere {n+1} \to
\htrunc{A}{n}$ we let $f' = \lambda s. \, \alpha(\lambda\kappa. \, (\capp f)(s))$
and $p_s : \lambda \kappa. \, \capp[\kappa_0] s = s$ is the path extracted from clock irrelevance
of $\sphere {n+1}$. Note that from the definition of $\alpha$, we get the following reductions:
\begin{align*}
\alpha(\lambda \kappa . \mathsf{in}(\capp a)) & \equiv \mathsf{in}(a) \\
\alpha(\lambda \kappa . \mathsf{hub}(\capp f)) & \equiv \mathsf{hub}(f') \\
\alpha(\lambda\kappa. \, \mathsf{spoke}(\capp s, \capp f, i)) &\equiv \\
\hcomp^j  [(i=0) \mapsto & \alpha(\lambda \kappa . \, (\capp f) (\capp {(p_s\, j)})),\\
          (i=1) \mapsto & \mathsf{hub}(f')]  \\
            \mathsf{spoke} & (\capp [\kappa_0] s,  f', i)
\end{align*}

For ease of reasoning, we write out the reductions for the canonical map $\beta$:

\begin{align*}
\beta(\mathsf{in}(a)) & \equiv \lambda \kappa . \mathsf{in}(\capp a) \\
\beta(\mathsf{hub}(f)) & \equiv \lambda\kappa. \mathsf{hub}(\lambda s. \, \capp {(\beta (f (s)))}) \\
\beta(\mathsf{spoke}(s, f, i)) &\equiv \lambda\kappa. \, \mathsf{spoke}(s, \lambda s. \, \capp {(\beta (f (s)))}, i)
\end{align*}

We show that the two maps are inverse to one another by induction, leaving out the
trivial first case. For $\mathsf{hub}$ we reason as follows:

\begin{align*}
\alpha(\beta(\mathsf{hub}(f)))
& \equiv \alpha(\lambda\kappa. \mathsf{hub}(\lambda s. \, \capp {(\beta (f (s)))}))\\
& \equiv \mathsf{hub}(\lambda s . \, \alpha(\lambda\kappa. \, \capp {(\beta (f (s)))}))\\
& \equiv \mathsf{hub}(\lambda s . \, \alpha(\beta (f (s))))\\
& = \mathsf{hub}(\lambda s . \, f (s))\\
& \equiv \mathsf{hub}(f)
\end{align*}
\begin{align*}
\beta(\alpha(\lambda \kappa . \mathsf{hub} & (\capp f)))
 \equiv \beta(\mathsf{hub} (\lambda s. \, \alpha(\lambda\kappa. \, (\capp f)(s))))\\
& \equiv \lambda\kappa' . \, \mathsf{hub}(\lambda s . \, \capp [\kappa'] {\beta(\alpha(\lambda \kappa. \, (\capp f)(s)))})\\
& = \lambda\kappa' . \, \mathsf{hub}(\lambda s . \, \capp [\kappa'] {(\lambda \kappa. \, (\capp f)(s))})\\
& \equiv \lambda\kappa . \, \mathsf{hub}(\lambda s . \, (\capp f)(s))\\
& \equiv \lambda\kappa . \, \mathsf{hub}(\capp f)
\end{align*}
Note that the only non-judgemental equality is the inductively justified cancellation
of the compositions i.e., the path $q : \lambda s . \, \alpha(\beta (f (s))) = f$
and $r : \lambda s.\, \beta(\alpha(\lambda \kappa. \, (\capp f)(s))) = \lambda s.
\, \lambda \kappa. \, (\capp f)(s)$ under the hub constructor.

We now supply the first calculation for the $\mathsf{spoke}$ constructor:
\begin{align*}
\alpha(\beta(  \mathsf{spoke} & (s, f, i))) \equiv
\alpha(\lambda\kappa. \, \mathsf{spoke}(s, \lambda s. \, \capp {(\beta (f (s)))}, i))\\
\equiv \hcomp^j [& (i=0) \mapsto \alpha(\lambda \kappa. \, \capp {(\beta (f ((\capp {p_{\const \, s}\, j}))))}) ,\\
                 & (i=1) \mapsto \mathsf{hub}(\lambda s . \, \alpha(\lambda\kappa. \, \capp {(\beta (f (s)))}))]\\
                 & \, \mathsf{spoke}(s,  \lambda s . \, \alpha(\lambda\kappa. \, \capp {(\beta (f (s)))}), i) \\
& = \mathsf{spoke}(s,  \lambda s . \, \alpha(\lambda\kappa. \, \capp {(\beta (f (s)))}), i)\\
& \equiv \mathsf{spoke}(s,  \lambda s . \, \alpha(\beta(f(s))), i)\\
& = \mathsf{spoke}(s, f, i)
\end{align*}

For us to apply induction with the above path it would need to reduce strictly to
the recursive call at $s$, $\lambda j. \, q(j)(s)$, on $i=0$ and the $\mathsf{hub}$
case, $\lambda j. \, \mathsf{hub}(\lambda s.\, q(j)(s))$, on $i=1$. This is too
tall an ask, but luckily it is sufficient for our purposes that we can show such
reductions up to a path. The path used for the induction is then defined by an
appropriate composition as in the definition of the map $\alpha$.

The calculation above uses two non-trivial paths. The first is the reduction of a
composition to its base, which we note restricts to $\lambda j.\, \alpha(\lambda
\kappa. \, \capp {(\beta (f ((\capp {p_{\const \, s}\, j}))))})$ and reflexivity on $i=0$ and $i=1$
respectively. The second is exactly an application of $q$ from above to the
function input, or more concretely the path $\lambda j. \mathsf{spoke}(s,
\lambda s. \, q(j)(s)), i)$. The reductions of spoke then mean that this reduces
to exactly the paths we are looking for. Since $p_{\const}$ is given by clock
irrelevance at a constant function it is path equal to reflexivity by lemma~\ref{lem:cirr:coh}.
Cancelling this path at $i=0$ and reflexivity at $i=1$ with a composition then
yields the desired path.

\begin{align*}
\beta(\alpha(\lambda\kappa. \,  \mathsf{spoke}(\capp s, & \capp f, i))) \\
\equiv \beta(\hcomp^j  [(i=0) & \mapsto \alpha(\lambda\kappa. \, (\capp f) (\capp {(p_s\, j)})), \\
                        (i=1) & \mapsto  \mathsf{hub}(\lambda s. \, \alpha(\lambda\kappa. \, (\capp f)(s))))] \\
                    \mathsf{spoke} & (\capp [\kappa_0] s,  \lambda s. \, \alpha(\lambda\kappa. \, (\capp f)(s))), i))\\
\equiv \hcomp^j  [(i=0) \mapsto & \beta(\alpha(\lambda\kappa. \, (\capp f) (\capp {(p_s\, j)}))),\\
                  (i=1) \mapsto & \beta(\mathsf{hub}(\lambda s. \, \alpha(\lambda\kappa. \, (\capp f)(s)))))] \\
                  \beta  (\mathsf{spoke}  ( & \capp [\kappa_0] s,  \lambda s. \, \alpha(\lambda\kappa. \, (\capp f)(s))), i)))\\
 \equiv \hcomp^j  [(i=0) \mapsto & \beta(\alpha(\lambda\kappa. \, (\capp f) (\capp {(p_s\, j)}))),\\
                   (i=1) \mapsto & \lambda\kappa' . \,
                            \mathsf{hub}(\lambda s . \, \capp [\kappa'] {\beta(\alpha(\lambda \kappa. \, (\capp f)(s)))})] \\
                   \lambda\kappa' . \mathsf{spoke} & (\capp [\kappa_0] s, \lambda s. \, \capp [\kappa'] {\beta(\alpha(\lambda\kappa. \, (\capp f)(s))))}, i)\\
= \lambda\kappa' . \mathsf{spoke} & (\capp [\kappa_0] s, \lambda s. \, \capp [\kappa'] {\beta(\alpha(\lambda\kappa. \, (\capp f)(s))))}, i)\\
= \lambda\kappa' . \mathsf{spoke} & (\capp [\kappa']  s, \lambda s. \, \capp [\kappa'] {\beta(\alpha(\lambda\kappa. \, (\capp f)(s))))}, i)\\
= \lambda\kappa . \mathsf{spoke}( & \capp s, \capp f, i)
\end{align*}

This time the boundary obligation is that on $i=0$ the above must reduce to
$\lambda j .\lambda \kappa. \, \capp {r(j)(\capp s)}$ and on $i=1$ it must reduce
to $\lambda j. \,\lambda\kappa.\, \mathsf{hub} (\lambda s. \, \capp{r(j)(\capp s)})$.
The calculation is a composition of three non-trivial paths. We first
reduce the composition to its base, which reduces to reflexivity on $i=1$ and
$\lambda j.\, \beta(\alpha(\lambda\kappa. \, (\capp f) (\capp {(p_s\, j)})))$ on
$i=0$. The next path is again the path extracted from clock irrelevance, this time
the one that shows that $\capp [\kappa_0] s = \capp [\kappa'] s$, which is
exactly the inverse application of clock irrelevance. This reduces to reflexivity
on $i=1$ and the path $\lambda j.\, \beta(\alpha(\lambda\kappa. \, (\capp f)
(\capp {(p_s^{-1}\, j)})))$ on $i=0$. At this point the composite path has the
shape $\refl \circ \lambda j. \, \beta(\alpha(p'(j)) \circ \lambda j. \, \beta(\alpha(p'^{-1}(j))$
on $i=0$ and $\refl^3$ on $(i=1)$, meaning that it reduces to $\refl$ up to a path
in either case. The last path is then the inductively obtained path
\[\lambda j. \, \lambda\kappa . \mathsf{spoke} (\capp s, \lambda s. \, \capp {r(j)(s)}, i).\]
The desired reduction now again follows from an application of the equalities governing
the behavior of $\mathsf{spoke}$.

\emph{Proving \autoref{prop:cirr:HITs}.}
We proceed by constructing a map \[(h : \forall \kappa.
\, \HIT H \delta) \to \forall \kappa. \, \capp h = \capp[\kappa_0] h.\] To do
this we first show a slightly modified induction under clocks principle with
constant $\delta$ and $\Gamma_{(-)}$ parameters. Using this modified principle
we construct the map in two stages: first we provide a candidate case for each
constructor and secondly we show that the candidate terms satisfy the appropriate
boundary condition up to path equality. The latter will allow us to rectify the
terms given by the former constructions via a composition to obtain the input for
the modified principle, thus allowing us to define a map of the desired type.
The new induction principle needs a modified version of the boundary condition,
given by transporting the usual one along the equivalence between the two, which
means that we need to, at each stage, cohere with the transport in the earlier
stages.

Before beginning the proof we need to introduce some notation and prove a small
lemma about the structure of HIT constructors.

\begin{lemma}
Let $\HIT H \delta$ be a HIT with constructors taking input
\begin{align*}
& \gamma^0 : \Gamma_i^0[\delta], \gamma^1 : \Gamma_i^1[\delta],\dots,
& x : \Xi_i^0 [\delta, \bar\gamma] \to \HIT H \delta, \dots \, \text{and} \,\,
\bar i : \Psi_i
\end{align*}
and boundary conditions $\phi_i \vdash e_i$. Then there exists an equivalent HIT
$\HIT {H'} \delta$ with constructors taking input of the form
\begin{align*}
& \gamma : \Gamma_i[\delta],
x : \Xi_i [\delta, \bar\gamma] \to \HIT {H'} \delta \, \text{and} \,\,
\bar i : \Psi_i
\end{align*}
with boundary conditions $\phi_i \vdash e_i$ where $\phi_i$ is of the form
$\bigvee (i=0) \vee (i=1)$ with $i$ ranging over all variables in $\Psi$. We say
that HITs satisfying these three criteria are of reduced form.
\end{lemma}
\begin{proof}
Modifying the $\Gamma$ and $\Xi \to \HIT H \delta$ input is trivial, simply take
$\Gamma_i [\delta]$ to be an iterated $\Sigma$ type consisting of the $\Gamma_i^j$
and $\Xi_i[\delta, \gamma] \defeq \Xi_i^j[\delta, \gamma] + \Xi_i^j[\delta, \gamma]
+ \dots$. This procedure clearly yields an equivalent HIT, so we make this assumption
freely.

We achieve boundaries of the desired shape by adding constructors to specify the
full boundary, noting that $\bigvee (i=0) \vee (i=1)$ as above is the second
largest element of the face lattice specifying the entire boundary of a cube but
not the interior. We call it the total face relative to $\Psi$. Let $\con_i$ be
a constructor of $\HIT H \delta$ with boundary
extent $\phi_i$ and interval input $\Psi_i$. We add a number of constructors to
$\HIT {H'} \delta$ recursively in the following way: Write $\phi_i$ in disjunctive
normal form. For each $j \in \Psi_i$ if $(j=0)$ and $(j=1)$ appear as disjuncts
of $\phi_i$ we add $\con_i$ as is. Say that $(j=0)$ is missing; we then add a constructor
$\con_i^{(j=0)}$ with the same $\Gamma$ and $\Xi$ as $\con_i$. This new constructor
then has interval input $\Psi$ with the $j$ variable deleted. Now consider the face
$\phi_i \wedge (j=0)$. If this is the total face relative to $\Psi$ without $j$,
we define the boundary term of $\con_i^{(j=0)}$ to be the boundary of $\con_i$
restricted to $(j=0)$ and proceed with the next interval variable in $\Psi$. If
it is not the total face and is missing for instance $k\in \Psi$, we repeat the
procedure, defining a new constructor $\con_i^{(j=0)\wedge (k=0))}$ in the same way.
Having recursively defined this new constructor, we use it for the boundary at $k=0$.

It is immediate that we can define mutually inverse maps between $\HIT H \delta$
and $\HIT {H'} \delta$ using their respective elimination principles.
\end{proof}

The point of the modified boundary shape is that constructors with such boundaries
exactly correspond to constructors for heterogeneous, iterated path types. This means that for
HITs of the above form we can treat the end result $(i:\Psi) \to A[\phi \mapsto
e]$ as a proper type, allowing it to appear in e.g., $\Sigma$-types. We treat
it as we would path types, introducing terms of them by abstraction and eliminating
from them by application. As for path types we can compose iterated paths, and we
record this fact in the following lemma:

\begin{lemma}\label{lemma:transp:ext}
\[
\inferrule*{\hastype{\Gamma}{p}{\Pi (i : \Psi).\, A[ \phi \mapsto u ]} \\
\hastype{\Gamma, (i : \Psi), \phi}{q}{u = v}}
{\hastype{\Gamma}{(i.\,q)^*\,p}{\Pi (i : \Psi).\, A[ \phi \mapsto v ]}}
\]
\end{lemma}
\begin{proof}
\(
(i.\,q)^*\,p \defeq \lambda i.\, \hcomp^j\,[\phi \mapsto q\,j]\,(p\,i)
\)
\end{proof}

\newcommand{\Elim}{\mathsf{Elim}}
\newcommand{\Elimc}{\mathsf{Elim}^{\const}}


We define the following types in context where $\delta : \Delta$, and $D$ a family over $\forall \kappa. \HIT{H}{\delta}$, where
$\bar{u_{<\ell}} : \mathcal{E}_{<\ell} (\delta,D)$ as specified below, and $\delta : \forall\kappa. \Gamma_\ell(\delta)$:
\begin{align*}
\mathcal{P}_\ell &(\delta,D,\bar{u_{<\ell}},\gamma) \defeq \\\Pi
        & (x : (\forall \kappa .\, \Xi_\ell[\delta,\capp \gamma] \to
             \HIT H \delta))\\
        & (y : \Pi(\xi : \forall \kappa.\, \Xi_\ell[\delta,\capp \gamma]).
        D[\lambda \kappa . \capp x (\capp \xi)])\\
        & (i : \Psi_\ell).\\
        & D[\lambda \kappa. \, \con_\ell (\capp \gamma, \capp x, i)]
        [\phi_\ell \mapsto \, \evalbsynclock{\delta}{\bar{u_{<\ell}}}{x \mapsto y}{e_\ell}] \\
\mathcal{E}_\ell &(\delta,D,\bar{u_{<\ell}}) \defeq
    \Pi (\gamma : \forall \kappa . \Gamma_\ell[\delta])
        .\, \mathcal{P}_\ell(\delta,D,\bar{u_{<\ell}},\gamma)\\
\mathcal{E}'_\ell &(\delta, D, \bar{u_{<\ell}}) \defeq
    \Pi (\gamma : \Gamma_\ell[\delta])
        .\, \mathcal{P}_\ell(\delta,D,\bar{u_{<\ell}},\lambda \_.\,\gamma)
\end{align*}
where
\begin{align*}
\mathcal{E} (\delta,D) & =  \Sigma (u_{\ell_0} : \mathcal{E}_{\ell_0} (\delta)).\,
  \Sigma (u_{\ell_1} : \mathcal{E}_{\ell_1} (\delta, u_{\ell_0})).\\
  &  \ldots\,
  \mathcal{E}_{\ell_n} (\delta, (u_{\ell_0},u_{\ell_1}, \ldots, u_{\ell_{n-1}}))
\end{align*}
and $\mathcal{E}_{<\ell} (\delta,D)$ is the prefix of the above iterated
$\Sigma$ types containing fields only for the labels below $\ell$.
The above definition is well-founded because $\mathcal{E}_\ell(\delta,-)$ only refers to $\mathcal{E}_{\ell'}(\delta,-)$ for labels $\ell' < \ell$.

\newcommand{\cirrg}[1]{\mathsf{cirr}_{\Gamma_{#1}}}
\newcommand{\cirrgcoh}[1]{\mathsf{cirr\!\!-\!\!coh}_{\Gamma_{#1}}}
\begin{lemma}\label{lem:cirr:coh}
  For each $\ell$ we have
  \begin{align*}
  \cirrg{\ell} & : \Pi (\gamma : \forall \kappa. \Gamma_\ell[\delta]).\, \lambda \_.\, \capp[\kappa_0] \gamma = \gamma\\
  \cirrgcoh{\ell} & : \Pi (\gamma : \Gamma_\ell[\delta]).\, \mathsf{refl} = \mathsf{cirr}_{\Gamma_\ell}\,(\lambda \_.\, \gamma)
  \end{align*}
\end{lemma}
\begin{proof}
From clock irrelevance of $\Gamma_\ell$ we get a proof that the constant map from $\Gamma_\ell[\delta]$ to $\forall \kappa. \Gamma_\ell[\delta]$ has a right inverse. Moreover it has application to $\kappa_0$ as a left inverse with reflexivity as the proof. From this we obtain a proof that the constant map is an half adjoint equivalence, from which we can project $\cirrg{\ell}$ and $\cirrgcoh{\ell}$.
\end{proof}

\begin{lemma}
For each $\ell$, $\delta$ and $u_{<\ell}$ we have an equivalence of types
$f_\ell : \mathcal{E}'_\ell(\delta,D,u_{<\ell}) \simeq \mathcal{E}_\ell(\delta,D,u_{<\ell})$.
\end{lemma}
\begin{proof}
We will use $\cirrg{\ell}$ to transport
$\mathcal{P}_\ell(\delta,D,u_{<\ell},\lambda \_.\, \capp[\kappa_0] \gamma)$ to
$\mathcal{P}_\ell(\delta,D,u_{<\ell},\gamma)$. Concretely we define
\[
f_\ell(u') \defeq \lambda \gamma.\, \cirrg{\ell}(\gamma)^* (u'\,(\capp[\kappa_0] \gamma))
\]
\end{proof}

We then define $\mathcal{E}' (\delta,D)$ as the iterated sigma type
\begin{align*}
\Sigma (u_{\ell_0} : \mathcal{E}'_{\ell_0} (\delta,D)).
  \Sigma (u_{\ell_1} : \mathcal{E}'_{\ell_1} (\delta,D, f_{\ell_0}(u_{\ell_0})).\\
  \ldots
  \mathcal{E}'_{\ell_n} (\delta,D, f(\bar{u_{<\ell_{n}}})
\end{align*}
where we wrote $f(\bar{u_{<\ell_{n}}})$ in place of
\[
f_{\ell_0}(u_{\ell_0}),f_{\ell_1}(u_{\ell_1}), \ldots, f_{\ell_{n-1}}(u_{\ell_{n-1}})
\]
as we will do going forward. In fact the family $f_\ell$ can be collected into an equivalence of type
$\mathcal{E}' (\delta,D) \simeq \mathcal{E} (\delta,D)$ which also restricts to the $<\ell$ case.

\begin{lemma}\label{lem:pi-eliminators}
Let $\delta : \Delta$, and $\istype
        {h : \forall \kappa. \, \HIT H (\delta)} D$ and let $t : \forall \kappa . \,
        \HIT H \delta$.
\begin{itemize}
  \item[(a)] Elimination under a single clock allows us to
        produce a term of the type $\Elim(\delta, D, t) \defeq \mathcal{E}(\delta,D) \to D[t]$.
  \item[(b)] 
        Clock elimination with
        constant $\Gamma_\ell$ parameter allows us to produce a term of the type
        $\Elimc(\delta, D, t) \defeq \mathcal{E}'(\delta,D) \to D[t]$
\end{itemize}
\end{lemma}
\begin{proof}
Case (a) is direct from typing of elimination under a single clock,
case (b) follows by composing (a) with the equivalence $f$.
\end{proof}

In the following we will fix $D[t]$ to be $\forall \kappa. \capp[\kappa_0] t = \capp t$.

\begin{lemma}\label{lem:cirrhits:mapcases}
For each $\ell$ we can type $b_\ell$ as shown:
\begin{align*}
\mathcal{T}_\ell(\delta) \defeq
    \Pi & (\gamma : \Gamma_\ell[\delta])\\
        & (x : \forall \kappa .\, \Xi_\ell[\delta,\gamma] \to
             \HIT H \delta)\\
        & (y : \Pi(\xi : \forall \kappa.\, \Xi_\ell[\delta,\gamma]).
        D[\lambda \kappa . \capp x (\capp \xi)]).\\
        & (i : \Psi_\ell).\\
        & D[\lambda \kappa. \, \con_\ell (\gamma, \capp x, i)]\\
        &\hspace{.4cm}[\phi_\ell \mapsto \lambda \kappa.\, \lambda j.\, e_\ell[\gamma,\lambda \xi.\, \capp{y\,(\lambda \_.\, \xi)}\, j, i]] \\
\hastype{\delta : \Delta}
{b_\ell \defeq& \lambda \gamma\,x\,y\,i\,\kappa\,j.\, \con_\ell(\gamma, \lambda \xi.\, \capp{y\,(\lambda \_.\, \xi)}\,j, i)}
{\mathcal{T}_\ell(\delta)}
\end{align*}
\end{lemma}
\begin{proof}
Follows directly by the typing rule for $\con_\ell$.
\end{proof}
Note that $\mathcal{T}_\ell(\delta)$ differs from $\mathcal{E}'_\ell(\delta,D,\bar{u_{<l}})$ only in the boundary of the final result. We bridge this gap with lemma~\ref{lem:cirrhits:boundary}.

\newcommand{\lconst}[1]{\lambda \_.\, {#1}}
\begin{lemma}\label{lem:cirrhits:boundary}
The terms defined in lemma \ref{lem:cirrhits:mapcases} satisfy the boundary
conditions of lemma \ref{lem:pi-eliminators} (b) up to path equality, i.e., for each $\ell$ we have a
term $\lemmaD{\ell}$ of type
\[
\lambda \kappa.\, \lambda j .\, e_\ell[\gamma,\lambda \xi.\, \capp{y\,(\lconst \xi)}\, j, i]
=
\evalbsynclock{\lconst \delta}{\bar{v_{<\ell}}}{x \mapsto y}{e_\ell}[\lambda \_.\, \gamma]
\]
in the appropriate context, and where each $v_\ell$ is defined as
\begin{align*}
v_\ell :\, & \mathcal{E}(\delta,D,\bar{v_{<\ell}})\\
v_\ell \defeq & f_\ell(g_\ell(b_\ell))\\
g_\ell :\, & \mathcal{T}_\ell(\delta) \to \mathcal{E}'(\delta,D,\bar{v_{<l}})\\
g_\ell(t) \defeq & \lambda \gamma.\lambda x.\lambda y. (i. \lemmaD{\ell})^* (t\,\gamma\,x\,y)
\end{align*}
\end{lemma}

\begin{proof}
To arrive at the desired conclusion we need a generalized version of the lemma
which keeps track of how the proof in each case coheres with the earlier stages.
We prove the following:

\begin{align*}
\Omega_\ell \defeq &\delta : \Delta, \gamma : \Gamma_\ell, x : \forall \kappa.\, \Xi_\ell \to \HIT{H}{\delta},\\
& y : \Pi(\xi : \forall \kappa.\, \Xi_\ell) . D[\lambda \kappa.\, \capp x\,(\capp \xi)],\\
& y_K : \Pi(\xi : \Xi_\ell) . D [\lambda \kappa.\, \capp x\, \xi],\\
& \tilde y : \Pi(\xi : \Xi_\ell) . y_K\,\xi = y (\lambda \_.\, \xi) ,\\
&i : \Psi_\ell,\phi_\ell
\end{align*}
\begin{align*}
\hastype{\Omega_\ell, \hat\Gamma&
}
{\lemmaDg \ell M\\ &}
{
(\lambda \kappa.\, \lambda j .\, M[\lambda \xi.\, \capp{y_K\,\xi}\,j/x])
\\&=
\evalbsynclockg{\lconst \delta}{\bar{v_{<\ell}}}{x \mapsto y}{\hat\gamma}{M}
[\lambda \_.\, \gamma/\gamma,\lambda \_.\,\hat\gamma/\hat\gamma]
}
\end{align*}
$\lemmaD \ell$ is then defined as $$\lemmaDg \ell {e_\ell}[\lambda \xi.\, y\,(\lambda \_. \xi)/y_K, \lambda \xi.\, \mathsf{refl}/ \tilde y]$$
we will write $\sigma$ for $[\lambda \xi.\, \capp{y_K\,\xi}\,j/x]$ and
$\tau$ for $[\lambda \_.\, \gamma/\gamma,\lambda \_.\,\hat\gamma/\hat\gamma]$.
The extra $y_K$ and $\tilde{y}$ parameters provide what to do for the
case $M = x_j\,\bar{u}$ where we will apply $\tilde{y}$ to
$\bar{u}$ to obtain the necessary equality
between the applications of $y_K$ and $y$.
This allows us to derive that $\lemmaDg {\ell} {-}$ commutes with substitution in the following way,
\[
\lemmaDg{\ell}{M[\bar{t},\lambda \xi.\, M',\bar{r}]} \jeq
\lemmaDg{\ell}{M}[\bar{t},S,R,R_K,\tilde{R},\bar{r}]
\]
where
\begin{align*}
S &\defeq \lambda \kappa. \lambda \xi.\, M'[\capp x/x]\\
R &\defeq \lambda \xi.\, \evalbsynclockg{\lconst\delta}{\bar{v_{<\ell}}}{x \mapsto y}{(\hat\gamma,\xi)}{M'}\tau\\
R_K &\defeq \lambda \xi. \lambda \kappa.\lambda j.\, M'\sigma\\
\tilde{R}& \defeq \lambda \xi.\, \lemmaDg{\ell}{M'}
\end{align*}
and moreover we have $\lemmaDg{\ell'}{M} = \lemmaDg{\ell}{M}$ for $\ell' < \ell$, whenever $M$ only
contains constructor nodes with labels smaller than $\ell'$.
We will use these properties for the constructor case of $\lemmaDg{\ell}{-}$.

Finally, we need to show that $\lemmaDg{\ell}{-}$ preserves judgemental equality in the sense that $M \beq N$ implies
\[\lemmaDg{\ell}{M} \jeq \lemmaDg{\ell}{N}\]
This is needed to show that $\lemmaDg{\ell}{-}$ is well-defined in the case
of systems as in the proof of Lemma~\ref{lem:evalbsyn:typing}, and to define $\lemmaDg{\ell}{-}$ in the case of homogeneous
compositions.

We induct on the structure of $M$.
In case $ M = x\,\bar{u}$, we need to build a path between
$\lambda \kappa.\, \lambda j .\, \capp{(y_K\,(\bar{u}))}\,j$
and
\[y (\lambda \kappa.\, \bar u [\cappp{(\lambda \_.\, \gamma)}/\gamma,\cappp{(\lambda \_.\,\hat\gamma)}/\hat\gamma]).\]
The clock applications on the right hand side simplify, making the body of the lambda abstraction constant in $\kappa$. The left hand side $\eta$-contracts to $y_K\,\bar{u}$, so we can conclude by setting $\lemmaDg{\ell}{x_j\,\bar{u}}$ equal to $\tilde{y}\,\bar{u}$.

In the case $M = \hcomp^{j'}\,[\psi \mapsto M']\,M_0'$, we must build a path between
$\lambda \kappa.\, \lambda j .\, \hcomp^{j'}\,[\psi \mapsto M'\sigma]\,M_0'\sigma$
and
\[\comp^{j'}_{D[v\tau\,j']}\,[\psi \mapsto \evalbsynclockg{\lconst \delta}{\bar{v_{<\ell}}}{x \mapsto y}{(\hat\gamma,j')}{M'}\tau]\,(\evalbsynclockg{\delta}{\bar{v_{<\ell}}}{x \mapsto y}{\hat\gamma}{M_0'}\tau)
.\]
Let $p$ be the path connecting
\begin{align*}
\lambda \kappa. \lambda j. \hcomp^{j'}\,[ &\psi \mapsto \capp{\evalbsynclockg{\lconst \delta}{\bar{v_{<\ell}}}{x \mapsto y}{(\hat\gamma,j')}{M'}\tau}\,j] \\ & \,(\capp{\evalbsynclockg{\delta}{\bar{v_{<\ell}}}{x \mapsto y}{\hat\gamma}{M_0'}\tau}\,j)
\end{align*}
to the right hand side, obtained from the fact that both terms fill the same open box.
We let $q$ be the path connecting the left hand side to $p\,0$, obtained by combining $\lemmaDg{\ell}{M'}$ and $\lemmaDg{\ell}{M'_0}$ with $\hcomp^{j'}$. Note that this is well-defined because $\lemmaDg{\ell}{-}$ preserves judgemental equality.
By transitivity we get $q \cdot p$ connecting the desired endpoints.

To prove that $\lemmaDg{\ell}{-}$ preserves judgemental equality, we need that $\lemmaDg{\ell}{M}$ when restricted by $\psi$ is strictly equal to $\lemmaDg{\ell}{M'}$. Fortunately that is already true for $q$, while $p$ is a constant path under those conditions, so by the right unit law we have path from $q \cdot p$ to $\lemmaDg{\ell}{M'}$. Using an $\hcomp$ with this latter path we define $\lemmaDg{\ell}{M}$ so that it satifies the strict equality.

In case $M = \con_{\ell'}(\bar{t},\lambda \xi.\, M',\bar{r})$, we need to build a path between
$\lambda \kappa.\, \lambda j .\, \con_{\ell'}(\bar{t},\lambda \xi.\, M'\sigma,\bar{r})$
and $v_{\ell'}(\lconst {\bar t},S,R,\bar{r})$ where
$S = \lambda \kappa. \lambda \xi.\, M'[\capp x/x]$
and $R = \lambda \xi.\, \evalbsynclockg{\lconst\delta}{\bar{v_{<\ell}}}{x \mapsto y}{(\hat\gamma,\xi)}{M'}\tau$. We will work right to left by expanding the right hand side definition.
By definition we have $v_{\ell'} = f_{\ell'}(g_{\ell'}(b_{\ell'}))$, so $v_{\ell'}(\lconst {\bar t})$ is equal to
\[\cirrg{{\ell'}}(\lconst {\bar t})^*(g_{\ell'}(b_{\ell'})(\bar{t})),\] 
so that by $\cirrgcoh{\ell}(\bar{t})$ the right hand side is equal to $g_{\ell'}(b_{\ell'})(\bar{t},S,R,\bar{r})$, let us call this path $p_1$. Note that $p_1$ will be a constant path when $\phi_{\ell'}[\bar r] = 1_\F$ as it will collapse to an equality between the boundaries of its endpoints, specified by their type.
By definition we have
$g_{\ell'}(b_{\ell'})(\bar{t},S,R,\bar{r})$ equal to
$((i_{\ell'}.\lemmaD{{\ell'}}[\bar{t},S,R])^*(b_{\ell'}(\bar{t},S,R)))(\bar{r})$,
which is defined as an $\hcomp$ and so by filling it is equal to the base of the composition $b_{\ell'}(\bar{t},S,R,\bar{r})$, let us call this path $p_2$. Note that when $\phi_{\ell'}[\bar r] = 1_\F$ we will have $p_2 = \lambda i_{\ell'}.\, \lemmaD{{\ell'}}[\bar{t},S,R]$.
Finally $b_{\ell'}(\bar{t},S,R,\bar{r})$ is equal to
\[\lambda \kappa. \lambda j.\,\con_{\ell'}(\bar{t},\lambda \xi.\, \capp{R\,(\lconst \xi)}\,j,\bar{r}),\]
which is path equal to the left hand side by $\lemmaDg{\ell}{M'}$ and $\con_{\ell'}$ itself, we call the resulting path $p_3$. Note that when $\phi_{\ell'}[\bar r] = 1_\F$ we will have $p_3$ built from $e_{\ell'}$ and $\lemmaDg{\ell}{M'}$ instead.
By transitivity, $p_3 \cdot p_2 \cdot p_1$ forms a path between the left and right hand sides.

To preserve judgemental equality, when $\phi_{\ell'}[\bar r] = 1_\F$, the path $\lemmaDg{\ell}{M}$ must be
equal to $\lemmaDg{\ell}{e_{\ell'}[\bar{t},\lambda \xi.\, M',\bar{r}]}$ which in turn is equal to
$\lemmaDg{{\ell'}}{e_{\ell'}}[\bar{t},S,R,R_K,\tilde{R},\bar{r}]$ by the commuting with substitution property, and where
$R_K = \lambda \xi. \lambda \kappa.\lambda j.\, M'\sigma$ and $\tilde{R}$ is given by $\lemmaDg{\ell}{M'}$. To build the necessary path we first contract the singleton pair $(R_K,\tilde{R})$, so that the only non-trivial path in the composition is $p_2$, which as we noted matches $\lemmaD{{\ell'}}$, i.e.
\[\lemmaDg{{\ell'}}{e_{\ell'}}[\bar{t},S,R,R,\xi.\,\mathsf{refl},\bar{r}].\]
Using this, we define $\lemmaDg{\ell}{M}$ as an $\hcomp$ of the path just defined and $p_3 \cdot p_2 \cdot p_1$.
\end{proof}

\begin{proofof}{\autoref{prop:cirr:HITs}}
The $v_\ell$ terms from lemma \ref{lem:cirrhits:boundary} collectively form a proof of $\mathcal{E}(\delta,D)$, so by lemma \ref{lem:pi-eliminators} we conclude $\Pi(t : \forall \kappa. \HIT{H}{\delta}). D[t] \jeq \Pi(t : \forall \kappa. \HIT{H}{\delta}). \forall \kappa. \capp[\kappa_0] t = \capp t$, as required.
\end{proofof}

\subsection{Composition Structure for Higher Inductive Types (Sect.~\ref{sec:hits:syn})}
\label{appendix:trans}
Following \cite{CTTHITS}, for any HIT $\istype {\delta : \Delta} {\HIT{H}{\delta}}$ we define its composition operation,
$\comp$, in terms of the $\trans$ and $\hcomp$ operations, resulting in the following judgemental equality rule
\[
\comp^i_{\HIT{H}{\delta}}\,[\phi \mapsto u]\,u_0 \jeq
  \hcomp^i_{\HIT{H}{\delta\subst i 1}}\,[\phi \mapsto v(i) ]\,
     (\trans^i_{\HIT{H}{\delta}}\,\phi\,u_0)
\]
where $v(i)$ is $\trans^j_{\HIT{H}{\delta \subst i {i \vee j}}}\,(\phi \vee i = 1)\,(u\,i)$.

Furthermore we include judgemental equalities for $\trans$ when applied
to elements of the HIT built by homogeneous composition or by
constructors.
In Section~$3.4$ of \cite{CTTHITS}, the authors describe how $\trans$ computes when applied to constructors specified by a signature of the form
\[
\mathsf{c} : (\bar{x} : \bar{A}(\delta)) (\bar{i} : \I) \to \HIT{H}{\delta} [ \phi \mapsto e]
\]
where $\bar{A}(\delta)$ is a telescope including both non-recursive
and recursive arguments. They are able to treat both kinds of
arguments uniformly by working in a variation of cubical type theory
where $\trans$ and $\hcomp$ are the primitive operations for all types
while $\comp$ is derived, so that they can use
$\trans^i_{\bar{A}(\delta)}$ to transport all the arguments at once.
We have kept $\comp$ as the primitive in our type theory, but we can
reuse their description as long as we show how to transport the
arguments for the constructors in our schema, i.e., provide a
replacement for $\trans^i_{\bar{A}(\delta)}$.

Given a constructor declaration $(\Gamma,\bar{\Xi},\Psi,\phi,e)$ and $\delta : \I \to \Delta$ we have to define $\trans^i_{\Gamma[\delta\,i],\Theta_{\bar{\Xi}[\delta\,i],\HIT{H}{(\delta\,i)}}}\,\psi\,(\bar{t},\bar{a})$. The $\Gamma[\delta\,i]$ part of the telescope can be dealt with using $\comp$, as shown in \cite{CTTHITS} by the definition of $\ctrans$: a transport operation derived from composition.
We are left with having to define $\trans^i_{\Theta_{\bar{\Xi}[\delta\,i,\tilde{t}[i]],\HIT{H}{\delta\,i}}}\,\psi\,(\bar{t},\bar{a})$ where $\tilde{t}$ connects $\bar{t}$ to the result of transporting it. It is then sufficient to show how to transport elements of
$C(i) := \Xi_k[\delta\,i,\tilde{t}[i]] \to \HIT{H}{\delta\,i}$ for each $\Xi_k$ in $\bar{\Xi}$.
Here we follow the recipe for transport in function types \cite{HCompNote}, like so
\[
\trans^i_{C(i)}\,\psi\,a_k = \lambda \xi.
\trans^i_{\HIT{H}{\delta\,i}}\,\psi\,(a_k\,(\ctrans^i_{\Xi_k[\delta\,(1 - i),\tilde{t}[1 - i]]}\,\psi\,\xi))
\]
where we make use of the fact that $a_k\,(-)$ is a subtree of $\con(\bar{t},\bar{a},\bar{i})$ so it is well-founded to recursively transport.
On top of the above, the transport operation commutes
with homogeneous composition, as described in Sect. 3.2 of \cite{CTTHITS}.

\subsection{Detailed version of Section~\ref{sec:model}}
\label{app:model}

The standard models of both Cubical Type Theory and Clocked Type Theory are based
on presheaf categories. In this section we recall these models and show how to
model the combined \cctt\ in a presheaf category over the
product of the categories used in the interpretations of Cubical and Clocked Type Theory.
One of the challenges in constructing the model is to equip the types of Clocked Type Theory
(such as $\latbind \tickA\kappa A$) with composition structures as required to model
types in Cubical Type Theory.

In this paper, following the convention of~\citet{clottmodel} (but breaking with the convention of
\citet{CTT}) we will work with \emph{covariant} presheaves. Recall that a covariant presheaf over a
category $\cat C$ is a family of sets $X(c)$ indexed by objects of $\cat C$ together with
a map mapping $f : c \to c'$ in $\cat C$  and $x \in X(c)$ to $f\cdot x \in X(c')$,
respecting identities and composition:
\begin{align} \label{eq:functoriality}
  \id_{c} \cdot x & = x &
 g \cdot (f\cdot x) & = (gf)\cdot x
\end{align}
We write $\PSh{\cat C}$ for the category of covariant presheaves on $\cat C$.
We recall the notion of Category with Family (CwF)~\cite{dybjer1996} a standard notion of
model of dependent type theory.

\begin{definition} \label{def:cwf}
 A category with family comprises:
 \begin{enumerate}
\item A category $\cat C$. We use $\Gamma, \Delta$ to range over objects of $\cat C$
\item A functor $\Fam : \opcat{\cat C} \to \Set$. Elements $A \in \Fam(\Gamma)$ are
referred to as \emph{families} over $\Gamma$. If $\sigma : \Delta \to \Gamma$
and $A \in \Fam(\Gamma)$ we write $A[\sigma]$ for $\Fam(\sigma)(A)$
\item A functor $\trm$ associating to each $\Gamma$ and each $A\in \Fam(\Gamma)$
a set $\trm (A)$ of \emph{elements} of $A$, and to each $\sigma : \Delta \to \Gamma$ a
mapping $(-)[\sigma] : \trm(A) \to \trm{(A[\sigma])}$.
\item A \emph{comprehension} operation mapping a family $A$ over $\Gamma$
to an object $\compr\Gamma A$ such that maps $\Delta\to \compr\Gamma A$
correspond bijectively to pair of maps $\sigma : \Delta \to \Gamma$ and elements
$t \in A[\sigma]$, naturally in $\Delta$.
\end{enumerate}
\end{definition}
We often refer to a CwF simply by the name of the underlying category $\cat C$. A CwF gives
rise to a model of dependent type theory~\cite{Hofmann1997} in which contexts are modelled
as objects of $\cat C$, types as families and terms as elements.  Recall that any presheaf
category is the underlying category of a CwF where families over an object $\Gamma$
are indexed families of sets $X(c, \gamma)$ for $\gamma \in \Gamma(c)$ with maps
$f\cdot (-) : X(c, \gamma) \to X(c', f\cdot \gamma)$ satisfying the equations
(\ref{eq:functoriality}), and elements  are assignments mapping
each $\gamma \in \Gamma(c)$ to $t(\gamma) \in X(c,\gamma)$ such that
$t(f\cdot \gamma) = f\cdot t(\gamma)$.

We recall also the category $\int \Gamma$ of elements for a covariant presheaf $\Gamma$
over a category $\cat C$. This has as objects pairs $(c, \gamma)$ where $\gamma \in \Gamma(c)$
and morphisms from $(c, \gamma)$ to $(c', \gamma')$ morphisms $f: c \to c'$ such that
$f \cdot \gamma = \gamma'$. Note that a family over $\Gamma$ is precisely the same as a
presheaf over $\int \Gamma$. We will also use the equivalence
\begin{equation} \label{eq:slice:presheaf}
 \slice{\PSh{\cat C}}{\Gamma} \equi \PSh{\textstyle{\int} \Gamma}
\end{equation}
which states that a slice of a presheaf category is itself a presheaf category.

\subsubsection{Modelling Cubical Type Theory}
\label{sec:model:cubical}

The standard model of Cubical Type Theory~\cite{CTT} uses presheaves over the
\emph{category of cubes}. Here, since we use covariant presheaves, we will write $\CubeCat$
for the opposite of the category used by \citet{CTT}. So $\CubeCat$ has as objects
finite sets $I, J ,K$ and as morphisms from $I$ to $J$ maps $f : I \to \dm J$ where
$\dm J$ is the free de Morgan algebra on the set $J$. Composition is the standard Kleisli
composition, using the fact that $\dm{-}$  is a monad.
From the model construction we recall in particular the interval
object $\sI(I) \defeq \dm{I}$, representing the singleton set, and the face lattice
$\sF$, which at stage $I$ is the free distributive lattice generated by
elements $(i = 0)$ and $(i = 1)$ for each $i \in I$, and relations $(i
= 0) \wedge (i = 1) = \bot_\sF$.

This model construction can be extended to work for any
category of the form $\PSh{\CubeCat \times \cat D}$.
To do so we rely on the Orton and Pitts' axiomatisation of models of CTT~\cite{OPAx,LOPS}.
Rather than recalling these axioms, we recall the sufficient conditions summarised by
\citet{sheafModels20} for a presheaf topos to satisfy these axioms.
These conditions are as follows.
\begin{itemize}
\item The interval object $\sI$ is connected, and a bounded distributive
  algebra structure with distinct 0 and 1 elements.
  Exponentiation by $\sI$ has a right adjoint.
\item The canonical map from $\sF$ to the subobject classifier is a
monomorphism, and the universal cofibration $\top : 1 \to \sF$ is a levelwise
  decidable inclusion. Monomorphisms that can be described as pullbacks
  of the latter are referred to as \emph{cofibrations}.
  The interval endpoint inclusions $0, 1 : 1
  \to \sI$ must be cofibrations, and cofibrations must be closed under finite
  union (finite disjunction), composition (dependent conjunction), and
  universal quantification over $\sI$.
\end{itemize}
A \emph{cubical model} is a presheaf category satisfying the axioms above.
In $\PSh{\CubeCat \times \cat D}$, defining both objects as constant
on the $\cat D$ component, using the corresponding objects in
$\PSh{\CubeCat}$
\begin{align*}
  \sI(I,d) &\defeq \sI(I) &
  \sF(I,d) & \defeq \sF(I)
\end{align*}
make $\PSh{\CubeCat \times \cat D}$  a cubical model,
as also observed by \citet{sheafModels20}.

Given such a model, \citet{OPAx} express
the structure sufficient to model CTT using the internal language
of the presheaf topos as an extensional type theory.
We recall here some definitions that will be relevant later. We assume a
$\omega+1$ long hierarchy of Grothendieck universes, leading to a
corresponding hierarchy of universes a la Russell $\Univ_i$ in the
internal language, each classifying presheaves of the appropriate size.

\begin{definition}
A \emph{CCHM fibration} $(A,\alpha)$ over a type $\Gamma : \Univ_\omega$ is a family
$A : \Gamma \to \Univ_\omega$ with a fibration structure $\alpha : \mathsf{isFib}\,\Gamma\,A$ where
\begin{align*}
  \mathsf{isFib}\,\Gamma\,A \defeq &\, (e : \{0,1\})(p : \sI \to \Gamma) \to \mathsf{Comp}\,e\,(A \circ p)\\
  \mathsf{Comp}\,e\,A \defeq &\, (\phi : \sF)(u : [\phi] \to (i : \sI) \to A\,i) \\
  & \to \{ u_0 : A\,e \mid \phi \Rightarrow u\,e = u_0 \}\\
  & \to \{ u_1 : A\,\bar{e} \mid \phi \Rightarrow u\,\bar{e} = u_1 \}
\end{align*}
where $\bar{e}$ sends $0 : \sI$ to $1$ and vice versa.
\end{definition}
Notice
$\mathsf{Comp}\,0\,A$ closely matches the signature of the composition
operation from CTT.

We can then build the following Category with Families of fibrant types, which we call $\fibCwF$:
\begin{itemize}
\item A context $\Gamma$ is a global element of $\Univ_\omega$.
\item A family over $\Gamma$ is a global CCHM fibration over $\Gamma$.
\item An element $t$ of a family $(A,\alpha)$ over $\Gamma$ is a global element of $(\gamma : \Gamma) \to A\,\gamma$.
\end{itemize}
it then follows that $\fibCwF$ models the type formers of Cubical Type Theory.

\begin{remark}
  Cubical Type Theory as presented by \citet{CTT} includes specific
  judgemental equalities for the composition operator $\comp^i_A$
  according to the shape of $A$, matching the behaviour of the given fibration structures in their model.
  The constructions by \citet{OPAx} do not necessarily
  produce fibration structures that satisfy the same equalities.
  One way to deal with this is to use
  alternative formulations of these rules~\cite{HCompNote,CubicalAgda}.
  Since these equalities are mainly relevant for operational
  properties of CTT, we will ignore this issue in this paper.
\end{remark}
\subsubsection{Modelling Clocked Cubical Type Theory}

In the previous section we saw how to model Cubical Type Theory in any category of
the form $\PSh{\CubeCat \times \cat D}$. We now define the category $\TimeCat$ of
\emph{time objects} and extend the model to a model of \cctt\
in $\PSh{\CubeCat \times \TimeCat}$. The objects of $\TimeCat$ are pairs
$\timeobj{\FSA}\delta$, where $\FSA$ is a finite set (to be though of as a set of semantic clocks),
and $\delta : \FSA \to \nats$ is a map associating to each clock a finite amount of time steps
that are left on that clock. A morphism $\sigma : \timeobj{\FSA}\delta \to \timeobj{\FSB}{\delta'}$
is a map $\sigma : \FSA \to \FSB$ such that $\delta' \sigma \leq \delta$ in the
pointwise order. This can be understood as a generalisation of the topos of trees
model $\PSh{\opcat{\omega}}$ of guarded recursion~\cite{ToT},
where the indexing category is restricted to objects where the first component $\FSA$
is a singleton. The category $\PSh{\TimeCat}$ has previously been used to model
type theories with guarded recursion and multiple clocks~\cite{GDTTmodel,clottmodel},
and (in a slight variation) Guarded Computational Type Theory~\cite{sterling2018guarded}.

The category $\PSh{\CubeCat \times \TimeCat}$ has an object of clocks $\clk$ defined as
$\clk(I, \timeobj{\FSA}{\delta}) = \FSA$. This can be used to model assumptions of the
form $\kappa : \clocktype$ in contexts, and the types $\forall\kappa . A$ can be
modelled as $\Pi$-types. To model $\later$, recall that this is a Fitch-style modal
type operator and that these can be modelled using dependent right adjoint types~\cite{drat}.

\begin{definition} \label{def:drajoint}
 Let $\cat C, \cat D$ be CwFs and let $\L : \cat C \to \cat D$ be a functor between
 the underlying categories. A dependent right adjoint to $\L$ is a mapping associating
 to a family $A$ over $\L\Gamma$ a family $\R A$ over $\Gamma$ and a bijective
 correspondence between elements of $A$ and elements of $\R A$, both
 natural in $\Gamma$.
\end{definition}

The naturality requirement for $\R$ means that if $\gamma : \Gamma' \to \Gamma$
and $A$ is a family over $\L \Gamma$, then $(\R A)[\sigma] = \R(A [\L\sigma])$.
Writing $\overline{(-)}$ for both directions of the bijective correspondence on elements,
the naturality condition for this means that
if 
$t$ is an element of $A$ over $\L\Gamma$, then
$\overline t[\gamma] = \overline{t[\L\gamma]}$,
and similarly for the opposite direction.

Given an endofunctor $\L$ on a CwF $\cat C$ as well as a dependent right adjoint $\R$
to $\L$, one can model a Fitch-style modal operator by modelling
extensions of contexts with ticks by $\L$ and the modal operator
by $\R$. In the case of Clocked Type Theory,
the ticks, as well as the modal operators are indexed over an object $\clk$ in
the model: Semantically, the hypothesis $\hastype{\Gamma}\kappa\clocktype$
of the context extension rule for $\wfcxt{\Gamma, \tickA : \kappa}$
corresponds to an element of the slice category over $\clk$.
By (\ref{eq:slice:presheaf}) the slice category over $\clk$
is itself a presheaf category
and therefore carries a natural CwF structure. In fact, if $\clkel : \Gamma \to \clk$
is an object in the slice category, then families over $\clkset$ in the slice category
correspond bijectively to families in $\PSh{\CubeCat \times \TimeCat}$
over $\Gamma$. Exploiting this, \citet{clottmodel} describe how
to model ticks on clocks using a dependent right adjoint $\slater$ to an endofunctor
$\searlier$ on the slice category over $\clk$. In this model
$\wfcxt{\Gamma, \tickA : \kappa}$ is interpreted as the domain of
$\searlier(\den\Gamma, \den\kappa)$ and $\latbind\tickA\kappa A$ as the
dependent right adjoint $\slater$ applied to $\den A$.
The bijective correspondence on elements then models the
bijective correspondence between terms $\hastype{\Gamma, \tickA : \kappa}{t}A$ and
terms $\hastype{\Gamma}u{\latbind\tickA\kappa A}$ given by tick abstraction and application.
Tick weakening, which syntactically corresponds to context projections
$\Gamma, \tickA : \kappa \to \Gamma$ can be modelled using a natural transformation
from $\searlier$ to the identity.

\subsubsection{Composition structure for dependent right adjoints}
\label{sec:comp:dep:right:adj}

Before recalling the dependent right adjoint structure on
the slice category $\slice{\PSh{\CubeCat \times \TimeCat}}{\clk}$
we now give general conditions ensuring that dependent right adjoint types
on a cubical model 
carry composition structure. The conditions are all on the left adjoint functor.

\begin{definition} \label{def:pres:cofib}
 Let $\cat C$ be a cubical model
 with interval object $\sI$ and let $\L : \cat C \to \cat C$ be a finite product
 preserving functor.
 \begin{enumerate}
  \item We say $\L$ preserves the interval if there is an isomorphism $\L(\sI) \cong \sI$
    preserving endpoints, i.e., such that the following commutes for $e = 0,1$.
    \[
      \begin{tikzcd}
        1 \ar{r}{e} \ar{d}{\cong} & \sI \ar{d}{\cong} \\
        \L1 \ar{r}{\L e} & \L\sI
      \end{tikzcd}
    \]
   \item We say $\L$ \emph{preserves cofibrations} if $\L i : \L A \to \L B$ is a cofibration
   whenever $i : A \to B$ is, and if the diagram on the right below is a pullback whenever
   $i$ is a cofibration and the diagram on the left is a pullback
    \[
      \begin{tikzcd}
        C \ar{r}{a} \ar{d}{j} & A \ar{d}{i} \\
        D \ar{r}{b} & B
      \end{tikzcd} \hspace{1cm}
      \begin{tikzcd}
        \L C \ar{r}{\L a} \ar{d}{\L j} & \L A \ar{d}{\L i} \\
        \L D \ar{r}{\L b} & \L B
      \end{tikzcd}
    \]
\end{enumerate}
\end{definition}
The condition of preserving cofibrations corresponds to giving an operation
mapping cofibrations $\hastype{\Gamma}{\phi}{\sF}$ on $\Gamma$ to
cofibrations $\hastype{\L\Gamma}{\LCof(\phi)}{\sF}$  such that
$\compr{\L\Gamma}{[\LCof(\phi)]} \cong \L(\compr{\Gamma}{[\phi]})$ as
subobjects of $\L\Gamma$ and satisfying $\LCof(\phi[\sigma]) = \LCof(\phi)[\L\sigma]$.
Here $[\phi]$ is the family classified by $\phi$.

\begin{theorem} \label{thm:dep:right:adj}
 Let $\cat C$ be a cubical model, let $\L : \cat C \to \cat C$
 be a functor preserving finite products, the interval and cofibrations, and let $\R$ be
 a dependent right adjoint to $\L$. If a family $A$ over $\L \Gamma$ carries a global
 composition structure, so does $\R A$ over $\Gamma$. Moreover, this assignment
 is natural in $\Gamma$.
\end{theorem}

Note that this is a statement about \emph{global} composition structures in the model.
The theorem can not be proved in the internal logic of the topos, but can be proved
in an extension of this using crisp type theory, similarly to
the construction of universes for cubical type theory in crisp type theory~\cite{LOPS}.
The reason is that the proof uses the bijective correspondence of
Definition~\ref{def:drajoint} which only applies to global terms.

To prove Theorem~\ref{thm:dep:right:adj}, we need the following lemma
giving an alternative description of the data of a composition structure.
The lemma uses standard CwF notation writing $\p : \compr{\Delta}A \to \Delta$
for the projection out of a comprehension object, and (less standard) notation
$\compr{\sigma}{A} : \compr\Delta{A[\sigma]} \to \compr\Gamma{A}$ if
$\sigma : \Delta \to \Gamma$ for the functoriality of comprehension in the
first component defined in the language of CwFs as $\pair{\sigma\p}{\q}$.

\begin{lemma} \label{lem:comp:struct:alternative}
 Let $\Gamma$ be a global element of $\Univ_\omega$ and let
 $A : \Gamma \to \Univ_\omega$. To give a composition structure
 on $A$ corresponds to giving, for each global element $\Delta$ of
 $\Univ_\omega$ and map $\sigma : \Delta \times \sI \to \Gamma$,
 an assignment expressed as the rule
 \[
   \inferrule*
  {
   \hastype{\Delta}{\phi}{\sF}
   \and
   \hastype{\compr{\compr{\Delta}{\sI}}{[\phi[\p]]}}{u}{A[\sigma\p]}
   \\
   \hastype{\Delta}{u_e}{A[\sigma \circ \pair{\id}e]}
   \and
   \hasnotype{\compr{\Delta}{[\phi]}}{u_e[\p] = u[\compr{\pair{\id_{\Delta}}{e}}{[\phi]}]}{}
  }
  {\hastype{\Delta}{c_{\sigma}\, \phi\, u\, u_e}{A[\sigma \circ \pair{\id}{1-e}]}}
 \]
 for each $e \in \{0,1\}$, natural in $\Delta$ satisfying
 \[
  \hasnotype{\compr{\Delta}{[\phi]}}{(c_{\sigma}\, \phi\, u\, u_e)[\p] = u[\compr{\pair{\id}{1-e}}{[\phi]}]}{}
 \]
\end{lemma}


\begin{proof}
 Given $e$, to give the part of $\mathsf{isFib}$ corresponding to $e$ corresponds to giving
 $\hastype{p: \sI \to \Gamma}c{\mathsf{Comp}\,e\,(A\circ p)}$ in the model. This in turn
 corresponds to an assignment of morphisms $\tau : \Delta \to (\sI \to \Gamma)$ to terms
 $\hastype{\delta : \Delta}{c_\tau}{\mathsf{Comp}\,e\,(A\circ \tau(\delta))}$ natural in $\Delta$.
 By uncurrying, the latter corresponds to an assignment mapping $\sigma : \Delta \times \sI \to \Gamma$ to
 $\hastype{\delta : \Delta}{c_\sigma}{\mathsf{Comp}\,e\,(A\circ \sigma(\delta, -))}$, also natural
 in $\Delta$. By further uncurrying the arguments to the composition operator
 and using similar naturality arguments in each case, we arrive at the description in the lemma.
\end{proof}

\begin{proofof}{Theorem~\ref{thm:dep:right:adj}}
 By Lemma~\ref{lem:comp:struct:alternative} it suffices to give an assignment
 mapping $\sigma : \Delta \times \sI \to \Gamma$ to a rule
  \[
   \inferrule*
  {
   \hastype{\Delta}{\phi}{\sF}
   \and
   \hastype{\compr{\compr{\Delta}{\sI}}{[\phi[\p]]}}{u}{(\R A)[\sigma\p]}
   \\
   \hastype{\Delta}{u_e}{(\R A)[\sigma \circ \pair{\id}e]}
   \and
   \hasnotype{\compr{\Delta}{[\phi]}}{u_e[\p] = u[\compr{\pair{\id_{\Delta}}{e}}{[\phi]}]}{}
  }
  {\hastype{\Delta}{c^{\R A}_{\sigma}\, \phi\, u\, u_e}{(\R A)[\sigma \circ \pair{\id}{1-e}]}}
 \]
 natural in $\Delta$ and satisfying the equality of Lemma~\ref{lem:comp:struct:alternative}.
 Using $(\R A)[\sigma\p] = \R(A[\L(\sigma\p)])$ the assumptions correspond to
\begin{align*}
    \hastype{\L(\compr{\compr{\Delta}{\sI}}{[\phi[\p]]})}{\overline u}{A[\L(\sigma\p)]}
   & &
   \hastype{\L\Delta}{\overline{u_e}}{A[\L(\sigma \circ \pair{\id}e)]}
\end{align*}
satisfying
\begin{equation} \label{eq:boundary:ubar}
\hasnotype{\L(\compr{\Delta}{[\phi]})}{\overline{u_e}[\L\p] =
\overline u[\L(\compr{\pair{\id_{\Delta}}{e}}{[\phi]})]}{}
\end{equation}
Since $\L$ preserves cofibrations, by the notation introduced after
Definition~\ref{def:pres:cofib},
$\L(\compr{\compr{\Delta}{\sI}}{[\phi[\p]]} \cong
\compr{\L(\compr{\Delta}{\sI})}{[\LCof(\phi[\p])]}$ as subobjects of
$\L(\compr{\Delta}{\sI})$. For simplicity we will leave this isomorphism implicit.
Moreover, since $\L$ preserves finite products
and the interval, there is an isomorphism
$\xi_{\Delta} : \compr{\L\Delta}{\sI} \cong \L(\compr\Delta\sI)$
natural in $\Delta$ such that
\begin{equation} \label{eq:L:sI:endpoints}
      \begin{tikzcd}
        & \L \Delta \ar{ld}[swap]{\pair{\id{}}e} \ar{rd}{\L\pair{\id{}}e}  \\
        \compr{\L\Delta}{\sI} \ar{rr}{\xi_{\Delta}} && \L(\compr\Delta\sI)
      \end{tikzcd}
\end{equation}
Then, since $\L(\sigma\p)\circ \compr{\xi_\Delta}{[\LCof(\phi[\p])]} = \L(\sigma) \circ \xi_\Delta\circ \p$
\begin{align*}
    \hastype{\compr{\compr{\L(\Delta)}{\sI}}{[\LCof(\phi)[\p]]}}{\overline u[\compr{\xi_\Delta}{[\LCof(\phi[\p])]}]}{A[\L(\sigma) \circ \xi_\Delta][\p]}
\end{align*}
and by (\ref{eq:L:sI:endpoints})
\begin{align*}
   \hastype{\L\Delta}{\overline{u_e}}{A[\L(\sigma)\circ \xi_\Delta][\pair{\id}e]}
\end{align*}
We can therefore apply the composition structure for $A$ as in
Lemma~\ref{lem:comp:struct:alternative} in the case of
$L(\sigma)\circ \xi_\Delta : \compr{\L\Delta}{\sI} \to \L\Gamma$,
the cofibration $\hastype{\L\Delta}{\LCof(\phi)}{\sF}$,
and the terms $\overline u[\compr{\xi_\Delta}{[\LCof(\phi[\p])]}]$ and
$\overline{u_e}$, if only we can prove that
\[
  \hasnotype{\compr{\L \Delta}{[\LCof(\phi)]}}{\overline{u_e}[\p] =
  \overline u[\compr{\xi_\Delta}{[\LCof(\phi[\p])]}][\compr{\pair{\id_{\L\Delta}}{e}}{[\LCof(\phi)]}]}{}
\]
Up to the isomorphism $\compr{\L \Delta}{[\LCof(\phi)]} \cong \L(\compr{\Delta}{[\phi]})$ the context
projection $\p$ equals $\L\p$, and so in context $\L(\compr{\Delta}{[\phi]})$ the left hand side
reduces to $\overline{u_e}[\L\p]$.
The right hand side reduces using (\ref{eq:L:sI:endpoints}) to
\[
 \overline u[\compr{\L\!\pair{\id_{\Delta}}{e}}{\LCof(\phi)}]
\]
which up to the isomorphism $\compr{\L \Delta}{[\LCof(\phi)]} \cong \L(\compr{\Delta}{[\phi]})$
equals \[\overline u[\L(\compr{\pair{\id_{\Delta}}{e}}{[\phi])}],\] and so the required
equality follows from (\ref{eq:boundary:ubar}). The composition structure for $A$ therefore gives
\[
 \hastype{\L\Delta}
   {c^{A}_{L(\sigma)\circ \xi_\Delta} \, \LCof(\phi)\, \overline u[\dots]\, \overline{u_e}}
   {A[L(\sigma)\circ \xi_\Delta \circ \pair{\id}{1-e}]}
\]
which, since $A[L(\sigma)\circ \xi_\Delta \circ \pair{\id}{1-e}] = A[L(\sigma\circ \pair{\id}{1-e})]$,
corresponds to a term
\[
\hastype{\Delta}{c_\sigma^{\R A} \, \phi\, u\, u_e}{(\R A)[\sigma\circ \pair{\id}{1-e}]}
\]
To show the equality
\[
\hasnotype{\compr{\Delta}{\phi}}{(c_\sigma^{\R A} \, \phi\, u\, u_e)[\p] =
 u[\compr{\pair{\id}{1-e}}{[\phi]}]}{}
\]
is equivalent to showing
\[
\hasnotype{\L(\compr{\Delta}{\phi})}{\overline{(c_\sigma^{\R A} \, \phi\, u\, u_e)[\p]} =
 \overline{u[\compr{\pair{\id}{1-e}}{[\phi]}]}}{}
\]
which up to the isomorphism $\L(\compr{\Delta}{\phi})\cong \compr{\L\Delta}{[\LCof(\phi)]}$
corresponds to showing that the term
\begin{equation} \label{eq:cRA:lhs}
 \hasnotype{\compr{\L\Delta}{\LCof(\phi)}}
   {(c^{A}_{L(\sigma)\circ \xi_\Delta} \, \LCof(\phi)\, \,
     \overline u[\compr{\xi_\Delta}{[\LCof(\phi[\p])]}]
   \,\, \overline{u_e})[\p]}
   {A[L(\sigma)\circ \xi_\Delta \circ \pair{\id}{1-e}]}
\end{equation}
equals
\begin{equation}\label{eq:cRA:rhs}
 \hasnotype{\compr{\L\Delta}{\LCof(\phi)}}{\overline u[\compr{\L(\pair\id{1-e})}{[\LCof(\phi)]}]}{}
\end{equation}
By the equality rule for the composition structure on $A$, (\ref{eq:cRA:lhs}) equals
\[
 \overline u[\compr{\xi_\Delta}{[\LCof(\phi[\p])]}][\compr{\pair\id{1-e}}{[\LCof(\phi)]}]
\]
which equals (\ref{eq:cRA:rhs}) by (\ref{eq:L:sI:endpoints}).
\end{proofof}

\subsubsection{The dependent right adjoint}
\label{sec:dep:right:adj}


We now specialise to the particular cubical model structure on $\PSh{\CubeCat\times \TimeCat}$.
The only assumption we will need for this is that $\sI$ and $\sF$ are included from
$\PSh{\CubeCat}$ as in the cubical model structures used by~\citet{sheafModels20}.

We first recall the structure of the dependent right adjoint in details.
\citet{clottmodel} define this structure for $\PSh{\TimeCat}$ but the constructions
carry over directly to $\PSh{\CubeCat\times \TimeCat}$.
The structure arises as an extension of an endo-adjunction on the slice category
as in the following lemma slightly generalised from~\citet{drat}, to which we refer for
details.

\begin{lemma} \label{lem:dradjoint}
 Let $\cat C$ and $\cat D$ be CwFs and let $\L : \cat C \to \cat D$ be a functor between the
 underlying categories with a right adjoint $\R$. Suppose $\R$ extends to families and elements as
 in the following data
 \begin{enumerate}
\item An operation mapping families $A$ over $\Gamma$ in $\cat D$ to families
$\RFam (A)$ over $\R \Gamma$ satisfying 
$\RFam(A[\gamma]) = (\RFam (A))[R\gamma]$
\item An operation mapping elements $t$ of $A$ to elements
$\REl (t)$ of $\RFam (A)$ satisfying 
$\REl(t[\gamma]) = (\REl (t))[R\gamma]$.
\end{enumerate}
Then $\L$ has a dependent right adjoint mapping families $A$ over $\L \Gamma$ to
$\R A = (\RFam A)[\eta]$ where $\eta : \Gamma \to \R \L \Gamma$ is the unit of
the adjunction.
\end{lemma}

The endoadjunction on the slice category is best described by using the equivalent description
of the slice category as $\PSh{\int \clk}$. The right adjoint is the simplest to describe and
is similar to the functor $\slater$ on the topos-of-trees~\cite{ToT}:
\[
 \slater\Gamma(I,\timeobj{\FSA}{\delta}, \lambda) =
\begin{cases}
 \Gamma(I,\timeobj{\FSA}{\subex\delta{\lambda}{n}}, \lambda) & \text{if } \delta(\lambda) = n\!+\!1 \\
 1 & \text{if }\delta(\lambda) = 0
\end{cases}
\]
Here $\subex\delta{\lambda}{n}(\lambda) = n$
and $\subex\delta{\lambda}{n}(\lambda') = \delta(\lambda')$ for $\lambda' \neq \lambda$
and $1$ in the second clause is a singleton set. This lifts to families and elements in the sense
of Lemma~\ref{lem:dradjoint}, for example, if $A$ is a family over $\Gamma$ then
\[
 \slater_{\Fam}(A)(I,\timeobj{\FSA}{\delta}, \lambda)(\gamma) =
\begin{cases}
 A(\gamma) & \text{if } \delta(\lambda) = n\!+\!1 \\
 1 & \text{if }\delta(\lambda) = 0
\end{cases}
\]
The left adjoint can be concretely described as
$\searlier \Gamma(I,\timeobj{\FSA}{\delta}, \lambda)$ having as elements pairs $\pair{\sigma}{x}$
such that $\pair{\id_I}{\sigma} : (I,\timeobj{\FSB}{\delta'}, \lambda')
\to (I,\timeobj{\FSA}{\delta}, \lambda)$ with $\delta'(\lambda') > \delta (\lambda)$ and
$x \in \Gamma(I,\timeobj{\FSB}{\delta'}, \lambda')$ considered up to the
equivalence relation $\sim$ generated by $\pair{\sigma\tau}{x} \sim \pair{\sigma}{\pair\id\tau \cdot x}$.
We refer to~\citet{clottmodel} for the details of the adjunction structure as well as an
abstract description of the left adjoint.

There is a natural transformation $\psearlier : \searlier \to \id$ defined as $\psearlier\pair{\sigma}x =
\pair\id\sigma \cdot x$. This is used in our model to interpret tick-weakening.
%

\begin{theorem} \label{thm:earlier:comp}
 Let $\Gamma$ be an object of $\PSh{\CubeCat \times \TimeCat}$ and $\kappa : \Gamma \to \clk$.
 Suppose $A$ is a family over the domain of $\searlier(\Gamma, \kappa)$ which carries a composition
 structure, then $\slater A$ carries a composition structure as a family over $\Gamma$.
\end{theorem}

\begin{lemma} \label{lem:int:clk:coprod}
 The category $\int \clk$ for $\clk$ considered an object in $\PSh{\TimeCat}$ has coproducts.
\end{lemma}

\begin{proof}
 The coproduct of $(\timeobj{\FSA}{\delta}, \lambda)$ and $(\timeobj{\FSB}{\delta'}, \lambda')$
 is the object $(\timeobj{\FSC}{\delta''}, \lambda'')$ where $\lambda''$ is fresh, $\FSC$ is the
 disjoint union of $\FSA \setminus \{\lambda\}$ and $\FSB \setminus \{\lambda'\}$ and $\{\lambda''\}$
 and $\delta''$ agrees with $\delta$ on $\FSA \setminus \{\lambda\}$, with $\delta'$ on
 $\FSB \setminus \{\lambda'\}$ and maps $\lambda''$ to the minimum of $\delta(\lambda)$ and
 $\delta'(\lambda')$.
\end{proof}

\begin{proofof}{Theorem~\ref{thm:earlier:comp}}
 Recall that families and terms in the CwF of the slice category $\PSh{\int\clk}$ over an object
 corresponding to an element $\kappa : \Gamma \to \clk$ correspond bijectively to families and
 elements of the CwF of $\PSh{\CubeCat \times \TimeCat}$ over $\Gamma$. Since this
 correspondence also respects the interpretation of the internal dependent type theory, it suffices
 to show that $\slater A$ carries a composition structure as expressed on the slice category, when
 $A$ does. By Theorem~\ref{thm:dep:right:adj} this reduces to showing that $\searlier$ preserves
 finite products, the interval and cofibrations.

 Each component $\searlier 1(I,\timeobj{\FSA}{\delta}, \lambda)$ of $\searlier 1$
 is easily seen to be inhabited. If $(\sigma, \star)$ and
 $(\tau, \star)$ are two elements of $\searlier 1(I,\timeobj{\FSA}{\delta}, \lambda)$, then both of these are related
 under the equivalence relation used in the definition of $\searlier$ to
 $([\sigma, \tau], \star)$, where $[\sigma, \tau]$ is the copairing of $\sigma$ and
 $\tau$ out of the coproduct of their domains, which exists by Lemma~\ref{lem:int:clk:coprod}.
 So $\searlier$ preserves the terminal object.

 Writing $[\pair\sigma {\pair xy}]$ for the equivalence class represented by $\pair\sigma {\pair xy}$, the map
 \[\searlier (A \times B)(I,\timeobj{\FSA}{\delta}, \lambda) \to (\searlier A \times \searlier B)(I,\timeobj{\FSA}{\delta}, \lambda)\]
 maps $[\pair\sigma {\pair xy}]$ to $\pair{[\pair\sigma x]}{[\pair\sigma y]}$, and the inverse maps
 \[\pair{[\pair\sigma x]}{[\pair\tau y]}\] to $[\pair{[\sigma, \tau]} {\pair {\inl (x)}{\inr (y)}}]$.

 For the interval, the map $\psearlier : \searlier \sI \to \sI$, which (as decribed in the main
 text) is defined as $\psearlier\pair{\sigma}x = \pair\id\sigma \cdot x$ has an inverse which at
 $(I,\timeobj{\FSA}{\delta}, \lambda)$ maps $x \in \sI(I)$ to $\pair \sigma x$ where
 $\sigma : \timeobj{\FSA}{\delta'} \to \timeobj{\FSA}{\delta}$ is tracked by the identity
 and $\delta'$ agrees with $\delta$ everywhere except at $\lambda$ where it is one higher.
 This clearly preserves endpoints.

 For cofibrations, we first show that $\searlier$ preserves pullbacks of cofibrations.
 Suppose that the diagram on the left below is a pullback with $i$ and $j$ cofibrations.
     \[
      \begin{tikzcd}
        C \ar{r}{a} \ar{d}{j} & A \ar{d}{i} \\
        D \ar{r}{b} & B
      \end{tikzcd} \hspace{1cm}
      \begin{tikzcd}
        \searlier C \ar{r}{\searlier a} \ar{d}{\searlier j} & \searlier A \ar{d}{\searlier i} \\
        \searlier D \ar{r}{\searlier b} & \searlier B
      \end{tikzcd}
    \]
We must show that also the diagram in the right is a pullback. Since $\sF$ is constant
in the time dimension, it follows, since $i$ is a fibration that any naturality square of the form
     \[
      \begin{tikzcd}
        A(I,\timeobj{\FSA}{\delta}, \lambda) \ar{r}{\pair{\id{}}{\sigma}\cdot (-)} \ar{d}{i}
        & A(I,\timeobj{\FSB}{\delta'}, \lambda') \ar{d}{i} \\
        B(I,\timeobj{\FSA}{\delta}, \lambda) \ar{r}{\pair{\id{}}{\sigma}\cdot (-)} & B(I,\timeobj{\FSB}{\delta'}, \lambda')
      \end{tikzcd}
    \]
 is a pullback. From this it follows that the square on the right below is a pullback diagram.
      \[
      \begin{tikzcd}
        \searlier C \ar{r}{\searlier a} \ar{d}{\searlier j} & \searlier A \ar{d}{\searlier i} \ar{r}{\psearlier} & A \ar{d}{i} \\
        \searlier D \ar{r}{\searlier b} & \searlier B \ar{r}{\psearlier} & B
      \end{tikzcd}
      \]
  By the pullback pasting diagram it therefore suffices to show that the outer diagram is a pullback,
  which follows again by the pullback pasting diagram applied to
      \[
      \begin{tikzcd}
        \searlier C \ar{r}{\psearlier} \ar{d}{\searlier j} & C \ar{d}{j} \ar{r}{a} & A \ar{d}{i} \\
        \searlier D \ar{r}{\psearlier} & D \ar{r}{b} & B
      \end{tikzcd}
      \]
  and naturality of $\psearlier$. 

 Suppose now that $\chi_A : B \to \sF$ classifies the cofibration $i : A \to B$.
 A similar argument to the one above for $\sI$ shows that $\psearlier : \searlier \sF \to \sF$
 is an isomorphism, which preserves truth. Then
 $\psearlier \circ \searlier (\chi_A)$ classifies $\searlier i$, so also $\searlier i$ is a cofibration.
\end{proofof}

\subsubsection{Interpreting ticks and tick application}
\label{sec:interp:ticks}

A simple tick judgement $\istick\Gamma \kappa u{\Gamma'}$ is interpreted as a map in $\PSh{\int\clk}$ from
$\Gamma$ to $\searlier{\Gamma'}$.
Here, for simplicity, we keep the clock $\kappa$ implicit. An element $t$
of the family $\slater A$ over $\Gamma'$ corresponds bijectively to an element $\overline t$
of $A$ over $\searlier {\Gamma'}$, which can then be reindexed along the tick $u$ to interpret tick application.
A forcing tick judgement $\isbtick{\Gamma}{\kappa}{u}{\Gamma'}$ is interpreted as a map from $\Gamma$
to $\searlier (\compr{\Gamma'}{\clk})$ in $\PSh{\int\clk}$, where the domain is considered an element of the slice
category over $\clk$ with map given by $\kappa$, and $\compr{\Gamma'}{\clk}$ an object with map
given by the second projection. The assumption of the rule for forcing tick application corresponds semantically
to an element of a family $\slater A$ over an object $\compr{\Gamma'}{\clk}$, which corresponds bijectively
to an element of $A$ over $\searlier (\compr{\Gamma'}{\clk})$. Tick application is then interpreted by reindexing
this element along the interpretation of the forcing tick.

To define the interpretation of ticks, say that family $A$ over an object $\Gamma$
is \emph{timeless} if the map
\[
\pair{\searlier(\p)}{\q[\psearlier]}: \searlier(\compr{\Gamma}A) \to \compr{(\searlier\Gamma)}{A[\psearlier]}
\]
is an isomorphism.

\begin{lemma} \label{lem:timeless}
 The families $\sI$ and $\clk$, as well as any cofibration, are timeless. As a consequence
 \[
 \searlier(\den{\Gamma, \cuba{\Gamma'}}) \cong \compr{\searlier(\den{\Gamma})}{\den{\cuba{\Gamma'}}}
 \]
 in $\PSh{\int\clk}$.
\end{lemma}

\begin{proof}
As noted in the proof of Theorem~\ref{thm:earlier:comp}, cofibrations $A$ over $\Gamma$ are constant in the time
dimension in the sense that for any
$\sigma : (\timeobj{\FSA}{\delta}, \lambda) \to (\timeobj{\FSB}{\delta'}, \lambda')$
the map
\[
\pair{\id}\sigma \cdot (-) : A(\gamma) \to A(\pair{\id}\sigma \cdot \gamma)
\]
is an isomorphism.
This means that the inverse to the map of the lemma can be defined to map the element
\[(\pair\sigma\gamma, a)\in (\compr{(\searlier\Gamma)}{A[\psearlier]})(I,\timeobj{\FSA}{\delta},\lambda)\]
to $[(\sigma, (\gamma, a'))] \in \searlier(\compr{\Gamma}A)(I,\timeobj{\FSA}{\delta}, \lambda)$ where
$a'$ is the unique element such that $\pair {\id_I}\sigma \cdot a' = a$. A similar argument applies in the case of $\sI$.

In the case of $\clk$, first note that
\begin{align*}
 \searlier(\compr \Gamma \clk) & = \searlier(\Gamma \times \clk) \cong \searlier\Gamma\times \searlier\clk
\end{align*}
since, as we saw in the proof of Theorem~\ref{thm:earlier:comp}, $\searlier$ preserves finite products.
It thus suffices to show that $\clk \cong \searlier\clk$. In order to prove this, we first prove that any element in
$\searlier(\clk)(I,\timeobj{\FSA}{\delta},\lambda)$
can be represented on the form $\pair{\id}{\kappa}$, for
$\id : (\timeobj{\FSA}{\delta+1},\lambda) \to (\timeobj{\FSA}{\delta},\lambda)$
and $\kappa \in \FSA$. Here $\delta+1$ is the composition of $\delta$ with the successor function. This is true,
because an element represented by $(\sigma, \kappa)$ where
$\sigma : (\timeobj{\FSB}{\delta'},\lambda') \to (\timeobj{\FSA}{\delta},\lambda)$
and $\kappa \in \FSB$ is equivalent to the element given by $(\sigma, \kappa)$ where $\sigma$ now is considered a map from
$(\timeobj{\FSB}{\delta'+1},\lambda')$ to $(\timeobj{\FSA}{\delta},\lambda)$. This element is in turn equivalent to
$\pair \id{\sigma(\kappa)}$ as required. As a consequence, we can define an inverse to $\psearlier$ to map
$\kappa \in \clk(\timeobj{\FSA}{\delta},\lambda)$ to $[\pair{\id}{\kappa}]$.

The last statement of the lemma now follows by induction on $\Gamma'$.
\end{proof}

For $\istick\Gamma \kappa u{\Gamma'}$, let $n$ be the number of ticks on clock $\kappa$
in $\Gamma \setminus \Gamma'$. Then the interpretation of $u$ factors through $\searlier^n \Gamma'$, and
a tick assumption is interpreted as the appropriate projection $\searlier^n \to \searlier$. Likewise
if $\isbtick{\Gamma}{\kappa}{u}{\Gamma'}$ let $n$ be the number of ticks on $\kappa$ in $\Gamma \setminus \Gamma'$,
then the interpretation of $u$ factors through $\searlier^n(\compr{\Gamma'}{\clk})$ with tick variables interpreted
as projections. The construction $\tirr$ is interpreted using the following lemma.

\begin{lemma} \label{lem:interp:tirr}
 For any $n\geq 1$ there exists a unique natural transformation $\alpha : \searlier^n \to \searlier$ such that
 $\psearlier \circ \alpha = \psearlier^n$, and a unique $\beta : \searlier^n(\compr{-}{\clk}) \to \searlier(\compr{-}{\clk})$
 commuting with the projection to $\compr{-}{\clk}$.
\end{lemma}

In particular, this means that $\tirr$ is interpreted using constant paths in our model.
Before proving Lemma~\ref{lem:interp:tirr} we first give an alternative description of $\searlier^n$ similar
to the explicit description for $\searlier$ itself.

\begin{lemma}
 Let $n\geq 0$. Up to isomorphism, the functor $\searlier^n$ maps $\Gamma$ to the presheaf given
 at $(I,\timeobj{\FSA}{\delta}, \lambda)$
 by equivalence classes of pairs $\pair{\sigma}{x}$ such that $\pair{\id_I}{\sigma} : (I,\timeobj{\FSB}{\delta'}, \lambda')
\to (I,\timeobj{\FSA}{\delta}, \lambda)$ with $\delta'(\lambda') \geq \delta (\lambda) + n$ and
$x \in \Gamma(I,\timeobj{\FSB}{\delta'}, \lambda')$ considered up to the
equivalence relation $\sim$ generated by $\pair{\sigma\tau}{x} \sim \pair{\sigma}{\pair\id\tau \cdot x}$.
The projection $\psearlier : \searlier^{n+1} \to \searlier^n$ maps $[\pair{\sigma}{x}]$ to $[\pair{\sigma}{x}]$.
\end{lemma}

We omit the routine proof of this, which is by induction on $n$.

\begin{proofof}{Lemma~\ref{lem:interp:tirr}}
  Suppose $\alpha : \searlier^n \to \searlier$ is such that $\psearlier \circ \alpha = \psearlier^n$. We will show that
  $\alpha([\pair\sigma x]) = ([\pair\sigma x])$ for any
  $[\pair\sigma x] \in (\searlier^n\Gamma)(I,\timeobj{\FSA}{\delta}, \lambda)$.

  First consider the case of a representable object \linebreak
  $\Gamma = y(I,\timeobj{\FSB}{\delta'}, \lambda')$, and an element of the form
  \[([(\sigma, \id{})]) \in \searlier^n(y(I,\timeobj{\FSB}{\delta'}, \lambda'))(I,\timeobj{\FSA}{\delta}, \lambda).\]
  That is, $\pair{\id_I}\sigma : (I,\timeobj{\FSB}{\delta'}, \lambda') \to (I,\timeobj{\FSA}{\delta}, \lambda)$
  and $\delta'(\lambda') \geq \delta(\lambda) + n$ and $\id{}$ is the identity on $(I,\timeobj{\FSB}{\delta'}, \lambda')$
  which is an element of $y(I,\timeobj{\FSB}{\delta'}, \lambda')$. Then $\alpha([(\sigma, \id{})])$ must be on the
  form $([\pair\tau\rho])$, and since $\psearlier([\pair\tau\rho]) = \pair {\id_I}\tau\circ \rho$ and
  $\psearlier^n([(\sigma, \id{})]) = \pair{\id_I}\sigma$, the assumption implies
  $\pair {\id_I}\tau\circ\rho = \pair {\id_I}\sigma$, and in particular that $\rho$ is on the form
  $\pair{\id_I}{\rho'}$. Therefore
  \begin{align*}
    \alpha([(\sigma, \id{})]) &= [\pair\tau{\pair{\id_I}{\rho'}}] =
    [\pair\tau{\pair{\id_I}{\rho'}\cdot \id{}}] \\ &= [\pair{\tau\rho'}{\id{}}] = [(\sigma, \id{})]
  \end{align*}

  Now, consider a general $[\pair\sigma x] \in (\searlier^n\Gamma)(I,\timeobj{\FSA}{\delta}, \lambda)$. That is,
  $\pair\id\sigma : (I,\timeobj{\FSB}{\delta'}, \lambda') \to (I,\timeobj{\FSA}{\delta}, \lambda)$
  and $\delta'(\lambda') \geq \delta(\lambda) + n$ and $x \in \Gamma (I,\timeobj{\FSB}{\delta'}, \lambda')$.
  Let $\gamma : y(I,\timeobj{\FSB}{\delta'}, \lambda') \to \Gamma$ be the element corresponding to $x$ by
  the yoneda lemma, i.e., $\gamma\pair f\sigma = \pair f\sigma\cdot x$. Then
\begin{align*}
 \alpha([\pair\sigma x]) & = \alpha([\pair\sigma{\gamma (\id)}]) \\
  & = \alpha(\searlier^n(\gamma)[\pair\sigma\id]) \\
  & = \searlier(\gamma)(\alpha([\pair\sigma\id])) \\
  & = \searlier(\gamma)([\pair\sigma\id]) \\
  & = [\pair\sigma x]
\end{align*}

For the second statement of the theorem, recall that \citep[Proposition 7.1]{clottmodel} states that
$\psearlier : \searlier(\compr\Gamma\clk) \to \compr\Gamma\clk$ is an isomorphism. Therefore,
also the projection $\psearlier^n : \searlier^n(\compr\Gamma\clk) \to \compr\Gamma\clk$
is an isomorphism, from which the statement follows.
\end{proofof}

\subsubsection{Semantics for HITs}
\label{sec:hits:sem}

We provide semantics for any HIT definable in the schema with a direct generalization of the method employed in \cite{CTTHITS}.
To each HIT signature given we associate a notion of algebra.
An initial such algebra is then constructed and shown to support the necessary structure to model the HIT.
For this section we denote objects of the time category by $x, x',\dots$ to avoid notational clashes.

Reflecting the syntax of higher inductive types, we will work with homogeneous composition and transport separately.
The semantic counterpoint to $\hcomp$ is the notion of a fibrancy structure.
A fibrancy structure for a type global type $A$ is an operation taking a face $\phi:\F$, an element $u_0:A$ and a partial path $p:A^{\I}$ defined on the extent $\phi$ and equal to $u_0$ at the endpoint $0:\I$.
It then produces an element $h_A(\phi, u_0, u) : A$ which is equal to $u(1)$ on extent $\phi$.

Consider a list of constructors $\isconstrs \Delta {\mathcal{K}}$, consisting of constructors $(\ell_i,(\Gamma_i;\bar{\Xi_i};\Psi_i;\varphi_i;e_i)) \in \mathcal{K}$ for each $1\leq i\leq n$, and let $\Delta' \vdash \delta: \den \Delta$.
A $\mathcal{K}, \delta$-algebra is a presheaf $A\in \slice{\PSh{\CubeCat \times \TimeCat}}{\Delta'}$ with a fibrancy structure $h_A$, and for each $(\ell_i,(\Gamma_i;\bar{\Xi_i};\Psi_i;\varphi_i;e_i))$ semantic constructors $c_i^A$.
A semantic constructor is a map $c_i : \Pi(\gamma : \den {\Gamma_i}(\delta)) .\overline{(\den{\Xi_i}(\delta, \gamma)} \to A) \to \den {\Psi} \to A$.
Given such $c_i^A$'s, we can interpret the boundary conditions given by the constructors, and the semantic constructors are subject to the boundary coherence conditions generated by this interpretation.

As in the syntax, we define the semantic boundary interpretation $\evalbsem{-}$ by induction over the grammar for these terms.
We denote the list of constructors in $\mathcal{K}$ before the $i$'th constructor by $\mathcal{K}_{<i}$, and the list of semantic constructors for these by $C_{<i}$.
The boundary coherence condition states that \[c_i^A(\gamma, \theta, \bar{i}) = \evalbsem{M_{i,j}}_{C_{<i}}(\gamma, \theta, \bar{i})\] when $\phi_{i,j}(\bar{i})$, which is the direct interpretation of \[[\varphi_{i,0}~\evalbsem{M_{i,0}}, \ldots, \varphi_{i,n}~\evalbsem{M_{i,n}}].\]
We define the boundary interpretation for the remaining grammar as follows:
\begin{align*}
\evalbsem{x_i(\bar{u}(\gamma, \theta, i))}_{C_{<i}} &=
\den{\theta_i}_{\id, \den{\bar{u}}} \\
\evalbsem{\con_j(\bar{t}, \overline{\lambda \xi. N}, \bar{i})}_{C_{<i}} &=
c_j^A(\den{\bar{t}}, \overline{\lambda \xi .\evalbsem{N}_{C_{<i}}}, \den{\bar{i}}) \\
\evalbsem{\hcomp[\phi \mapsto u] \, u_0 }_{C_{<i}} &= h_A(\phi, \evalbsem{u_0}_{C_{<i}}, \evalbsem{u}_{C_{<i}})
\end{align*}
A morphism of $\mathcal{K}, \delta$-algebras $f: A\to B$ is a natural transformation preserving all $\mathcal{K},\delta$-algebra structure up to equality.
For constructors this means that $f(c_i^A(\gamma, \theta, r)) = c_i^B(\gamma, f \circ \theta, r)$.

We will now define a presheaf $H^{pre}$, and carve out the initial $\mathcal{K},\delta$-algebra as a subpresheaf $H\subset H^{pre}$ following the construction exemplified in \cite{CTTHITS}.
Let $\rho \in \den{\Delta}(I,x)$.
The presheaf $H^{pre}$ consists of formal constructor elements, and formal solutions to composition problems.
Concretely the presheaf $H^{pre}(I, x, \rho)$ will contain elements of two forms:
Firstly formal constructors $con_i(\gamma, \theta, r)$, where $\gamma \in \den{\Gamma_i}(I, x, \rho)$, $r\in \den{\Psi}(I, x, \rho)$, $\den{\phi_i}(I,r) \neq 1_\F$, and $\theta$ is a family indexed by $f \colon (I,x) \to (J,x')$ and an element $\xi \in \den{\Xi_{i,j}}(J, x', f\cdot \rho, f \cdot \gamma)$ of elements of $H^{pre}(J, x', f\cdot \rho, f \cdot \gamma)$.
Secondly $\hcomp$ elements of the form $\hcomp[\psi \mapsto u]\, u_0$ where $\psi \neq 1 \in \F(I), u_0 \in H^{pre}(I, x, \rho)$ and $u$ is a family indexed by pairs of a map $f : I\to J$ in $\CubeCat$, such that $\pshfunc \phi f = 1_\F$, and $r\in \I(J)$ of elements $u_{f,r} \in H^{pre}(J, x, (f,\id) \cdot \rho)$.
To define the action of $H^{pre}$ on morphisms, we first need to give interpretations of boundary terms.
While this is not formally the same case as for general algebras, the interpretation follows the same structure and so we leave out the definition here.
We denote this boundary interpretation in $H^{pre}$ again by $\evalbsem{-}$, owing to the fact that when the definition is complete, the two will coincide.
Let $g: (I,x) \to (J,x')$.
We define the action on elements on $H^{pre}$ as follows:

\begin{align*}
H^{pre}(g)&( con_i(\bar{t}, \bar{a}, \bar{r})) := \\
& \begin{cases}
	\evalbsem{M(g\cdot \bar{t}, g.\bar{a}, g\cdot \bar{r})}_{C} & \text{if } \phi_i(g\cdot \bar{r})\\
	con_i(g\cdot \bar{t}, g.\bar{a}, g \cdot \bar{r}) & \text{else}
\end{cases}\\
H^{pre}(g)&(\hcomp[\psi\mapsto u] u_0) :=\\
&\begin{cases}
	u_{g,1} & \text{if } \psi(g \cdot \bar{r})\\
	\hcomp [ g \cdot \psi \mapsto g.u] (g \cdot u_0) & \text{else}
\end{cases}
\end{align*}
Here $g.\bar{a}$ is given by $(g.\bar{a})_{i,f,\xi} = a_{i,fg,\xi}$, $g.u$ is $u$ similarly reindexed by the cubical component of $g$ and $C$ is list of constructors $con_i$ on which the boundary interpretation is based.

Now we carve out those elements where the families taken as input are natural.
This means that we include elements $con_i(\gamma, \theta, r)$ where each $\theta_{f,\xi}$ is in $H$ and $\theta_{gf,g \cdot \xi} = g \cdot \theta_{f,\xi}$.
Additionally, we include those $\hcomp[\psi \mapsto u] \, u_0$ elements where $u_0$ and each $u_{f,r}$ is in $H$, $u_{gf,(g, \id) \cdot r} = (g, \id) \cdot u_{f,r}$ and $(f,\id) \cdot u_0 = u_{f, 0}$ for suitable maps $f$ such that $f \cdot \psi = 1_\F$.

\begin{lemma} \label{lem:initial:alg}
The presheaf $H$ is the initial $\mathcal{K}, \delta$-algebra for each $\delta:\Delta$. As a consequence, it has a unique map to any cubical type $A$ with the structure of a $\mathcal{K}, \delta$-algebra.
\end{lemma}
\begin{proof}
Let $A$ be a $\mathcal{K}, \delta$-algebra with fibrancy structure $h_A$ and constructors semantic $c_i$.
We are tasked with producing a natural transformation $e:H\to A$, and we define $e(u)$ by induction on the structure of $u$:
In the case where $u = con_i(\bar{t}, \bar{a}, \bar{r})$, we define $e(u) := c_i(\bar{t}, \overline{e(a)}, \bar{r})$, and in the case $u= \hcomp[\phi \mapsto u] u_0$ we define $e(u):= h_A(\phi, e(u), e(u_0))$.
In both cases, $e(x)$ is $e$ applied levelwise to the family $x$.
While making this definition, we have to verify naturality at each stage, which follows in each case from the conditions on the $c_i$'s or the computational behaviour of the fibrancy structures.
\end{proof}

\begin{lemma}\label{lem:semHIT:hascomp}
The presheaf $H$ carries a composition structure.
\end{lemma}

The proof of this lemma uses the fact that $\Delta$, $\Gamma_i$ and $\Xi_{i,j}$ all carry transport structures.
Given those we define a transport structure on $H$ by initiality,
following the syntactic description of the action of $\trans$ on
constructors in Section~\ref{appendix:trans}. Combining this with the
homogeneous composition structure we conclude the proof by Lemma 5 of
\cite{CTTHITS}.

\begin{proposition}\label{lem:semHIT:depElim}
The presheaf $H$ supports the dependent elimination principle of $\textsf{H}$.
\end{proposition}

This follows from Lemma~\ref{lem:initial:alg} as in~\cite{CTTHITS}.

In addition to Lemma~\ref{lem:semHIT:hascomp} and Proposition~\ref{lem:semHIT:depElim},
we need to verify that $H$ supports the formation and introduction rules, as well as the judgemental equalities
These verifications are routine.

\subsubsection{Verifying the principle of induction under clocks}

The principle of induction under clocks is verified in the model by the following theorem:
\begin{theorem}\label{thm:clockielim}
Consider a HIT $\istype{\Delta} {\HIT H {}}$, a type family
\[\istype{\Gamma, \forall \bar\kappa . \HIT H {(\capp[\bar\kappa]\delta)}}{D}\] and a term
$\hastype{\Gamma} {\delta} {\forall \bar\kappa . \Delta}$. Suppose that we
semantically have an elimination list, i.e., for each $i$ we have a term
\[u_i : \den{D[\lambda \bar\kappa . \con_i (\capp[\bar\kappa] \gamma, \bar{\capp[\bar\kappa] x}, \psi)]}\]
over the context
\begin{align*}
\left\llbracket \arraycolsep=1.4pt
\begin{array}{l}
	\Gamma,
	\gamma  : \forall \bar\kappa . \Gamma_i[\capp[\bar\kappa] \delta],
    \bar{x} : \bar{\forall\bar\kappa.\Xi_i[\capp[\bar\kappa] \delta,\capp[\bar\kappa] \gamma] \to \HIT H {(\capp[\bar\kappa] \delta)}},\\
	\bar{y} : \bar{\left( \forall \bar\kappa. \Xi_i[\capp[\bar\kappa]\delta,\capp[\bar\kappa] \gamma] \right)
	\to D[\lambda \bar\kappa. \lambda \xi. \capp[\bar\kappa] x\,\capp[\bar\kappa] \xi]},
	\psi : \Psi_i \vdash
\end{array} \right\rrbracket.
\end{align*}
Then we can construct a term of $\den{D}$ over
$\den{\Gamma.\forall\bar\kappa. \HIT H {(\capp[\bar\kappa] \delta)}}$.
\end{theorem}

\begin{proof}
The semantics of clock quantification of $A$ is given by the semantic $\Pi$ over the clock object, $\Pi \clk . \den{A}$.
It is shown in \cite{GDTTmodel} this is equivalent at $I,x,\rho$ to the limit
$$\lim_n \den{A}(I, x+ \lambda_n, \iota \cdot \rho, \lambda)$$
with structure maps given by the map $\tick^\lambda : (I, x+ \lambda_{n+1}) \to (I, x+ \lambda_n)$. Here $(\mathcal{E},f) + \lambda_n = (\mathcal{E}\cup \{ \lambda\} ,f[\lambda \mapsto n])$.
The map $\tick^\lambda$ is given by the identity maps, so it can be viewed as letting a single time step pass on the clock $\lambda$.
Nominally such a family depends on the choice of fresh clock $\lambda$, but of course this dependence goes away since we have isomorphisms $x+\lambda_n \cong x+\lambda'_n$ in the time category.
Hence we will suppress the change of $\lambda$, as long as these are fresh meaning that the environment variable is of the form $\iota \cdot \rho$.

The action of morphisms on HIT's in the model can only change the constructor type by acting on the cube
component. As a consequence of this, the structure of elements in $\den{\forall \kappa . \HIT H {(\capp
\delta)}}$ are determined at each level by their structure at 0.

It is sufficient to work with $d:\den{\forall \kappa . \Delta}(I,x)$, with the proof for
general $\Gamma$ following by pulling back along the term $\hastype{\Gamma} {d} {\forall \kappa.
\Delta}$.
We start by making explicit the data of the assumed terms $u_i$. The input to a $u_i$ is as
follows:
\begin{itemize}
\item[(a)]
We have a family $t = (t_n) \in \den{\forall \kappa. \capp{\Gamma_i}}(I,x,d)$.
\item[(b)]
For each $j$ we have a family
\[a=(a^j_n) \in \den{\forall\kappa .\Xi_{i,j}
[\capp \delta, \capp \gamma] \to \HIT H {(\capp \delta)}}(I, x, d, t).\]
This can be further unfolded as follows: We have for each $n$,
$f: (I, x+\lambda_n) \to (J,x')$ and
$$z \in \den{\Xi_{i,j}[\capp \delta, \capp \gamma]}(J,x',f\cdot d_n,f\cdot t_n, f(\lambda))$$
an element $a^j_{n,f,z} \in \den{\HIT H {(\capp\delta)}}(J,x',f\cdot d_n, f(\lambda))$.
\item[(c)]
The next family is indexed by $j$ as above, a morphism $g:(I,x) \to (J,x')$ and a family
\[z:=(z_n) \in \den{\forall \kappa . \Xi_{i,j}[\capp\delta,\capp \gamma]}(J,x',g\cdot d, g\cdot t).\]
Note here that this means we have a family $h:=(h_n)\in \den{\forall \kappa. \HIT H {(\capp \delta)}}
(J,x',g\cdot d)$ given by $h_n := a^j_{n, g + \lambda_n, z_n}$. We then have an element
$b^j_{n,z} \in \den{D}(J,x',g\cdot d,h)$, and we denote the family by $b$.
\item[(d)]
An element $r\in \Psi_i(I)$.
\end{itemize}
When given this, we have an element
$$(\con_i(t_n, a_n, r))_n \in \den{\forall \kappa. \HIT H {(\capp \delta)}}(I,x,d)$$
and we assume that we can inhabit $D$ over this element, i.e., that we have
$$u_i(t,a,b,r) \in \den{D}(I,x,d,(\con_i(t_n, a_n, r))_n).$$

We need to inhabit
$\den{D}$ at an arbitrary element
\[h \in \den{\forall \kappa. \HIT H {(\capp \delta)}}(I,x,d).\]
Observe that the structure of the family $h$
is completely determined by the structure of $h_0$. This holds because the family
is compatible with the tick maps, and these are trivial in the cube component which
means in particular that they cannot change the outer constructor of an element of
the HIT. This means that we only have to inhabit it at families either of the form
$(\con_i(t_n, a_n, r))$ or of the form $(\hcomp[\phi \mapsto u^n] \, u^n_0)$. We can
do this by induction on the structure, since it is the same for each component of the
family, and can therefore assume that we can inhabit $\den{D}$ at structurally simpler
families. Explicitly, the induction is done on the structure of the 0 component in the
family, so that the inductive hypothesis says that whenever we have a family $(h'_n)$
such that $h'_0$ is structurally simpler than $h_0$, we can inhabit $\den{D}$ over this
family. We denote the element over a family $h$ by $\alpha(h)$.

We consider first the case where we are given a family
$$(\con_i(t_n, a_n, r))_n \in \den{\forall \kappa. \HIT H {(\capp \delta)}}(I,x,d).$$
Note that $(t_n)$ is forced to be compatible with tick maps and that $r$ is constant.
This means that we have $(t_n)_n \in \den{\forall\kappa. \Gamma_i[\capp \delta]}(I,x,d)$
and $r\in \Psi(I)$, and we have the data of (a) and (d).
The implied typing of $a^j_n$ is a family over $f:(I,x+\lambda_n) \to
(J,x')$ and $z\in \den{\Xi_{i,j}[\capp \delta, \capp \gamma]}(J,x',f\cdot d_n, f\cdot t_n, f(\lambda))$
of elements $a^j_{n,f,z}\in \den{\HIT H {(\capp\delta)}}(J,x',f\cdot d_n, f(\lambda))$. This is exactly
the typing of (b), so we need only to construct the family in (c). We need this to be a family over
$g:(I,x) \to (J,x')$ and $z:=(z_n) \in \den{\forall \kappa .
\Xi_{i,j}[\capp\delta,\capp \gamma]}(J,x',g\cdot d, g\cdot t)$. Given this, we again define
the family
\[h:=(h_n)\in \den{\forall \kappa. \HIT H {(\capp \delta)}} (J,x',g\cdot d)\]
by $h_n := a^j_{n, g + \lambda_n, z_n}$. Note that this family is compatible with ticks, and
hence we obtain $\alpha((h_n)) \ in \den{D}(J,x',g\cdot d,h)$ by inductive hypothesis,
since $h_0$ is given by a component of the structurally simpler $a$. We define
$$\alpha((\con_i(t_n, a_n, r))_n) = u_i(t, a, \alpha((h_n)), r).$$

Consider now the $\hcomp$ case. Since $\den{D}$ is a type, it carries a composition structure.
The strategy for inhabiting $\den{D}$ will be to construct an appropriate composition problem
and apply the inductive hypothesis. Consider a family $h=(\hcomp [\phi \mapsto u^n]\, u^n_0) \in
\den{\forall \kappa. \HIT H {(\capp \delta)}}(I,x,d)$. Here the implicit typing is as follows:
For each $n$, we have for each $f: I \to J$ such that $f \cdot \phi = \top$ and $s \in \Psi(J)$
an element $u^n_{f,s} \in \den{\HIT H {(\capp \delta)}}(J,x+\lambda_n, f\cdot d_n, \lambda)$ and
an element $u^n_0 \in \den{\HIT H {(\capp \delta)}}(J,x+\lambda_n, d_n, \lambda)$. The $u^n_0$
family has the right shape for us to employ our induction hypothesis to inhabit $\den{D}$ over
it. Because the shape of $f$ and $s$ does not depend on the time component, we also have for each
$f$ a family in $n$ given by $(u^n_{f,s})$. All these families have 0 component structurally simpler,
hence we can inhabit $\den{D}$ over each of them by $\alpha((u_0^n))$ and $\alpha((u_{f,i}^n))$. We will
write $\alpha((u^n))$ for the family given at $f,i$ by $\alpha((u_{f,i}^n))$.
We can then define $\alpha(h) = \comp_{D[v] }[\phi \mapsto \alpha((u^n))]\, \alpha((u^n_0))$ where $v  =
(\mathsf{hfill}_{\den{\HIT H {\capp \delta}}}[\phi \mapsto u^n]\, u^n_0)$.

It must be verified that $\alpha$ is natural in morphisms of the category while
defining it. Both because we need $\alpha$ to be natural for the desired conclusion
and because the recursive calls used to define $\alpha$ require it. Let $f:(I,x)\to
(J,x')$. Naturality means that $f\cdot \alpha(h) = \alpha(f\cdot h)$, which unfolds as follows:
\begin{itemize}
\item
Consider a family $h=(\con_i(t_n,a_n,r))$. By definition $\alpha(h) = u_i(t, a, b, r)$ where
$b_{g,z} = \alpha((a^j_{n, g + \lambda_n, z_n}))$.
If $f$ does not trigger a boundary condition for $\con_i$ we have
$$f\cdot \alpha(h) = u_i (f\cdot t, f . a, f . b, f\cdot r) \text{ and }$$
$$(f\cdot h)_n = \con_i((f+\lambda_n)\cdot t_n, (f+\lambda_n). a_n, (f+\lambda_n)\cdot r).$$
Recall here that $g.a_n$ is defined to be the family $a_n$ is given by composing the map index
with $g$ and acting on the input index.
Applying $\alpha$ to this yields $u_i( ((f+\lambda_n)\cdot t_n), ((f+\lambda_n) . a_n),
b', f \cdot r)$ where $b'$ is defined from $((f+\lambda_n) . a_n)$ as above, i.e.,
$$b'_{g,z} = \alpha((  ((f+\lambda_n) .a)_{n, g + \lambda_n, z_n})).$$
By definition of
the morphism action on a limit we have $((f+\lambda_n)\cdot t_n) = f\cdot (t_n)$ and $((f+\lambda_n).
a_n) = f. (a_n)$. The action of $f$ and $f+\lambda_n$ are the same on $r$, since they have the same cubical
component. It remains to argue that $f. b$ coincides with the $b'$ defined from the family $f. a$, which
already a part of the naturality of the $u_i$'s. We check this componentwise for the family, so we have to
show that $(f.b)_{g,z}  = b'_{g,z}$. The left hand side is by definition equal to $b_{gf, f\cdot z} =
\alpha((a^j_{n, gf + \lambda_n,(f+\lambda_n)\cdot z_n}))$, while the right hand side is equal to
$\alpha((((f+\lambda_n).a)^j_{g+\lambda_n, z_n}))$. Unfolding the right hand side we get
$\alpha(((a)^j_{(g+\lambda_n)(f+\lambda_n),(f+\lambda_n)\cdot z_n}))$, which is the same as the left hand
side since $(g+\lambda_n)(f+\lambda_n) = gf + \lambda_n$.

If $f$ triggers a boundary condition of $\con_i$ we have by assumption that
$$f\cdot \alpha(h) = f\cdot u_i(t,a,b,r) = \den{\evalbsynclock{\delta}{\mathcal{E}}
{\bar{x} \mapsto \bar{y}}{e}}(f\cdot t,f\cdot a,f\cdot  b,f\cdot r)$$
where $e$ is the boundary term of $\con_i$.
On the other hand, the structure of $\HIT H {}$ means that $\alpha(f\cdot h) =
\alpha((\evalbsem{e}((f+\lambda_n) \cdot t_n,(f+\lambda_n) . a_n,(f+\lambda_n) \cdot r))$ where
$\evalbsem{e}$ is the boundary evaluation for $\HIT H {(\capp\delta)}$. Equality of these two can be shown
by induction on the structure of $e$.

\item
Consider a family $h=(\hcomp[\phi \mapsto u^n]\, u_0^n)$. By definition $\alpha(h) =
\comp_{D[v]}[\phi \mapsto \alpha((u^n))]\, \alpha((u_0^n))$. If $f$ does not make $\phi$ true
we have to verify that
$$f\cdot \comp_{D[v]}[\phi \mapsto \alpha((u^n))]\, \alpha((u_0^n)) = \alpha(f\cdot h).$$
We unfold the right hand side:
\begin{align*}
& \alpha(f\cdot h) = \\
& \textcolor{white}{==} \alpha((\hcomp[f\cdot \phi \mapsto (f+\lambda_n) . u^n]\, (f+\lambda_n) \cdot u_0^n)) \\ &=
\comp_{D[f\cdot v]}[f \cdot\phi \mapsto
\alpha(((f+\lambda_n) . u^n))] \\
& \textcolor{white}{====================}\alpha(((f+\lambda_n) \cdot u_0^n))\\
&=\comp_{D[f\cdot v]}[f\cdot \phi \mapsto f . \alpha((u^n))]\, f\cdot\alpha((u_0^n))\\ &=
f\cdot \comp_{D[v]}[\phi \mapsto \alpha((u^n))]\, \alpha((u_0^n))
\end{align*}
The second to last step above follows by the inductive hypothesis that $\alpha$ is natural on
simpler families, and the last line is exactly the left hand side of the original equation.

If $f$ makes $\phi$ true, both composition in $D$ and $\hcomp$ reduce to appropriate
components of the input, i.e., we have $\alpha(f\cdot h) = \alpha((u_{(f+\lambda),1}^n))$
and $f\cdot \alpha(h) = \alpha((u^n))_{f,1}$ for the right and left hand sides of the
equality we are trying to show. These agree by definition of $\alpha((u^n))_{f,1}$.

\end{itemize}

This means that $\alpha$ assembles into a term of $\den{D}$ over $\den{\Gamma.\forall\kappa.
\HIT H {(\capp \delta)}}$ as desired. It is moreover clear from the definition of $\alpha$
that it will validate the computation rules of the induction principle, since the value on
constructor families was given directly by the appropriate $u_i$ and the value on $\hcomp$
families was given by the appropriate composition in $D$.
\end{proof}

\end{document}